\documentclass[11pt]{amsart}
\usepackage{amsmath,amssymb,amsthm}
\usepackage{tikz}
\usetikzlibrary{cd}
\usetikzlibrary{lindenmayersystems}
\usepackage{graphicx}
\usepackage[most]{tcolorbox}
\usepackage{tkz-euclide}
\usepackage{fnpct}
\usepackage[margin=0.8in,centering]{geometry}
\usepackage[short-journals,short-months,initials,backrefs]{amsrefs}
\usepackage[colorlinks=true,pagebackref]{hyperref}
\usepackage{enumitem}
\usepackage{caption, subcaption}
\usepackage{listings}

\theoremstyle{plain}
\newtheorem{theorem}{Theorem}
\newtheorem{lemma}{Lemma}[section]
\newtheorem{proposition}[lemma]{Proposition}
\newtheorem{corollary}[lemma]{Corollary}
\newtheorem{conjecture}[lemma]{Conjecture}

\theoremstyle{definition}
\newtheorem{definition}[lemma]{Definition}

\newtheorem{openquestion}{Open Question}

\theoremstyle{remark}
\newtheorem{remark}[lemma]{Remark}
\newtheorem*{remark*}{Remark}

\usetikzlibrary{arrows}
\newcommand{\midarrow}{\tikz \draw[-triangle 90] (0,0) -- (.3,0);}
\usetikzlibrary{arrows}
\tikzset{
  arrow/.pic={\path[tips,every arrow/.try,->,>=#1] (0,0) -- +(.1pt,0);},
  pics/arrow/.default={triangle 90}
}

\def\VR{\kern-\arraycolsep\strut\vrule &\kern-\arraycolsep}
\def\vr{\kern-\arraycolsep & \kern-\arraycolsep}
\usepackage{setspace}
\newcommand{\mult}{\textup{mult}}
\usepackage{anyfontsize}
\usepackage{setspace}

\usepackage{listings}
\usepackage{color} 
\definecolor{mygreen}{RGB}{28,172,0} 
\definecolor{mylilas}{RGB}{170,55,241}

\newtcolorbox{Box2}[2][]{
                lower separated=false,
                colback=white,
colframe=black,fonttitle=\bfseries,
colbacktitle=black,
coltitle=white,
enhanced,
attach boxed title to top left={yshift=-0.075in,xshift=0.15in},
                 boxed title style={boxrule=0pt,colframe=white,},
title=#2,#1}

\title[Magnetic spectral decimation on the Sierpinski gasket]{Spectral decimation of the magnetic Laplacian on the Sierpinski gasket: Solving the Hofstadter-Sierpinski butterfly}

\author{Joe P.\@ Chen}
\address[Joe P.\@ Chen]{Department of Mathematics, Colgate University, Hamilton, NY 13346, USA.}
\email{jpchen@colgate.edu}
\urladdr{\url{http://math.colgate.edu/~jpchen}}

\author{Ruoyu Guo}
\address[Ruoyu Guo]{Department of Mathematics, Colgate University, Hamilton, NY 13346, USA.}
\curraddr{\textsc{Department of Mathematics, Brandeis University, Waltham, MA 02453, USA.}}
\email{rguo@brandeis.edu}

\thanks{This work is an outgrowth of the High Honors Bachelor's thesis by the second-named author at Colgate University, advised by the first-named author. We acknowledge partial financial support from the Research Council of Colgate University, the Simons Foundation (Collaboration Grant for Mathematicians \#523544), and the National Science Foundation (DMS-1855604).
}

\begin{document}
\renewcommand{\theequation}{\thesection.\arabic{equation}}
\numberwithin{equation}{section}

\begin{abstract}
The magnetic Laplacian (also called the line bundle Laplacian) on a connected weighted graph is a self-adjoint operator wherein the real-valued adjacency weights are replaced by unit complex-valued weights $\{\omega_{xy}\}_{xy\in E}$,  satisfying the condition that $\omega_{xy}=\overline{\omega_{yx}}$ for every directed edge $xy$.
When properly interpreted, these complex weights give rise to magnetic fluxes through cycles in the graph.

In this paper we establish the spectrum of the magnetic Laplacian, as a set of real numbers with multiplicities, on the Sierpinski gasket graph ($SG$) where the magnetic fluxes equal $\alpha$ through the upright triangles, and $\beta$ through the downright triangles. 
This is achieved upon showing the spectral self-similarity of the magnetic Laplacian via a 3-parameter map $\mathcal{U}$ involving non-rational functions, which takes into account $\alpha$, $\beta$, and the spectral parameter $\lambda$.
In doing so we provide a quantitative answer to a question of Bellissard [\emph{Renormalization Group Analysis and Quasicrystals} (1992)] on the relationship between the dynamical spectrum and the actual magnetic spectrum.

Our main theorems lead to two applications. In the case $\alpha=\beta$, we demonstrate the approximation of the magnetic spectrum by the filled Julia set of $\mathcal{U}$, the Sierpinski gasket counterpart to Hofstadter's butterfly. Meanwhile, in the case $\alpha,\beta\in \{0,\frac{1}{2}\}$, we can compute the determinant of the magnetic Laplacian and the corresponding asymptotic complexity.
\end{abstract}

\date{\today}



\maketitle
\tableofcontents


\section{Introduction and main results}
\label{sec:intro}

Spectral analysis of the magnetic Laplacian on a planar lattice has both theoretical and practical implications.
The famous ``Hofstadter's butterfly'' \cite{hofstadter} describes the energy spectrum of a noninteracting electron gas moving on the planar integer lattice under a uniform magnetic field, \emph{i.e.,} the magnetic flux through every square cell is constant.
Understanding the fractal nature of this spectrum involves the interplay of analysis, geometry, topology, and number theory.

We can also study the magnetic spectrum on other periodic or quasi-periodic planar graph under a uniform magnetic field.
For instance, we can replace the square lattice by the triangular lattice. From the triangular lattice, we can remove vertices and attached edges in such a way that the remainder is an infinite blow-up of the Sierpinski gasket graph ($SG$); see Figure \ref{fig:SGgraph1} for an informal construction, and see Figure \ref{fig:SGgraph2} for the nesting property of $SG$ graphs. Formal definition of $SG$ is given in \S\ref{sec:DefSG}.

\def\trianglewidth{2cm}
\pgfdeclarelindenmayersystem{Sierpinski triangle}{
    \symbol{X}{\pgflsystemdrawforward}
    \symbol{Y}{\pgflsystemdrawforward}
    \rule{X -> X-Y+X+Y-X}
    \rule{Y -> YY}
}
\begin{figure}
\begin{center}
\foreach \level in {0,...,4}{
\tikzset{
    l-system={step=\trianglewidth/(2^\level), order=\level, angle=-120}
}
\begin{tikzpicture}[scale=1.2]
    \draw ++(0:\trianglewidth) -- ++(120:\trianglewidth) -- (0,0);
    \draw (0,0) l-system
    [l-system={Sierpinski triangle, axiom=X},fill=white];
    \draw (1,-0.5) node {$\mathfrak{G}_{\level}$};
\end{tikzpicture}
}
\end{center}
\caption{The (unscaled) Sierpinski gasket graphs from $\mathfrak{G}_0$ to $\mathfrak{G}_4$.}
\label{fig:SGgraph1}
\end{figure}
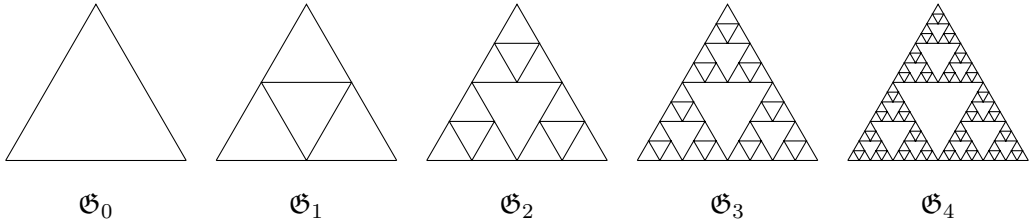

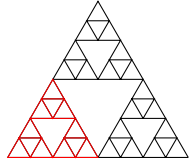
\begin{figure}
\begin{center}
\tikzset{
    l-system={step=\trianglewidth/8, order=3, angle=-120}
}
\begin{tikzpicture}[scale=1.2]
    \draw ++(0:\trianglewidth) -- ++(120:\trianglewidth) -- (0,0);
    \draw (0,0) l-system
    [l-system={Sierpinski triangle, axiom=X},fill=white];
    \draw [red] (1,0) -- (60:1) -- (0,0) -- cycle;
    \draw [red] (0.5,0) -- ++(60:0.5) -- +(-0.5,0) -- cycle;
    \draw [red] (0.25,0) -- ++(60:0.25) -- +(-0.25,0) -- cycle;
    \draw [red] (0.75,0) -- ++(60:0.25) -- +(-0.25,0) -- cycle;
    \draw [red] (60:0.75) -- ++(-60:0.25) -- +(60:0.25) -- cycle;
\end{tikzpicture}
\end{center}
\caption{The graph {\color{red}$G_{N-1}$} is nested in $G_N$}
\label{fig:SGgraph2}
\end{figure}

By trading a periodic crystal (the square lattice) for another graph without translational invariance ($SG$), it is natural to ask the following question: how does the magnetic spectrum change?
In fact, the problem of computing the magnetic spectrum on $SG$ started in the early 1980s \cites{DABK83,Rammal, A84, Ghez}.
Already then the authors have identified the ``nesting mechanism'' for generating the spectrum recursively, and pointed out the existence of localized eigenfunctions associated with certain ``exceptional'' eigenvalues.
Probably the most important claim made was that the magnetic spectrum is given by an analog of Hofstadter's butterfly shown in \cite{Ghez}*{Figure 2},\footnote{See also \cite{quasicrystals}*{Figure 2} for what appears to be a higher-resolution picture of \cite{Ghez}*{Figure 2}, though it is the authors' opinion that the two pictures have major differences.} which already differs qualitatively from the original butterfly on the square lattice.
We refer the reader to Bellissard's survey \cite{quasicrystals} for an overview of spectral problems on quasi-periodic lattices and renormalization group methods, which includes a discussion of the magnetic spectral problem on $SG$.

From the mathematics perspective, the aforementioned nesting mechanism can be formalized into an abstract framework known as \textbf{spectral decimation}.
Details are provided in \S\ref{sec:specdec} below.
Regarding its applicability, we note the well-known results of Fukushima-Shima \cite{FukushimaShima} and Shima \cite{Shima}, which establish the Laplacian spectrum (under zero magnetic field) on the scaling limit of certain self-similar fractals such as $SG$.
As for the infinite $SG$ lattice, the first complete characterization of the Laplacian spectrum was attained by Teplyaev \cite{Teplyaev}, based on an abstract formulation of spectral decimation by him and Malozemov \cite{Malozemov}.
The said techniques have since been applied to obtaining Laplacian spectra on a variety of fractals.
There are too many subsequent works to list here in this introduction, but we single out the pedagogically influential paper \cite{3n-gasket}.

Analysis of magnetic Laplacians on fractals has seen renewed interest in 2010s.
Probably closest to our present work is that of Hyde, Kelleher, Moeller, Rogers, and Seda \cite{Hyde17}, where they obtained the spectrum on $SG$ in which the magnetic $1$-form is locally exact and there is nonzero flux through only a finite number of triangles.
Another fractal graph whose magnetic spectrum can be solved exactly is the ``diamond fractal'' \cite{diamond}.
On the functional analytic side, we would like to mention recent results on the closability and self-adjointness of, and a Feynman-Kac formula corresponding to, magnetic Laplacians on compact fractal spaces \cites{HT13, H15, HR16} (or more generally, resistance spaces in the sense of Kigami \cite{Kigami}).

Despite the aforementioned progress, the original problem of identifying the spectrum of the magnetic Laplacian on $SG$ under a uniform magnetic field, posed more than 30 years ago \cites{Rammal, DABK83, A84}, does not have a complete mathematical solution. 
The outstanding issue reads, according to Bellissard \cite{quasicrystals}*{p.\@128}: \emph{``Is the \textbf{dynamical spectrum} equal to the \textbf{actual spectrum} of the original operator? This is a question with no answer yet.''} 
Here the \textit{dynamical spectrum} refers to the \textit{(filled) Julia set} of a certain dynamical system, while the \textit{actual spectrum} refers to the spectrum of the magnetic Laplacian operator.


\textbf{The main purpose of this work} is to provide a full solution to this long-standing problem.
Via the aforementioned \textbf{spectral decimation} (see \S\ref{sec:specdec} for details), we establish the magnetic spectrum on $SG$ as a set of real numbers with multiplicities, when the flux through each upright triangle (resp.\@ downright triangle) equals $\alpha$ (resp.\@ $\beta$), for any $\alpha, \beta \in [0,1)$.
We not only identify the portion of the spectrum which is generated recursively via a 3-parameter map, but also resolve the other portion of the spectrum which does not arise from the recursive mechanism (\emph{i.e.,} the values which lie in the \textbf{exceptional set} for spectral decimation).

In order to describe our results in more detail, we provide some definitions.

\subsection{The Sierpinski gasket}
\label{sec:DefSG}

Let $x_0=(0,0)$, $x_1=(\frac{1}{2}, \frac{\sqrt{3}}{2})$, and $x_2=(1,0)$ be the vertices of a unit equilateral triangle in $\mathbb{R}^2$, and $\mathfrak{G}_0$ be the complete graph on the vertex set $V_0 = \{x_0,x_1,x_2\}$. 
We introduce three contracting similitudes $\Phi_i : \mathbb{R}^2 \to\mathbb{R}^2$, $\Phi_i(x) = \frac{1}{2}(x-x_i) + x_i $ for each $i\in \{0,1,2\}$. 
The Sierpinski gasket fractal $K$ is the unique nonempty compact set $K$ such that $K= \bigcup_{i=0}^2 \Phi_i(K)$.
To obtain the associated level-$N$ pre-fractal graph $\mathfrak{G}_N$, $N\geq 1$, we define by induction $\mathfrak{G}_N = \bigcup_{i=0}^2 \Phi_i(\mathfrak{G}_{N-1})$.
To make all edges of the graph have unit length, we set $G_N:= 2^N \mathfrak{G}_N$, where for $\alpha >0$ and $\Omega \subset \mathbb{R}^2$ we denote $\alpha \Omega:= \{ \alpha x: x\in \Omega\}$ (see Figure \ref{fig:SGgraph2}).
The (one-sided) Sierpinski gasket graph $SG$ is then defined to be the union of a sequence of monotone increasing graphs $\bigcup_{N=0}^\infty G_N$. 
The number of vertices $|V_N|$ in $G_N=(V_N, E_N)$ is easily shown to be $\frac{3^{N+1}+3}{2}$, which we will denote by ${\rm dim}_N$.

\subsection{Magnetic Laplacian}
\label{sec:MagLap}
Let $G=(V,E)$ be a simple, locally finite, connected graph.
The (combinatorial) graph Laplacian on $G$ is $\Delta_G= D_G-A_G$, where $D_G$ and $A_G$ are the degree operator and the adjacency operator, respectively.
Equivalently,
\[
\left(\Delta_G u\right)(x) = \sum_{y\sim x} (u(x)-u(y)), \quad u\in \ell^2(V),
\]
where the sum is over vertices $y$ connected to $x$ by an edge.
Clearly $\Delta_G$ is self-adjoint on $\ell^2(V)$.
Sometimes it is more convenient to normalize the Laplacian by the degree, \emph{i.e.,} to define $\mathcal{L}_G=D_G^{-1} \Delta_G$, or equivalently,
\[
(\mathcal{L}_G u)(x) = \frac{1}{\deg_G(x)}\sum_{y\sim x} (u(x)-u(y)), \quad u \in \ell^2(V).
\]
This is called the probabilistic graph Laplacian, and it is self-adjoint on $L^2(V, \deg)$.
More generally, we introduce a \textbf{conductance} function ${\bf c}: \{\pm E\} \to \mathbb{R}_+$ on the set of oriented edges of $G$, and define the weighted graph Laplacian as
\[
(\mathcal{L}_{(G,{\bf c})} u)(x) = \sum_{y\sim x} {\bf c}_{xy} (u(x)-u(y)), \quad u\in\ell^2(V).
\]
We allow ${\bf c}_{xy} \neq {\bf c}_{yx}$: a natural example is to let ${\bf c}_{xy}$ be the transition probability $p(x,y)$ of an irreducible Markov chain on $G$.

Whereas the adjacency operator contains entries with values $0$ or $1$, we now replace the $1$'s by unit complex numbers to form the magnetic Laplacian.
The motivation behind the definition is from differential geometry.
Place a copy $\mathcal{W}_v$ of $\mathcal{W}=\mathbb{C}$ at each $v\in V$.
We call $\mathcal{W}= \bigoplus_{v\in V} \mathcal{W}_v$ a \emph{complex line bundle} on $G$.
On $\mathcal{W}$ we endow a \emph{unitary connection} $\Phi$ which satisfies the property that for every oriented edge $e=vv'$, the \emph{parallel transport} from $v$ to $v'$, $\phi_{vv'}: \mathcal{W}_v\to\mathcal{W}_{v'}$, is a unitary complex linear map such that $\phi_{v'v} = \phi_{vv'}^{-1}$.
By our choice that $\mathcal{W}_v=\mathbb{C}$, the action of $\phi_{vv'}$ is multiplication by a unit complex number $\omega_{vv'}$, satisfying $\omega_{v'v}=\overline{\omega_{vv'}}$.
From now on we will use the notation $\omega$ to denote a unitary ($U(1)$) connection on the complex line bundle.
We say that two $U(1)$ connections $\omega$ and $\omega'$ are \emph{gauge equivalent} if there exists a unitary map $\psi_v: \mathcal{W}_v \to\mathcal{W}_v$ such that $\psi_{v'} \omega_{vv'} = \omega'_{vv'} \psi_v$, that is, $\psi_v$ induces a change of angle in the unit complex exponentials $\omega_{vv'}$.

We can also extend the definition of the line bundle to $E$.
Place a copy $\mathcal{W}_e$ of $\mathcal{W}=\mathbb{C}$ at each $e\in E$.
Then define a connection isomorphism $\omega_{ve} = \omega_{ev}^{-1}$ for a vertex $v\in V$ and an edge $e$ containing $v$, satisfying the condition that if $e=vv'$, then $\omega_{vv'}= \omega_{ev'} \omega_{ve}$, where $\omega_{vv'}$ is the parallel transport from $v$ to $v'$.

Given a unitary connection $\omega$, we define the corresponding magnetic Laplacian as
\begin{align}
\label{eq:LBLap}
(\mathcal{L}^\omega_{(G,{\bf c})} u)(x) =  \sum_{y\sim x} {\bf c}_{xy} (u(x) - \omega_{xy} u(y)), \quad u\in \ell^2(V).
\end{align}
Another way to express \eqref{eq:LBLap} is through the identity ``Laplacian $=$ div grad.'' 
Let $\Lambda^0(G,\omega)= \ell^2(V)$ and $\Lambda^1(G,\omega) = \{u\in \ell^2(E): u(-e)=-u(e) \text{ for all oriented edges } e\in E\}$ denote, respectively, the space of square-summable $0$-forms and $1$-forms on $(G,\omega)$.
Then
$\mathcal{L}^\omega_{(G,{\bf c})} = d^*d$, where
\begin{align}
d: \Lambda^0(G,\omega) \to \Lambda^1(G,\omega)&, \quad (df)(e) = \omega_{ye} f(y)- \omega_{xe} f(x), \quad e=xy,\\
d^*: \Lambda^1(G,\omega) \to \Lambda^0(G,\omega)&, \quad (d^* \chi)(v) = \sum_{e=v'v} {\bf c}_{vv'} \omega_{ev} \chi(e),
\end{align}
are, respectively, the gradient and the divergence operators.

A sequence $P$ of vertices $\{x_0, x_1, x_2, \dotsc, x_{m-1}, x_m\}$ is a \emph{path} if $x_i \sim x_{i+1}$ for all $i\in \{0,1,\dotsc, m-1\}$.
The product of the parallel transports along $P$, $\omega(P):=\omega_{x_0 x_1} \omega_{x_1 x_2}\dotsc\omega_{x_{m-1} x_m}$, is called the \emph{holonomy} of $\omega$ along $P$.
We are particularly interested in the holonomy of $\omega$ along $P$ when $P$ is a \emph{simple cycle}: by \emph{simple} we mean that $x_i\neq x_j$ for any pair $i,j \in \{0,1,\cdots, m-1\}$, and by \emph{cycle} we mean that $x_m=x_0$.

\begin{definition}
\label{def:flux}
Given an oriented simple cycle $\gamma$ in $(G,\omega)$, the \textbf{magnetic flux} through $\gamma$ is defined as the number $\theta \in [0,1)$ such that the holonomy is $\omega(\gamma) = e^{2\pi i \theta}$.
\end{definition}

\subsection{Spectrum of the magnetic Laplacian on the Sierpinski gasket}

Denote by $\mathcal{L}^\omega_N$ the magnetic Laplacian on the level-$N$ gasket graph $G_N$ endowed with the unitary connection $\omega$, that is,
\begin{align}
\label{eq:SGML}
(\mathcal{L}^\omega_N u)(x) = \sum_{y\sim x}\frac{1}{\deg_{G_N}(x)} (u(x) -\omega_{xy} u(y)), \quad u\in \ell^2(V).
\end{align}
By embedding $SG$ into the plane, we can unambiguously assign the counterclockwise orientation to each simple cycle, and apply Definition \ref{def:flux}.

\begin{definition}
\label{def:flux2}
The magnetic Laplacian $\mathcal{L}^{(\alpha,\beta)}_N$, $\alpha,\beta \in [0,1)$, is defined by \eqref{eq:SGML} assuming that the magnetic flux through each upright (resp.\@ downright) triangle of side length $1$ in the graph distance on $G_N$ equals $\alpha$ (resp.\@ $\beta$), \emph{cf.\@} Figure \ref{fig:fluxevolve}.
\end{definition}

\begin{figure}
\begin{center}
\begin{tikzpicture}[scale=1.2]
\draw (0,0) -- (3,0) -- (60:3) -- cycle;
\draw (1.5,0) -- ++(60:1.5) -- +(-1.5,0) -- cycle;
\draw (0.75,0) -- ++(60:0.75) -- +(-0.75,0) -- cycle;
\draw (2.25,0) -- ++(60:0.75) -- +(-0.75,0) -- cycle;
\draw (60:2.25) -- ++(0.75,0) -- +(-120:0.75) -- cycle;
\draw [dashed] (60:0.75) ++(0.75,0) -- ++(0.75,0) -- ++(120:0.75) -- cycle;
\draw (60:3) -- (60:4.5);
\draw (60:3) -- +(1.5,0);
\draw [loosely dotted, very thick] (60:3) ++(30:0.4) -- +(30:1);
\draw (3,0) -- +(1.5,0);
\draw (3,0) -- +(60:1.5);
\draw [loosely dotted, very thick] (3,0) ++(30:0.4) -- +(30:1);
\draw (1.5,0.4) node {\small $\beta$};
\draw (1.5,0.85) node {\small $\alpha$};
\draw (1.125,1.05) node {\small $\beta$};
\draw (1.875,1.05) node {\small $\beta$};
\draw (1.125,1.5) node {\small $\alpha$};
\draw (1.875,1.5) node {\small $\alpha$};
\draw (1.5,1.7) node {\small $\beta$};
\draw (1.5,2.15) node {\small $\alpha$};
\draw (5,2) node {$\longrightarrow$};
\begin{scope}[xshift=6cm]
\draw (0,0) -- (3,0) -- (60:3) -- cycle;
\draw (1.5,0) -- ++(60:1.5) -- +(-1.5,0) -- cycle;
\draw (60:3) -- (60:4.5);
\draw (60:3) -- +(1.5,0);
\draw [loosely dotted, very thick] (60:3) ++(30:0.4) -- +(30:1);
\draw (3,0) -- +(1.5,0);
\draw (3,0) -- +(60:1.5);
\draw [loosely dotted, very thick] (3,0) ++(30:0.4) -- +(30:1);
\draw (1.5,1.8) node {$\alpha_\downarrow$};
\draw (1.5,0.8) node {$\beta_\downarrow$};
\end{scope}
\end{tikzpicture}
\end{center}
\caption{Left: The magnetic flux through each upright (resp.\@ downright) triangle of side length $1$ equals $\alpha$ (resp.\@ $\beta$). Right: Upon one step of spectral decimation, the magnetic flux through each upright (resp.\@ downright) triangle of side length $2$ equals $\alpha_\downarrow$ (resp.\@ $\beta_\downarrow$), defined in \eqref{eq:alphadown} and \eqref{eq:betadown}, \emph{cf.\@} Proposition \ref{prop:fluxevolve}.}
\label{fig:fluxevolve}
\end{figure}
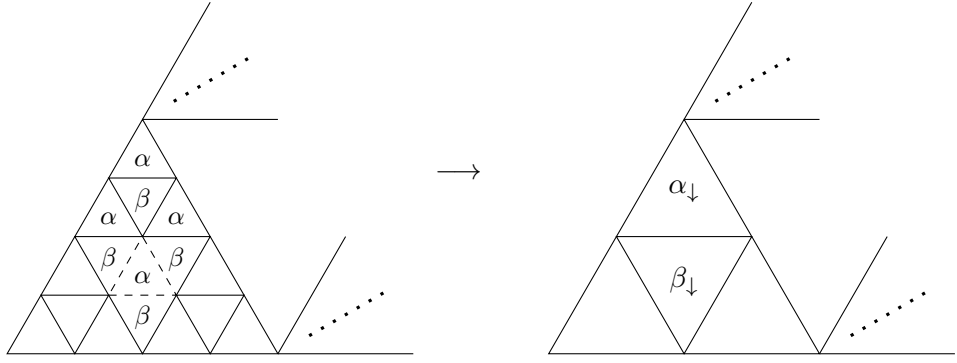

It is easy to check that all the triangles in $SG$ are independent cycles. 
Therefore, for every $N$, there is a well-defined unitary connection $\omega$ on $G_N$ which satisfies Definition \ref{def:flux2}, unique up to gauge equivalence.

\emph{\textbf{Notation.}}
For an operator $\mathcal{L}$, we denote by $\sigma(\mathcal{L}):= \{z\in \mathbb{C}: \mathcal{L}-zI \text{ is not invertible}\}$ the \emph{spectrum} of $\mathcal{L}$. 
The notation $\mult(\mathcal{L}, \lambda)$ (resp.\@ $\mult(P,\lambda)$) represents the \emph{multiplicity} of $\lambda \in \mathbb{C}$ in $\sigma(\mathcal{L})$  (resp.\@ in the zero set of a polynomial function $P$).
In particular, $\mult(\mathcal{L},\lambda)=0$ means that $\lambda \notin \sigma(\mathcal{L})$.
Last but not least, given a function $f: \mathbb{C}\to\mathbb{C}$ and $k\in \mathbb{N}$, we denote by $f^{-k}(a):=\{z\in \mathbb{C}: f^{\circ k}(z)=a\}$ the set of $k$th backward iterates of $a\in \mathbb{C}$ under $f$.
If $f$ is a polynomial of degree $d$, then $f^{-k}(a)$ consists of $d^k$ points, counted with multiplicity.


As mentioned above, we are interested in characterizing $\sigma(\mathcal{L}^{(\alpha, \beta)}_N)$ as a set of real numbers with multiplicities.
To describe our first main result, we introduce the quadratic polynomials
\begin{align*}
&R(0,0,\lambda)=\lambda(5-4\lambda),
\quad R\left(\frac{1}{2},\frac{1}{2},\lambda\right)=-(\lambda-2)(4\lambda-3),\\
&R\left(\frac{1}{2},0,\lambda\right)=-4\lambda^2+9\lambda-3,
 \quad R\left(0,\frac{1}{2},\lambda\right)=-4\lambda^2+7\lambda-1.
\end{align*}
These four polynomials appear as special cases of $R(\alpha,\beta,\lambda)$ in \eqref{eq:R} below.

In the case $\alpha=\beta=0$, \emph{i.e.,} the graph Laplacian, Fukushima and Shima \cite{FukushimaShima} showed that $\sigma(\mathcal{L}_N^{(0,0)})$ consists of the following eigenvalues:
\setlength{\tabcolsep}{10pt}
\renewcommand{\arraystretch}{1.5}
\begin{center}
\begin{tabular}{|c|c|c|}
\hline
Eigenvalue & Condition on $k$ & Multiplicity\\ \hline
$0$ & - & $1$ \\ \hline
${\color{red} \frac{3}{2}}$ & - & $\frac{3^N+3}{2}$ \\  \hline
$(R(0,0,\cdot))^{-k} \left(\frac{3}{4}\right)$ & $k\in \{0,1,\dotsc, N-1\}$ & $\frac{3^{N-k-1}+3}{2}$ \\ \hline
$(R(0,0,\cdot))^{-k} \left(\frac{5}{4}\right)$ & $k\in \{0,1,\dotsc, N-2\}$ & $\frac{3^{N-k-1}-1}{2}$ \\ \hline
\end{tabular}
\end{center}
The eigenvalue colored in {\color{red} red} lies in the exceptional set for spectral decimation. 
In the analysis on fractals literature \cites{3n-gasket, StrichartzBook}, the sets of preimages $\{(R(0,0,\cdot))^{-k}(\frac{3}{4})\}_k$ and $\{(R(0,0,\cdot))^{-k}(\frac{5}{4})\}_k$ are called the \emph{$\frac{3}{4}$-series} and \emph{$\frac{5}{4}$-series}, respectively. 
We will recall this terminology in \S\ref{sec:solution}.
As a sanity check, let us count the eigenvalues listed in the table, noting that $(R(0,0,\cdot))^{-k}(a)$ consists of $2^k$ points counted with multiplicity:
\begin{equation}
1+\frac{3^N+3}{2}+\sum_{k=0}^{N-1}\frac{3^{N-k-1}+3}{2}\cdot 2^k+\sum_{k=0}^{N-2}\frac{3^{N-k-1}-1}{2}\cdot 2^k=\frac{3^{N+1}+3}{2}=\dim_N.
\end{equation}

We claim that such an explicit characterization of $\sigma(\mathcal{L}_N^{(\alpha, \beta)})$ holds for any $\alpha, \beta \in \{0,\frac{1}{2}\}$.

\begin{theorem}
\label{thm:1}
Suppose $\alpha,\beta \in \{0,\frac{1}{2}\}$ but not $\alpha=\beta=0$.
Then $\sigma(\mathcal{L}_N^{(\alpha,\beta)})$ consists of the following eigenvalues counted with multiplicities. (Eigenvalues in {\color{red}red} lie in the exceptional set for spectral decimation.)
\begin{enumerate}[wide]
\item $\sigma(\mathcal{L}_N^{(\frac{1}{2},\frac{1}{2})})$:
\begin{center}
\begin{tabular}{|c|c|c|}
\hline
 Eigenvalue & Condition on $k$ & Multiplicity\\
\hline
{\color{red}$\frac{1}{2}$} & - & $\frac{3^N+3}{2}$ \\
\hline
{\color{red}$\frac{3}{4}$} & - & $\frac{3^{N-1}-1}{2}$ \\
\hline
$\frac{5}{4}$ & - &$\frac{3^{N-1}+3}{2}$ \\
\hline
$2$ & - & $1$ \\
\hline
 $\left(R\left(\frac{1}{2},\frac{1}{2},\cdot\right)\right)^{-1}\circ\left(R(0,0,\cdot)\right)^{-k}\left(\frac{3}{4}\right)$ &$k\in\{0,1,\dotsc,N-2\}$ & $\frac{3^{N-k-2}+3}{2}$  \\
\hline
 $\left(R\left(\frac{1}{2},\frac{1}{2},\cdot\right)\right)^{-1}\circ\left(R(0,0,\cdot)\right)^{-k}\left(\frac{5}{4}\right)$ &$k\in\{0,1,\dotsc,N-3\}$ &$\frac{3^{N-k-2}-1}{2}$ \\
\hline
\end{tabular}
\end{center}

\item $\sigma(\mathcal{L}^{(\frac{1}{2}, 0)}_N)$:
\begin{center}
\begin{tabular}{|c|c|c|}
\hline
Eigenvalue & Condition on $k$  & Multiplicity \\
\hline
{\color{red}$\frac{1}{2}$} &-& $\frac{3^N+3}{2}$ \\
\hline
 $1$ & -& $1$ \\
\hline
{\color{red}$\frac{5}{4}$} &-& $\frac{3^{N-1}-1}{2}$ \\
\hline
$\frac{7}{4}$ &-& $\frac{3^{N-1}+3}{2}$ \\
\hline
$\left(R\left(\frac{1}{2},0,\cdot\right)\right)^{-1}\left(\frac{3}{4}\right)$ &-& $\frac{3^{N-2}-1}{2}$ \\
\hline
$\left(R\left(\frac{1}{2},0,\cdot\right)\right)^{-1}\left(\frac{5}{4}\right)$ &-& $\frac{3^{N-2}+3}{2}$ \\
\hline
 $\left(R\left(\frac{1}{2},0,\cdot\right)\right)^{-1}\circ
\left(R\left(\frac{1}{2},\frac{1}{2},\cdot\right)\right)^{-1}\circ
\left(R(0,0,\cdot)\right)^{-k}\left(\frac{3}{4}\right)$ &$k\in\{0,1,\dotsc,N-3\}$ & $\frac{3^{N-k-3}+3}{2}$ \\
\hline
$\left(R\left(\frac{1}{2},0,\cdot\right)\right)^{-1}\circ
\left(R\left(\frac{1}{2},\frac{1}{2},\cdot\right)\right)^{-1}\circ
\left(R(0,0,\cdot)\right)^{-k}\left(\frac{5}{4}\right)$ &$k\in\{0,1,\dotsc,N-4\}$ &  $\frac{3^{N-k-3}-1}{2}$ \\
\hline
\end{tabular}
\end{center}

\item $\sigma(\mathcal{L}^{(0, \frac{1}{2})}_N)$:
\begin{center}
\begin{tabular}{|c|c|c|}
\hline
Eigenvalue & Condition on $k$ & Multiplicity \\
\hline
$\frac{1}{4}$ & -& $\frac{3^{N-1}+3}{2}$ \\
\hline
{\color{red}$\frac{3}{4}$} & - &$\frac{3^{N-1}-1}{2}$ \\
\hline
$1$ & - & $1$ \\
\hline
{\color{red}$\frac{3}{2}$} & - & $\frac{3^N+3}{2}$ \\
\hline
$\left(R\left(0,\frac{1}{2},\cdot\right)\right)^{-1}\left(\frac{3}{4}\right)$ & - & $\frac{3^{N-2}-1}{2}$ \\
\hline
$\left(R\left(0,\frac{1}{2},\cdot\right)\right)^{-1}\left(\frac{5}{4}\right)$ & - & $\frac{3^{N-2}+3}{2}$ \\
\hline
$\left(R\left(0,\frac{1}{2},\cdot\right)\right)^{-1}\circ
\left(R\left(\frac{1}{2},\frac{1}{2},\cdot\right)\right)^{-1}\circ
(R(0,0,\cdot))^{-k}\left(\frac{3}{4}\right)$ & $k\in\{0,1,\dotsc,N-3\}$ &  $\frac{3^{N-k-3}+3}{2}$ \\
\hline
$\left(R\left(0,\frac{1}{2},\cdot\right)\right)^{-1}\circ
\left(R\left(\frac{1}{2},\frac{1}{2},\cdot\right)\right)^{-1}\circ
(R(0,0,\cdot))^{-k}\left(\frac{5}{4}\right)$ & $k\in\{0,1,\dotsc,N-4\}$ &  $\frac{3^{N-k-3}-1}{2}$ \\
\hline
\end{tabular}
\end{center}
\end{enumerate}
\end{theorem}

The spectra $\{\sigma(\mathcal{L}_N^{(\alpha,\beta)}): \alpha,\beta\in \{0,\frac{1}{2}\}\}$ are related via spectral decimation, see Figure \ref{fig:SDdiagram}.


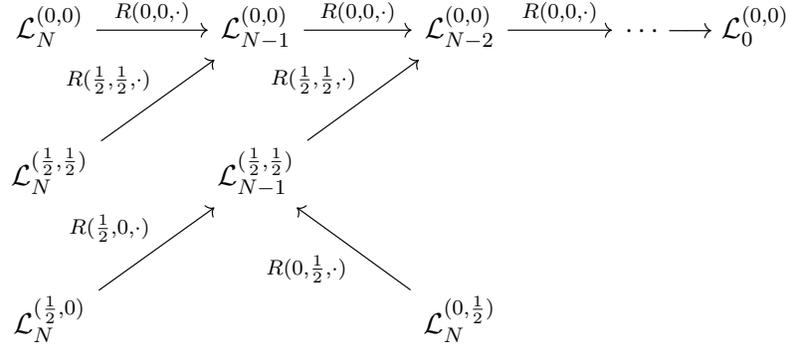
\begin{figure}[h!]
\begin{tikzcd}[column sep=large, row sep=large]
\mathcal{L}^{(0,0)}_N \arrow[r, "{R(0,0,\cdot)}"] & \mathcal{L}_{N-1}^{(0,0)} \arrow[r, "{R(0,0,\cdot)}"] & \mathcal{L}^{(0,0)}_{N-2} \arrow[r, "{R(0,0,\cdot)}"] &
\cdots \longrightarrow \mathcal{L}^{(0,0)}_0 \\
\mathcal{L}^{(\frac{1}{2}, \frac{1}{2})}_N \arrow[ru, "{R(\frac{1}{2},\frac{1}{2},\cdot)}"] &  \mathcal{L}^{(\frac{1}{2}, \frac{1}{2})}_{N-1} \arrow[ru, "{R(\frac{1}{2}, \frac{1}{2}, \cdot)}"] & \\
\mathcal{L}^{(\frac{1}{2},0)}_N \arrow[ru,"{R(\frac{1}{2},0,\cdot)}"] & & \mathcal{L}^{(0,\frac{1}{2})}_N \arrow[lu, "{R(0,\frac{1}{2},\cdot)}"] &
\end{tikzcd}
\caption{A mnemonic for Theorem \ref{thm:1} in the case where the fluxes $\alpha,\beta\in \{0,\frac{1}{2}\}$. Each arrow represents one step of spectral decimation from the magnetic Laplacian on $G_N$ to that on $G_{N-1}$. Details are given in \S\ref{sec:halfflux}.}
\label{fig:SDdiagram}
\end{figure}

It is natural to extend the spectral analysis to the infinite $SG$ lattice $G_\infty$.
Theorem \ref{thm:1} states that each of the 4 spectra contains certain sets of backward iterates under $R(0,0,\cdot)$.
As a consequence, we expect a portion of the (magnetic) spectrum to involve the \textit{Julia set} of $R(0,0,\cdot)$, \emph{viz.\@} the dynamical spectrum referred to by Bellissard.
Let us recall some basic notions from complex dynamics.
The \textbf{Fatou set} $\mathcal{F}(f)$ of a nonconstant holomorphic function $f$ on $\hat{\mathbb{C}}:=\mathbb{C}\cup \{\infty\}$ is the domain in which the sequence of iterates $\{f^{\circ n}\}_n$ converges uniformly on compacts.
The \textbf{Julia set} of $f$ is $\mathcal{J}(f) = \hat{\mathbb{C}} \setminus \mathcal{F}(f)$; by definition it is closed. 
By \cite{Milnor}*{Theorem 14.1}, the Julia set for any rational map of degree $\geq 2$ equals the closure of its set of repelling periodic points.
Also, by \cite{Milnor}*{Corollary 4.13}, if $z_0$ is any point of the Julia set $\mathcal{J}(f)$, then the set of all iterated preimages $\bigcup_{k=0}^\infty f^{-k}(z_0)$ is everywhere dense in $\mathcal{J}(f)$.

The polynomial $R(0,0,\cdot)$ has three fixed points: $\infty$ (attracting), $0$, and $1$ (the latter two are repelling).
Thus $\{0, 1\} \in \mathcal{J}(R(0,0,\cdot))$, and $\bigcup_{k=0}^\infty (R(0,0,\cdot))^{-k}(0)$ is everywhere dense in $\mathcal{J}(R(0,0,\cdot))$.
Since $\frac{5}{4} \in R(0,0,\cdot)^{-1}(0)$, it follows that the closure of $\{0\}\cup \bigcup_{k=0}^\infty R(0,0,\cdot)^{-k}(\frac{5}{4})$ equals $\mathcal{J}(R(0,0,\cdot))$.
Meanwhile, $R(0,0,\frac{3}{4})=\frac{3}{2}$ and $(R(0,0,\cdot))^k (\frac{3}{2}) \to \infty$ as $k\to\infty$. So $\frac{3}{2}$ belongs to the Fatou set $\left(\mathcal{J}(R(0,0,\cdot)\right)^c$, and the same goes for the set of all backward iterates $\bigcup_{k=0}^\infty (R(0,0,\cdot))^{-k} (\frac{3}{4})$ by the invariance of the Fatou set under backward/forward iterates. 
This justifies the decomposition
\[
\sigma(\mathcal{L}^{(0,0)}_\infty) = \mathcal{J}(R(0,0,\cdot)) \cup \left(\bigcup_{k=0}^\infty (R(0,0,\cdot))^{-k} \left(\frac{3}{4}\right)\right)\cup \left\{\frac{3}{2}\right\}
\]
as shown by Teplyaev \cite[Theorem 2]{Teplyaev}.\footnote{
For the graph Laplacian $\mathcal{L}_\infty$ on an infinite, locally finite, connected graph with geometric self-similarity, it is expected that $\sigma(\mathcal{L}_\infty) = \mathcal{J} \cup \mathcal{D}$,
where the set $\mathcal{D}$ depends on the self-similar structure of the graph under study.
See \cite{Malozemov} for illustrating examples.
An example where $\mathcal{D}=\emptyset$ appears in a one-parameter family of self-similar ``$pq$-Laplacians'' on $\mathbb{Z}_+$ \cite{Halfline}.
}
In the same paper Teplyaev proved that the spectral type of $\sigma(\mathcal{L}^{(0,0)}_\infty)$ is pure point, and that each eigenvalue has infinite multiplicity.

Following the same rationale as \cite{Teplyaev}, we arrive at the following corollary.
Here the shorthand $\mathcal{F}_{\frac{3}{4}} := \bigcup_{k=0}^\infty (R(0,0,\cdot))^{-k} \left(\frac{3}{4}\right)$ is used.

\begin{corollary}
\label{cor:magspecinfiniteSG}
We have
\begin{align*}
\sigma(\mathcal{L}^{(\frac{1}{2},\frac{1}{2})}_\infty) &= \left(R\left(\frac{1}{2}, \frac{1}{2},\cdot\right)\right)^{-1} \left[\mathcal{J}(R(0,0,\cdot))\cup \mathcal{F}_{\frac{3}{4}}\right] 
\cup \left\{\frac{1}{2}, \frac{5}{4}\right\};\\
\sigma(\mathcal{L}^{(\frac{1}{2},0)}_\infty) &= \left(R\left(\frac{1}{2}, 0,\cdot\right)\right)^{-1}\circ \left(R\left(\frac{1}{2}, \frac{1}{2},\cdot\right)\right)^{-1} \left[\mathcal{J}(R(0,0,\cdot))\cup \mathcal{F}_{\frac{3}{4}}\right]  \\
&\quad \cup \left\{ \frac{1}{2}, \frac{7}{4}\right\} \cup \left(R\left(\frac{1}{2},0,\cdot\right)\right)^{-1} \left(\left\{\frac{3}{4}, \frac{5}{4}\right\}\right);\\
\sigma(\mathcal{L}^{(0,\frac{1}{2})}_\infty) &=\left(R\left(0,\frac{1}{2},\cdot\right)\right)^{-1}\circ \left(R\left(\frac{1}{2}, \frac{1}{2},\cdot\right)\right)^{-1} \left[\mathcal{J}(R(0,0,\cdot))\cup \mathcal{F}_{\frac{3}{4}}\right]  \\
&\quad \cup \left\{ \frac{1}{4}, \frac{3}{2}\right\}\cup \left(R\left(0,\frac{1}{2},\cdot\right)\right)^{-1} \left(\left\{\frac{3}{4}, \frac{5}{4}\right\}\right).
\end{align*}
In particular, the type of each of the three spectra is pure point, and each eigenvalue has infinite multiplicity.
\end{corollary}

To summarize: in the case $\alpha,\beta \in \{0,\frac{1}{2}\}$, the spectrum $\sigma(\mathcal{L}_\infty^{(\alpha,\beta)})$ consists of (a preimage of) the Julia set of $R(0,0,\cdot)$, as well as points which are preimages of isolated points in the Fatou set of $R(0,0,\cdot)$.

The situation where not both of the fluxes $\alpha,\beta$ are in $\left\{0,\frac{1}{2}\right\}$ is more delicate.
To state our result, we introduce the \textbf{exceptional set} for spectral decimation,
\begin{align}
\label{eq:Eab}
\mathcal{E}(\alpha,\beta) &=\{\lambda\in\mathbb{R} : \Psi(\alpha,\beta,\lambda)=0 \text{ or } \mathcal{D}(\beta,\lambda)=0\},
\end{align}
as well as the following functions,
\begin{align}
\label{eq:firstR} R(\alpha,\beta,\lambda)&=1+\frac{A(\alpha,\beta,\lambda)-64\mathcal{D}(\beta,\lambda)(1-\lambda)}{16|\Psi(\alpha,\beta,\lambda)|};\\
A(\alpha,\beta,\lambda)&=16\lambda^2-(32+4\cos(2\pi\alpha))\lambda+15+4\cos(2\pi\alpha)+\cos(2\pi(\alpha+\beta));\\
\mathcal{D}(\beta,\lambda)&=-\lambda^3+3\lambda^2-\frac{45}{16}\lambda+\frac{13}{16}-\frac{1}{32}\cos(2\pi\beta);\\
\Psi(\alpha,\beta,\lambda)&=(1-\lambda)^2-\frac{1}{16}+\frac{1-\lambda}{4}(2e^{-2\pi i\alpha}+e^{-2\pi i(2\alpha+\beta)})+\frac{1}{16}(e^{-4\pi i \alpha}+2e^{-2\pi i (\alpha+\beta)});\\
\theta(\alpha,\beta,\lambda)&=\frac{\textup{arg} \hspace{0.1cm} \Psi(\alpha,\beta,\lambda)}{2\pi} \qquad ({\rm arg}: \mathbb{C} \to [0,2\pi));\\
\label{eq:alphadown} \alpha_\downarrow(\alpha, \beta, \lambda)&=3\alpha+\beta+3\theta(\alpha,\beta,\lambda) \pmod 1;\\
\label{eq:betadown} \beta_\downarrow(\alpha, \beta, \lambda)&=3\beta+\alpha-3\theta(\alpha,\beta,\lambda) \pmod 1.
\end{align}
Note that \eqref{eq:Eab} through \eqref{eq:betadown} are independent of $N$.

\begin{theorem}
\label{thm:2}
Suppose not both of $\alpha$ and $\beta$ are in $\{0,\frac{1}{2}\}$.
Then
\begin{equation}
\label{eq:spec2}
\sigma\left(\mathcal{L}_N^{(\alpha,\beta)}\right)=S_1(\alpha,\beta) \sqcup S_2(\beta) \sqcup S_3(\alpha),
\end{equation}
where
\begin{align}
S_1(\alpha,\beta) & = \left\{\lambda\in\mathbb{R} \setminus \mathcal{E}(\alpha,\beta): R(\alpha,\beta, \lambda)\in \sigma\left(\mathcal{L}^{(\alpha_\downarrow(\alpha,\beta,\lambda), \beta_\downarrow(\alpha,\beta,\lambda))}_{N-1}\right)\right\}, \\
S_2(\beta) & = \left\{\lambda \in \mathbb{R}: \mathcal{D}(\beta, \lambda)=0,~\mult\left(\mathcal{L}^{(\alpha,\beta)}_N, \lambda\right)>0\right\}, \\
\text{and }\quad S_3(\alpha) & = \left\{
\begin{array}{ll}
 \frac{3}{2},&\text{if } \alpha=0\\
\frac{1}{2},&\text{if } \alpha=\frac{1}{2}
\end{array}
\right\}.
\end{align}
Concerning the multiplicity of each eigenvalue: $\lambda \in S_1$ has multiplicity $\mult\left(\mathcal{L}^{(\alpha_\downarrow(\alpha,\beta,\lambda),\beta_\downarrow(\alpha,\beta,\lambda)}_{N-1}, R(\alpha, \beta, \lambda)\right)$; $\lambda \in S_3$ has multiplicity $\frac{3^N+3}{2}$; and $\lambda\in S_2$ has multiplicity given in Proposition \ref{prop:31}-\ref{G2} and Proposition \ref{prop:32}-\ref{II2} below, which is too complicated to be described here.
\end{theorem}

The set $S_1$ on the RHS of \eqref{eq:spec2} is driven by a 3-parameter map
\begin{align}
\label{eq:3param}
\mathcal{U}: \mathbb{T}^2 \times \mathbb{R} \to \mathbb{T}^2 \times \mathbb{R}, \quad 
(\alpha, \beta, \lambda) \mapsto \left(\alpha_\downarrow(\alpha, \beta, \lambda), \beta_\downarrow(\alpha, \beta, \lambda), R(\alpha, \beta, \lambda)\right),
\end{align}
where $\mathbb{T}=\mathbb{R}/\mathbb{Z}$ denotes the unit torus.
Observe the full dependence of the image triple on the domain triple.
Unlike Theorem \ref{thm:1}, the spectral decimation function $R(\alpha,\beta,\cdot)$ in Theorem \ref{thm:2} is a non-rational function.
And given how the flux variables evolve under $\mathcal{U}$, it is generally not possible to describe the backward iterates of \eqref{eq:3param}.
The more natural approach is to study forward iterates of \eqref{eq:3param}.
We will discuss the dynamical implications in the next subsection \S\ref{sec:butterfly}.

Note, however, that in $S_1$ we have excluded points in the exceptional set $\mathcal{E}(\alpha,\beta)$.
Determining which of the exceptional values belong to the spectrum is usually the trickiest part of the spectral decimation program.
In Theorem \ref{thm:2} we have identified them in $S_2$ and $S_3$ on the RHS of \eqref{eq:spec2}.

\begin{remark}
Historically, Alexander \cite{A84} had obtained a 3-parameter map for the magnetic adjacency operator.
The authors of \cites{A84, DABK83,Ghez} have shown existence of some of these exceptional values in an \emph{ad hoc} manner, without giving a systematic proof. 
The present work completes the enumeration of the exceptional values in a self-consistent framework.
\end{remark}

The set $S_2$ includes \emph{at most} the three zeros of the cubic polynomial $\mathcal{D}(\beta,\cdot)$, whose graph is shown in Figure \ref{fig:Dplot}.
It is easy to verify that $\mathcal{D}(\beta,\cdot)$ does not have a zero of multiplicity $3$, and has a double zero only when $\beta\in \{0,\frac{1}{2}\}$---namely, $\frac{5}{4}$ when $\beta=0$, and $\frac{3}{4}$ when $\beta=\frac{1}{2}$---see Lemma \ref{lem:multzero} below.
Moreover, since $\mathcal{D}(\beta,\cdot)$ for two different values of $\beta$ differ by an additive constant, we see that the smallest zero of $\mathcal{D}(\beta,\cdot)$ lies in $[\frac{1}{2}, \frac{3}{4}]$; the middle zero, $[\frac{3}{4}, \frac{5}{4}]$; and the largest zero, $[\frac{5}{4}, \frac{3}{2}]$.

\begin{figure}
\centering
\includegraphics[width=0.4\textwidth]{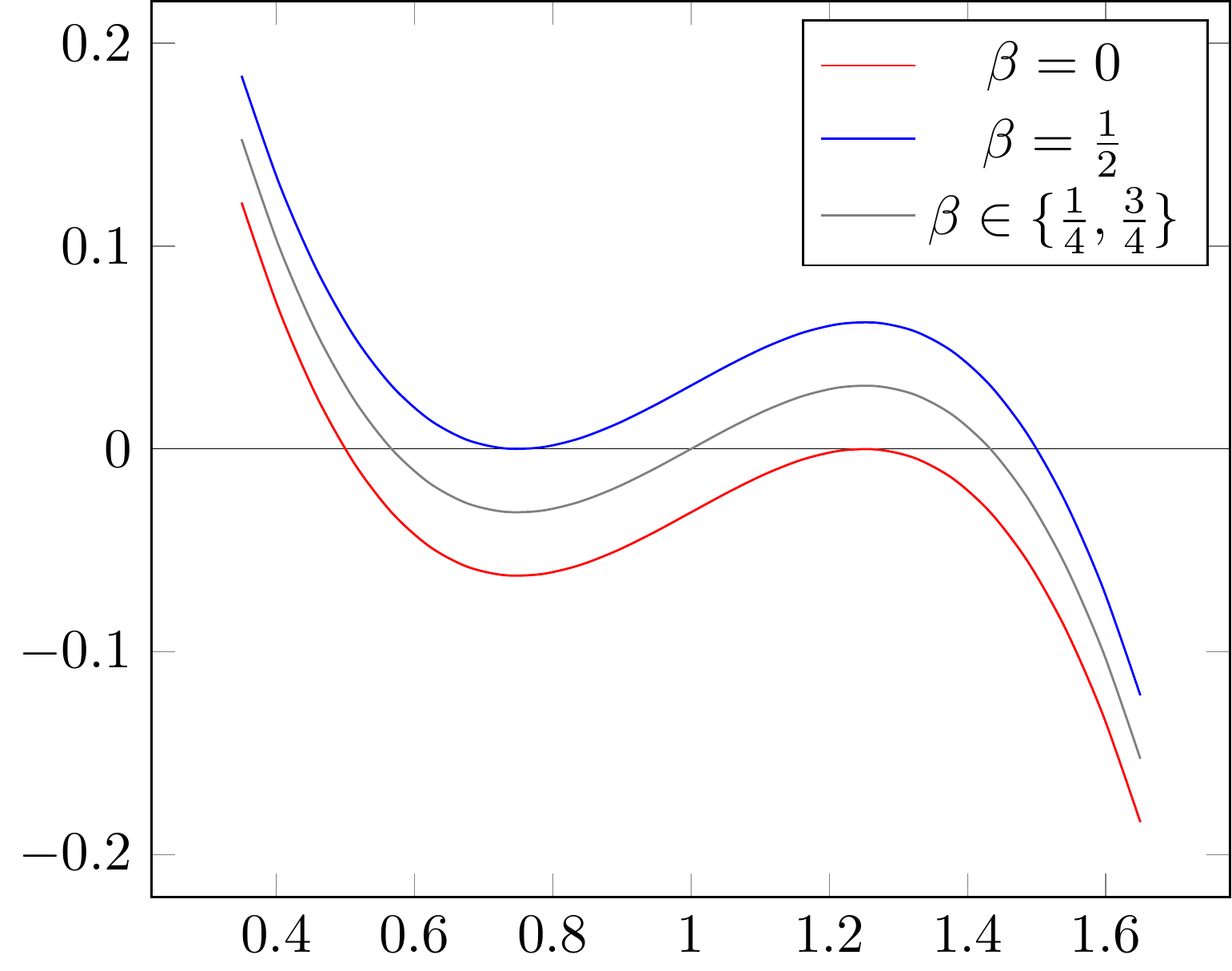}
\caption{The graph of $\mathcal{D}(\beta,\cdot)$.}
\label{fig:Dplot}
\end{figure}

Identifying which zeros of $\mathcal{D}(\beta,\cdot)$ appear in the spectrum is a complicated task, and we defer the case-by-case determination to the latter part of \S\ref{sec:nonhalfintflux}.
That said, we can make the following statements based on the proofs to be presented there.
\begin{proposition}
\label{prop:zerospec}
Suppose not both of $\alpha$ and $\beta$ are in $\{0,\frac{1}{2}\}$.
\begin{enumerate}
\item \label{p1} If $\lambda$ is a simple zero of $\mathcal{D}(\beta,\cdot)$ and $\lambda \in \sigma\left(\mathcal{L}^{(\alpha,\beta)}_N\right)$, then $R(\alpha, \beta, \lambda) \in \sigma\left(\mathcal{L}^{(\alpha_\downarrow(\alpha,\beta,\lambda),\beta_\downarrow(\alpha,\beta,\lambda))}_{N-1}\right)$.
\item \label{p2} If $\lambda$ is a double zero of $\mathcal{D}(\beta,\cdot)$, then for generic values of $\alpha$, we have that $\lambda \in  \sigma\left(\mathcal{L}^{(\alpha,\beta)}_N\right)$ whenever $N\geq 3$, with 
$$\mult\left(\mathcal{L}^{(\alpha, \beta)}_N,\lambda\right)= \frac{3^{N-1}-3}{2}+\mult\left(\mathcal{L}^{(\alpha_\downarrow(\alpha,\beta,\lambda), \beta_\downarrow(\alpha,\beta,\lambda))}_{N-1},R(\alpha,\beta,\lambda)\right).$$
The exceptions are when $\alpha\in \{\frac{1}{6},\frac{5}{6}\}$, $\beta=0$ and $\lambda=\frac{5}{4}$, or when $\alpha\in \{\frac{1}{3}, \frac{2}{3}\}$, $\beta=\frac{1}{2}$ and $\lambda=\frac{3}{4}$, in which case further analysis is required to determine the multiplicity of $\lambda$.
\end{enumerate}
\end{proposition}

Theorem \ref{thm:2} and Proposition \ref{prop:zerospec} imply the following result.
\begin{corollary}
\label{cor:supspec}
Suppose $\beta\notin\{0,\frac{1}{2}\}$.
Then
\begin{equation}
\label{eq:spec2ineq}
\sigma\left(\mathcal{L}_N^{(\alpha, \beta)}\right)
\subseteq
\left\{\lambda\in\mathbb{R} : \Psi(\alpha,\beta,\lambda)\neq 0,~ R(\alpha,\beta, \lambda)\in \sigma\left(\mathcal{L}^{(\alpha_\downarrow(\alpha,\beta,\lambda), \beta_\downarrow(\alpha,\beta,\lambda))}_{N-1}\right)\right\}
\sqcup
\left\{
\begin{array}{ll}
\frac{3}{2},&\text{if } \alpha=0\\
\frac{1}{2},&\text{if } \alpha=\frac{1}{2}
\end{array}
\right\}.
\end{equation}
If in addition $\alpha \notin \{0,\frac{1}{2}\}$, then
\begin{equation}
\label{eq:spec2ineqsimple}
\begin{aligned}
&\left\{\lambda\in\mathbb{R} : \mathcal{D}(\beta,\lambda) \neq0,~ R(\alpha,\beta, \lambda)\in \sigma\left(\mathcal{L}^{(\alpha_\downarrow(\alpha,\beta,\lambda), \beta_\downarrow(\alpha,\beta,\lambda))}_{N-1}\right)\right\}\\
&\subseteq 
\sigma\left(\mathcal{L}_N^{(\alpha, \beta)}\right)
\subseteq
\left\{\lambda\in\mathbb{R} :  R(\alpha,\beta, \lambda)\in \sigma\left(\mathcal{L}^{(\alpha_\downarrow(\alpha,\beta,\lambda), \beta_\downarrow(\alpha,\beta,\lambda))}_{N-1}\right)\right\}.
\end{aligned}
\end{equation}
\end{corollary}


\subsection{The Hofstadter-Sierpinski butterfly}
\label{sec:butterfly}

\begin{figure}
    \centering
    \includegraphics[width=0.85\textwidth]{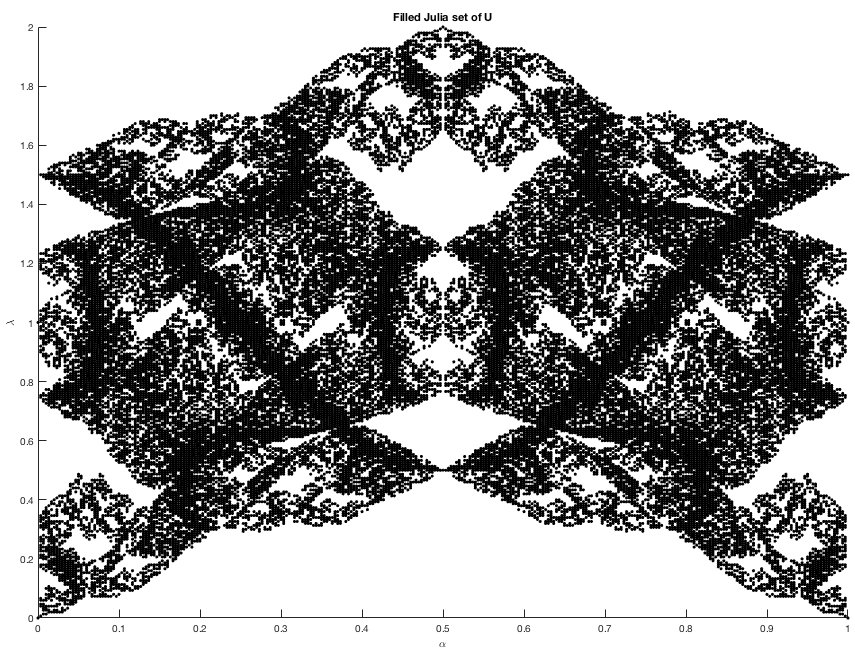}
    \caption{(The Hofstadter-Sierpinski butterfly.) The intersection of the filled Julia set of the 3-parameter map $\mathcal{U}:(\alpha,\beta, \lambda)\mapsto \left(\alpha_\downarrow(\alpha,\beta,\lambda), \beta_\downarrow(\alpha,\beta,\lambda),R(\alpha,\beta,\lambda)\right)$ and $(\alpha,\alpha) \times \mathbb{R}$. We claim that this gives the correct approximation of $\sigma(\mathcal{L}_\infty^{(\alpha,\alpha)})$ (see \S\ref{sec:butterfly} for details). See Appendix \ref{app} for the MATLAB code used to generate this figure.}
    \label{fig:newspectrum}
\end{figure}

We now discuss implications for the magnetic spectrum $\sigma(\mathcal{L}_\infty^{(\alpha,\beta)})$ on the infinite $SG$ lattice $G_\infty$.
Let $\mathcal{K}(\mathcal{U})$ denote the set of all $((\alpha,\beta),\lambda) \in \mathbb{T}^2 \times \mathbb{C}$ for which the forward orbit $\{\mathcal{U}^{\circ k} (\alpha,\beta,\lambda)\}_{k=0}^\infty$ is bounded, also known as the \textbf{filled Julia set} of $\mathcal{U}$.
By definition, $\mathcal{K}(\mathcal{U})$ is the complement of the basin of attraction to infinity, $\mathbb{T}^2 \times \{\infty\}$.
The topological boundary of $\mathcal{K}(\mathcal{U})$ is the Julia set $\mathcal{J}(\mathcal{U})$.
Thus $\mathcal{K}(\mathcal{U})$ is equal to the union of all bounded components of the Fatou set $(\mathbb{T}^2 \times \hat{\mathbb{C}}) \setminus \mathcal{J}(\mathcal{U})$, together with the Julia set $\mathcal{J}(\mathcal{U})$.

\begin{theorem}
\label{thm:specinf}
Suppose not both of $\alpha$ and $\beta$ are in $\{0,\frac{1}{2}\}$. Then
\begin{align}
\label{eq:sigmaLinfty}
\sigma(\mathcal{L}^{(\alpha,\beta)}_\infty) = S^\infty_1(\alpha,\beta) \cup S^\infty_2(\beta) \cup S^\infty_3(\alpha),
\end{align}
where
\begin{align}
S^\infty_1(\alpha,\beta) &\subset \mathcal{K}(\mathcal{U}) \cap \left( (\alpha,\beta) \times \mathbb{R}\right),\\
\emptyset \subset S^\infty_2(\beta) &\subset \{\lambda\in \mathbb{R}: \mathcal{D}(\beta,\lambda)=0\},\\
S^\infty_3(\alpha) &=  \left\{
\begin{array}{ll}
 \frac{3}{2},&\text{if } \alpha=0\\
\frac{1}{2},&\text{if } \alpha=\frac{1}{2}
\end{array}
\right\}.
\end{align}
\end{theorem}

\begin{conjecture}
\label{conj:S1}
$S^\infty_1(\alpha,\beta) \supset \mathcal{J}(\mathcal{U}) \cap ((\alpha,\beta)\times \mathbb{R})$.
\end{conjecture}

On the one hand, Theorem \ref{thm:specinf} follows from the strong convergence of $\mathcal{L}^{(\alpha,\beta)}_N$ to $\mathcal{L}^{(\alpha,\beta)}_\infty$, Theorem \ref{thm:2}, and the definition of the filled Julia set of $\mathcal{U}$.
On the other hand, we are unable to prove Conjecture \ref{conj:S1} at the moment, given the complicated dynamics of the 3-parameter non-rational map $\mathcal{U}$.
For instance, we do not know the answer to this basic question: If $(\alpha,\beta,z)\in \mathcal{J}(\mathcal{U})$, is the union of its backward iterates, $\bigcup_{k=0}^\infty \mathcal{U}^{\circ -k}(\alpha,\beta,z)$, everywhere dense in $\mathcal{J}(\mathcal{U})$?\footnote{The answer is yes for 1-parameter rational functions on $\hat{\mathbb{C}}$ of degree $\geq 2$ \cite[Corollary 4.13]{Milnor}. This provides an algorithm for numerically generating pictures of the Julia set of a rational function.}
If the answer is affirmative, then Conjecture \ref{conj:S1} may be proved in the same way as \cite[Theorem 5.8(3)]{Malozemov} or \cite[Theorem 11]{Halfline} using input from Theorem \ref{thm:2}.

In any case, this is our best quantitative answer to Bellissard's question, concerning the relationship between the dynamical spectrum and the actual spectrum.
To summarize: While part of the magnetic spectrum is recursively generated, there are exceptional values which do not arise via this mechanism and carry infinite multiplicity (\emph{cf.\@} Theorem \ref{thm:2}).

Let us specialize Theorem \ref{thm:specinf} to the case $\alpha=\beta$.
Figure \ref{fig:newspectrum} shows $\mathcal{K}(\mathcal{U}) \cap \{(\alpha,\alpha) \times \mathbb{R})$, which resembles a butterfly whose wings have self-similar patterns.
This is the Sierpinski gasket counterpart to the Hofstadter butterfly obtained originally on the square lattice \cite{hofstadter}; for a lack of better name, we shall call it the \textbf{Hofstadter-Sierpinski butterfly}.
To the best of the authors' knowledge, previous attempts at solving the Hofstatder-Sierpinski butterfly were all based on numerical computations on finite-level gasket graphs, \emph{cf.\@} \cite{quasicrystals}*{Figure 2}, \cite{UConnREU}, and \cite{BCN18}*{Figure 2(d)}.

Meanwhile, we would like to correct an inaccurate statement made in the physics literature.
Consider the 2-parameter map $\mathcal{U}_2:(\alpha,\lambda)\mapsto (4\alpha, R(\alpha,\alpha,\lambda))$, whose filled Julia set is shown in Figure \ref{fig:spectrum}, and was presented as \cite{Ghez}*{Figure 2}.
The authors of \cite{Ghez} claimed that this $2$-parameter map produce a good approximation of $\sigma(\mathcal{L}^{(\alpha,\alpha)}_\infty)$.
In the same paper they also provided data from superconductivity measurements showing good agreement between theory and experiment.
By comparing Figures \ref{fig:newspectrum} and \ref{fig:spectrum}, it is safe to conclude that this is not the case.
In fact, Figure \ref{fig:spectrum} produces the correct approximation of $\sigma(\mathcal{L}^{(\alpha,\alpha)}_\infty)$ only when $\alpha \in \{0,\frac{1}{2}\}$, the reason being that $\mathbb{R}\ni \lambda\mapsto \Psi(\alpha, \alpha,\lambda)$ is $\mathbb{R}$-valued at these flux values. Once $\alpha \notin \{0,\frac{1}{2}\}$, $\mathbb{R}\ni \lambda\mapsto \Psi(\alpha,\alpha, \lambda)$ is in general $\mathbb{C}$-valued, and its argument $\theta={\rm arg}\Psi$ must be taken into account when deducing the magnetic fluxes.

\begin{figure}
    \centering
    \includegraphics[width=0.85\textwidth]{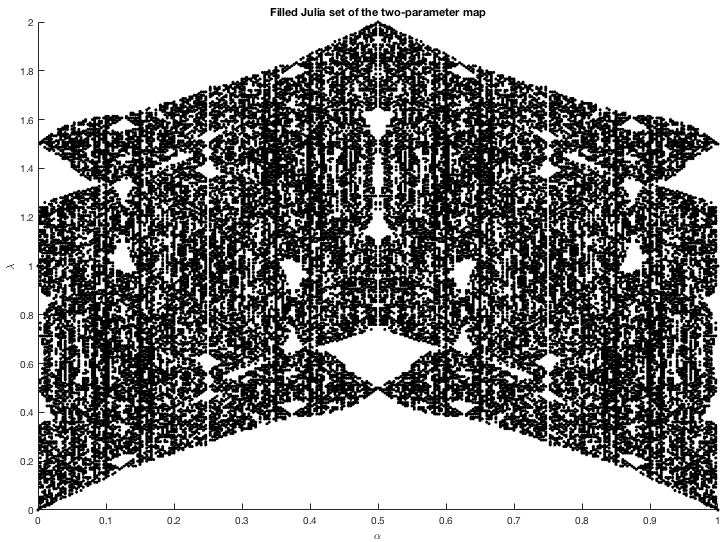}
    \caption{The filled Julia set of the 2-parameter map $\mathcal{U}_2:(\alpha, \lambda)\mapsto \left(4\alpha, R(\alpha, \alpha, \lambda)\right)$, as was used to generate \cite{Ghez}*{Figure 2}.
    This gives the correct approximation of $\sigma(\mathcal{L}^{(\alpha,\alpha)}_\infty)$ only when $\alpha\in \{0,\frac{1}{2}\}$.}
    \label{fig:spectrum}
\end{figure}

The reason for this discrepancy will become clear when we prove Proposition \ref{prop:fluxevolve} below.
But here is the take-away message: If the flux through every upright and downright cell equals $\alpha$, then upon decimation, it is \emph{false} that the flux through a next upright (or downright) cell equal $4\alpha$, despite the fact that the flux through a rhombus (formed by adjoining an upright cell to a downright one) always equals $8\alpha$. 
More importantly, the post-decimation fluxes through each triangle depend on the spectral parameter $\lambda$.

From a technical standpoint, the new aspect of our spectral decimation analysis is that the function $R(\alpha,\beta,\cdot)$ in \eqref{eq:firstR} is not rational, due to the appearance of $|\Psi(\alpha,\beta,\cdot)|$ in the denominator.
All previous mathematical works on spectral decimation \cites{RammalToulouse, FukushimaShima, Shima, Teplyaev, Malozemov, 3n-gasket, Halfline, diamond} involve $R$ rational.
While it may seem an unavoidable nuisance to deal with non-rational functions, we nevertheless can carry out spectral decimation after applying some care.


\subsubsection*{Organization for the rest of the paper}

In \S\ref{sec:specdec} we discuss the general mechanics of spectral decimation. 
In \S\ref{sec:ssSG} we demonstrate the spectral self-similarity of the magnetic Laplacian on $SG$, furnished with all the necessary computations.
These two sections lay the technical groundwork from which we solve the magnetic spectrum on $SG$ in \S\ref{sec:solution}, proving Theorems \ref{thm:1} and \ref{thm:2}, Proposition \ref{prop:zerospec}, Corollary \ref{cor:supspec}, and Theorem \ref{thm:specinf}.
Finally, in \S\ref{sec:det} we provide a combinatorial application of Theorem \ref{thm:1}, establishing formulae for the magnetic Laplacian determinants (Theorem \ref{thm:3}) and the corresponding asymptotic complexities (Corollary \ref{cor:ent}).

\subsubsection*{Acknowledgements}
We thank Alexander Teplyaev and Richard Kenyon for useful conversations during the initial stage of this work;
Quan Vu for his early numerical contributions to cycle-rooted spanning forests on the Sierpinski gasket;
and the anonymous referee for critical comments which helped us improve the paper.


\renewcommand{\arraystretch}{1.25}

\section{Mechanics of spectral decimation} 
\label{sec:specdec}

In this section we give a general account of Schur complementation and the spectral decimation procedure. 
While the essential ideas can be found in \cite{3n-gasket}, the methods therein apply only to spectral decimation functions $R$ which are rational. 
Thus we use this opportunity to explain not only the adaptation to non-rational functions, but also a number of subtleties in the spectral decimation procedure.


\subsection{Schur complement \& functional identities}
\label{sec:Schur}

We start with an elementary matrix identity.
Observe that $A-BD^{-1}C$ is the Schur complement of the block matrix with respect to the $D$ block.

\begin{proposition}
\label{prop:matrixid}
Suppose $\begin{bmatrix} A&B\\C&D\end{bmatrix}$ is a square block matrix with the square block $D$ invertible. Then
\begin{equation}
\label{eq:Schur0}
\begin{bmatrix} A&B\\C&D\end{bmatrix} = \begin{bmatrix}A-BD^{-1}C & BD^{-1} \\ 0 & I\end{bmatrix} \begin{bmatrix} I&0\\C&D\end{bmatrix}.
\end{equation}
Furthermore, the block matrix is invertible with inverse
\begin{equation}
\label{eq:matrixinverse}
\begin{aligned}
\begin{bmatrix}A&B\\C&D\end{bmatrix}^{-1} &
= \begin{bmatrix}I&0\\-D^{-1}C&D^{-1}\end{bmatrix}
\begin{bmatrix}(A-BD^{-1}C)^{-1}& -(A-BD^{-1}C)^{-1}BD^{-1}\\0&I\end{bmatrix} \\
&= \begin{bmatrix}0&0\\0&D^{-1}\end{bmatrix} + \begin{bmatrix}I\\-D^{-1}C\end{bmatrix} (A-BD^{-1}C)^{-1} \begin{bmatrix}I & -BD^{-1}\end{bmatrix}.
\end{aligned}
\end{equation}
\end{proposition}
Observe that \eqref{eq:Schur0} implies the determinant identity $\det\begin{bmatrix} A&B\\C&D\end{bmatrix} = \det(A-BD^{-1}C) \det D$.

Let $V$ be a countable set, $\mu$ be a nonnegative measure on $V$, and
$L^2(V,\mu)$ be the Hilbert space of $\mathbb{C}$-valued functions on $V$
with inner product $\langle f,g\rangle = \sum_{x\in V} \bar{f}(x)g(x) \mu(x)$.
Let $M: L^2(V,\mu)\to L^2(V,\mu)$; equivalently, we may regard $M$ as a square matrix of size $|V|$ with entries $M_{ij} = \langle \delta_i, M\delta_j\rangle$.

Suppose $V=V_\parallel \sqcup V_\perp$.
Naturally, we can project functions in $L^2(V,\mu)$ to $L^2(V_\parallel,\mu)$ and $L^2(V_\perp,\mu)$, respectively, and denote the corresponding projection operators by $P_\parallel$ and $P_\perp$.
Their conjugates are
\begin{align*}
P_b^*: L^2(V_b,\mu) \to L^2(V,\mu), \qquad(P_b^*f)(x) = f(x){\bf 1}_{\{x\in V_b\}},\qquad b\in \{\parallel, \perp\}
\end{align*}
Note that $P_b P_b^* = I_b$, the identity on $L^2(V_b,\mu)$.

Combining the preceding facts, we see that Proposition \ref{prop:matrixid} implies the following.
Suppose $M$ can be expressed in the ``block form''
\begin{align*}
A=P_\parallel M P_\parallel^*, \quad B=P_\parallel M P_\perp^*, 
\quad
C=P_\perp M P_\parallel^*, \quad D= P_\perp M P_\perp^*,
\end{align*}
with $D$ invertible. Then \eqref{eq:matrixinverse} says that
\begin{align*}
M^{-1} = P_\perp^* D^{-1} P_\perp + 
\left(P_\parallel^* - P_\perp^* D^{-1}C\right)
\left(A-BD^{-1}C\right)^{-1} 
\left(P_\parallel - BD^{-1} P_\perp\right).
\end{align*}
For spectral analysis it is more pertinent to consider the resolvent $(M-xI)^{-1}$, $x\in\mathbb{C}$. 
In this case, assuming that $D-x I_\perp$ is invertible, we have
\begin{equation}
\label{eq:M-x-1}
\begin{aligned}
(&M-x I)^{-1} = P_\perp^* (D-x I_\perp)^{-1} P_\perp \\&+ 
\left(P_\parallel^* - P_\perp^* (D-x I_\perp)^{-1}C\right)
\left((A-x I_\parallel)-B(D-x I_\perp)^{-1}C\right)^{-1} 
\left(P_\parallel - B(D-x I_\perp)^{-1} P_\perp\right).
\end{aligned}
\end{equation}
(From this point on we drop the notation $I_\parallel$ or $I_\perp$, unless the context demands its presence.)

Finally, recall the functional calculus $f(M) = \sum_{\lambda \in \sigma(M)} f(\lambda) E_\lambda(M)$, where $E_\lambda(M): L^2(V,\mu) \to L^2(V,\mu)$ is the eigenprojector of $M$ associated with eigenvalue $\lambda$.
It is then direct to verify that
\begin{align}
\label{eq:eigp}
E_\lambda(M) = \lim_{x\to\lambda} (\lambda-x)(M-x)^{-1}.
\end{align}
We will be especially interested in expressing the eigenprojector in terms of $A$, $B$, $C$, and $D$, using the RHS of \eqref{eq:eigp} in conjunction with the formula \eqref{eq:M-x-1}. 
Of course we will need to justify the limit as $x\to\lambda \in \sigma(M)$, which will be done on a case-by-case basis.

\begin{remark}
In the above discussion, there is no loss of generality replacing $L^2(V,\mu)$ by $\ell^2(V)$.
That said, we will soon assume that $M$ is a self-adjoint operator on a Hilbert space, and this will require invocation of the measure $\mu$.
\end{remark}


\subsection{Spectral decimation for the non-exceptional values}
\label{sec:non-exceptional}

Let us introduce the following condition which will be in force for the rest of the section.

\begin{definition}[Spectral similarity]
\label{def:ss}
Let $V_\parallel \subset V$.
We say that two self-adjoint operators $M: L^2(V,\mu) \to L^2(V,\mu)$ and $L:L^2(V_\parallel,\mu) \to L^2(V_\parallel,\mu)$ are \textbf{spectrally similar} if there exist scalar-valued functions $\phi$ and $R$ which map $\mathbb{R}$ to $\mathbb{R}$ such that
\begin{align}
\label{eq:specs}
(A-x)-B(D-x)^{-1} C = \phi(x)(L-R(x))
\end{align}
for all $x\in \mathbb{C}$.
It follows that
\begin{align}
\label{eq:specsim}
P_\parallel (M-x)^{-1} P_\parallel^*\left(\underset{\eqref{eq:M-x-1}}{=} \left((A-x)-B(D-x)^{-1} C\right)^{-1}\right) = [\phi(x)]^{-1}(L-R(x))^{-1}
\end{align}
for all $x\in \mathbb{C}$ whenever the RHS is defined.
\end{definition}

\begin{remark}
In Definition \ref{def:ss}, no assumption is made on the dependence of $M$ or $L$ on the spectral parameter $x$. 
(In the example of Corollary \ref{cor:specdec}-\ref{item:C}, $L=\mathcal{L}^\Omega_{N-1}$ depends on $x$.)
Also we do not specify extra conditions (such as continuity or differentiability) on $\phi$ and $R$ at the moment.
\end{remark}

In order for \eqref{eq:specs} and \eqref{eq:specsim} to make sense as they are, $D-x$ should be invertible, and $\phi(x)\neq 0$. 
Any $x$ that fails either condition is said to be \textbf{exceptional}, and we refer to the set of all such $x$ as the \textbf{exceptional set} for spectral decimation, denoted
\begin{equation}
\label{eq:eset}
\mathcal{E}=\{x\in \mathbb{C} : \text{$x\in\sigma(D)$ or $\phi(x)=0$}\}.
\end{equation}

Since $M$ is self-adjoint on $L^2(V,\mu)$, our goal is to determine which $\lambda\in \mathbb{R}$ belongs to the spectrum $\sigma(M)$.
The following result is the spectral decimation identity when $\lambda\in \mathbb{R}$ is not exceptional, which mirrors \cite[Proposition 4.1]{3n-gasket}.

\begin{lemma}[Spectral decimation for the non-exceptional values]
\label{specdecthm}
Suppose $\lambda\in \mathbb{R}$ is such that $\lambda \notin \mathcal{E}$, and moreover 
$
\displaystyle
\lim_{\mathbb{R}\ni x\to \lambda} \phi(x)\frac{R(\lambda)-R(x)}{\lambda-x}
$
exists and does not equal $0$.
Then
\begin{align}
\label{eq:eigp1}
E_{\lambda}(M)=\left(\lim_{\mathbb{R}\ni x\to \lambda}\frac{1}{\phi(x)}\frac{\lambda-x}{R(\lambda)-R(x)}\right)\left(P_\parallel^* - P_\perp^*(D-\lambda)^{-1}C\right) E_{R(\lambda)}(L) \left(P_\parallel- B(D-\lambda)^{-1} P_\perp\right).
\end{align}
Consequently, $\lambda \in \sigma(M)$ if and only if $R(\lambda)\in \sigma(L)$, and there is a one-to-one correspondence between eigenfunctions of $L$ with eigenvalue $R(\lambda)$ and eigenfunctions of $M$ with eigenvalue $\lambda$, given by
\[
{\rm Image}\left(E_{R(\lambda)}(L)\right)\ni f\mapsto \left(P^*_\parallel - P^*_\perp (D-\lambda)^{-1} C\right)f \in {\rm Image}\left(E_\lambda(M)\right).
\]
In particular, ${\rm mult}(M,\lambda) = {\rm mult}(L, R(\lambda))$.
\end{lemma}
\begin{proof}
Combining \eqref{eq:M-x-1} and \eqref{eq:specsim} we find
\begin{equation}
\label{eq:P0}
\begin{aligned}
(&\lambda-x)(M-x)^{-1} = {\color{blue}P_\perp^*}{\color{purple}(\lambda-x)}{\color{blue} (D-x)^{-1}P_\perp} \\
&+{\color{blue}\left(P_\parallel^* - P_\perp^* (D-x)^{-1} C\right)} (\lambda-x){\color{blue}[\phi(x)]^{-1}}(L-R(x))^{-1} {\color{blue}\left(P_\parallel - B(D-x)^{-1}P_\perp\right)}.
\end{aligned}
\end{equation}
According to \eqref{eq:eigp} it suffices to take the limit of \eqref{eq:P0} as $\mathbb{R}\ni x\to\lambda$.
Based on the assumptions, the quantities in {\color{blue} blue} (resp.\@ {\color{purple} purple}) remain bounded (resp.\@ vanish) in the limit, and in particular the first term on the RHS tends to $0$.
To unravel the second term on the RHS, we insert the identity $I_\parallel = E_{R(\lambda)}(L) + (I_\parallel - E_{R(\lambda)}(L))$ between $(L-R(x))^{-1}$ and $P_\parallel -B(D-x)^{-1}P_\perp$, resulting in the following expression:
\begin{align}
&\label{eq:P1} {\color{blue}\left(P_\parallel^* - P_\perp^* (D-x)^{-1} C\right)} (\lambda-x){\color{blue}[\phi(x)]^{-1}}(L-R(x))^{-1} E_{R(\lambda)}(L) {\color{blue}\left(P_\parallel - B(D-x)^{-1}P_\perp\right)}\\
&\label{eq:P2} +{\color{blue}\left(P_\parallel^* - P_\perp^* (D-x)^{-1} C\right)}{\color{purple} (\lambda-x)}{\color{blue}[\phi(x)]^{-1}(L-R(x))^{-1} (I_\parallel -E_{R(\lambda)}(L)) \left(P_\parallel - B(D-x)^{-1}P_\perp\right)}.
\end{align}
Observe that in \eqref{eq:P2}, the image of $I_\parallel-E_{R(\lambda)}(L)$ is the orthogonal complement of the eigenspace of $L$ with eigenvalue $R(\lambda)$, and $L-R(\lambda)$ is invertible on this space.
Therefore \eqref{eq:P2} vanishes in the limit $x\to\lambda$.
As for \eqref{eq:P1}, we are in the eigenspace of $L$ with eigenvalues $R(\lambda)$, and $L-R(\lambda)$ is not invertible. 
That said, we can multiply and divide \eqref{eq:P1} by $R(\lambda)-R(x)$,
\begin{align}
{\color{blue}\left(P_\parallel^* - P_\perp^* (D-x)^{-1} C\right)} \frac{\lambda-x}{R(\lambda)-R(x)} [\phi(x)]^{-1} (R(\lambda)-R(x))(L-R(x))^{-1} {\color{blue} E_{R(\lambda)}(L)\left(P_\parallel - B(D-x)^{-1}P_\perp\right)}.
\end{align}
By functional calculus again, $\displaystyle \lim_{x\to \lambda} (R(\lambda)-R(x))(L-R(x))^{-1} = E_{R(\lambda)}(L)$.
So the proof is complete provided that $\displaystyle \frac{1}{\phi(x)}\frac{\lambda-x}{R(\lambda)-R(x)}$ has a nonsingular limit as $x\to\lambda$.
\end{proof}

Actually \eqref{eq:eigp1} says more. 
Since the LHS of \eqref{eq:eigp1} is a bounded operator, if $\displaystyle \lim_{R\ni x\to\lambda} \left|\frac{1}{\phi(x)} \frac{\lambda-x}{R(\lambda)-R(x)}\right|=\infty$, then \eqref{eq:eigp1} holds only if $E_{R(\lambda)}(L)=0$.
In turn $E_\lambda(M)=0$.
In what follows, we will encounter similar situations where the scalar prefactor diverges, and we may argue using this rationale that this divergence should not exist.

\subsection{Spectral decimation for the exceptional values}
\label{sec:exceptional}

If $\lambda$ is exceptional, the spectral decimation argument is suitably modified.
Here are two items of note.

\begin{lemma}
\label{lem:mod}
Under Definition \ref{def:ss}:
\begin{enumerate}
\item \label{phineq0} If $\phi(\lambda)\neq 0$, then $(D-\lambda)^{-1}$ is bounded on the image of $E_{R(\lambda)}(L)$.
\item \label{ED0} If both $\phi$ and $\phi R$ are bounded in a neighborhood of $\lambda$, then $BE_\lambda(D)C=0$.
\end{enumerate}
\end{lemma}
\begin{proof}
\eqref{phineq0}: 
If $\phi(\lambda) \neq 0$, we use \eqref{eq:specs} to find that whenever $f$ is an eigenfunction of $L$ with eigenvalue $R(\lambda)$, then
\begin{align}
\label{eq:ERbounded}
\left((A-\lambda)-B(D-\lambda)^{-1} C\right) f = \phi(\lambda) (L-R(\lambda))f=0.
\end{align}
Given that $A-\lambda$, $B$, and $C$ are all bounded, it follows that $(D-\lambda)^{-1}$ must be bounded on the image of $E_{R(\lambda)}(L)$.

\eqref{ED0}:
Multiply \eqref{eq:specs} on both sides by $(\lambda-x)$ to get
\[
(\lambda-x)(A-x) - B(\lambda-x)(D-x)^{-1} C= (\lambda-x) \phi(x) (L-R(x)).
\]
Noting that both $A-x$ and $\phi(x)(L-R(x))$ remain bounded as $x\to\lambda$, we take the limit on the above equation to find $-BE_\lambda(D) C=0$.
\end{proof}

With the above in mind, we continue to use \eqref{eq:M-x-1}, \eqref{eq:eigp}, and \eqref{eq:specsim} altogether to derive an expression for the eigenprojector $E_\lambda(M)$.
The general strategy proceeds as follows: first decide whether $\lambda\in \sigma(D)$ (which determines the invertibility of $D-\lambda$), then insert the identity $I_\parallel =E_{R(\lambda)}(L) + (I_\parallel-E_{R(\lambda)}(L))$ in the expression for $(\lambda-x)(M-x)^{-1}$ \`a la \eqref{eq:P1} and \eqref{eq:P2}, and finally identify conditions which ensure the existence of the limits as $\mathbb{R}\ni x\to \lambda$.

The next result generalizes \cite[Proposition 4.1]{3n-gasket}, in the sense that we only require the existence of $\mathbb{R}$-limits (as opposed to $\mathbb{C}$-limits) of the various functions that arise naturally in the eigenprojector expression.
For the sake of easy reference, we keep the same numbering of the cases as in \cite[Proposition 4.1]{3n-gasket}.

\begin{lemma}
\label{multiplicity}
Suppose $\lambda\in\mathbb{R}$.
\begin{enumerate}[label=\textbf{(\roman*)}, wide]
\setcounter{enumi}{1}
\item \label{eq:SD2} If $\lambda\notin\sigma (D)$, $\phi(\lambda)=0$, and moreover
$\displaystyle
\lim_{\mathbb{R}\ni x\to\lambda} \frac{\phi(x)}{\lambda-x} \neq 0
$
and
$
\displaystyle
\lim_{\mathbb{R}\ni x\to\lambda}\phi(x)\frac{R(\lambda)-R(x)}{\lambda-x} \neq 0
$,
then
\begin{equation}
\begin{aligned}
&E_\lambda(M) = 
\left(P_\parallel^*-P_\perp^*(D-\lambda)^{-1}C\right)
\left(\lim_{\mathbb{R}\ni x\to\lambda}\frac{(\lambda-x)[\phi(x)]^{-1}}{R(\lambda)-R(x)}\right)
E_{R(\lambda)}(L)
\left(P_\parallel-B(D-\lambda)^{-1}P_\perp\right)\\
&+\left(P_\parallel^*-P_\perp^*(D-\lambda)^{-1}C\right)\left(\lim_{\mathbb{R}\ni x\to\lambda}\frac{\lambda-x}{\phi(x)}\right)
(L-R(\lambda))^{-1}
\left(I_\parallel-E_{R(\lambda)}(L)\right)
\left(P_\parallel-B(D-\lambda)^{-1}P_\perp\right).
\end{aligned}
\end{equation}
In particular, ${\rm mult}(M,\lambda)=|V_\parallel|$.

\item \label{eq:SD3} 
If $\lambda\in\sigma (D)$, $\displaystyle \lim_{\mathbb{R}\ni x\to\lambda}[\phi(x)]^{-1}=0$, and moreover
$\displaystyle
\lim_{\mathbb{R}\ni x\to\lambda} \phi(x)(\lambda-x) \neq 0
$
and
$\displaystyle
\left|\frac{\lambda-x}{R(\lambda)-R(x)}\right|
$
is bounded in a neighborhood of $\lambda$,
then
\begin{align}
\label{eq:eigp3}
E_\lambda(M) = P_\perp^* E_\lambda(D) P_\perp + P_\perp^* E_\lambda(D) C\left(\lim_{\mathbb{R}\ni x\to\lambda} \frac{1}{\phi(x)(\lambda-x)}\right)
(L-R(\lambda))^{-1} (I_\parallel-E_{R(\lambda)}(L)) BE_\lambda(D) P_\perp.
\end{align}
In particular, $E_\lambda(M) (P_\perp^* E_\lambda(D) P_\perp)=E_\lambda(M)$, so any eigenfunction of $M$ with eigenvalue $\lambda$ vanishes on $V_\parallel$, and
$
\mult(M,\lambda)=\mult(D,\lambda)-\left(|V_\parallel|-\mult(L,R(\lambda))\right)
$.

\item \label{eq:SD4} 
If $\lambda\in\sigma (D)$, both $\phi$ and $\phi R$ are bounded in a neighborhood of $\lambda$, $\phi(\lambda)\neq0$, and moreover 
$\displaystyle
\lim_{\mathbb{R}\ni x\to\lambda} \phi(x) \frac{R(\lambda)-R(x)}{\lambda-x} \neq 0
$,
then
\begin{equation}
\label{eq:eigp4}
\begin{aligned}
E&_\lambda(M) = P_\perp^* E_\lambda(D) P_\perp \\
&+ \left(P_{\parallel}^*-P_{\perp}^*(D-\lambda)^{-1}C\right) \left(\lim_{\mathbb{R} \ni x\to\lambda}\frac{1}{\phi(x)}\frac{\lambda-x}{R(\lambda)-R(x)}\right)
 E_{R(\lambda)}(L)
\left(P_{\parallel}-B(D-\lambda)^{-1}P_{\perp}\right).
\end{aligned}
\end{equation}
In particular, $E_\lambda(M) (P_\perp^* E_\lambda(D) P_\perp) = P_\perp^* E_\lambda(D) P_\perp$, the two components on the RHS of \eqref{eq:eigp4} are mutually orthogonal in $L^2(V,\mu)$, and
$
\mult(M,\lambda)=\mult(D,\lambda)+\mult(L, R(\lambda))
$.

\setcounter{enumi}{5}
\item \label{eq:SD6}
If $\lambda\in \sigma(D)$, $\displaystyle\lim_{\mathbb{R}\ni x\to \lambda}[\phi(x)]^{-1}=0$, and moreover $\displaystyle \lim_{\mathbb{R}\ni x\to \lambda}\phi(x)(\lambda-x)\neq 0$ and $\displaystyle\lim_{\mathbb{R}\ni x\to\lambda} \frac{1}{\phi(x)}\frac{\lambda-x}{R(\lambda)-R(x)} \neq 0$, then
\begin{equation}
\label{eq:eigp6}
\begin{aligned}
E&_\lambda(M)= P_\perp^* E_\lambda(D) P_\perp + P_\perp^* E_\lambda(D) C\left(\lim_{\mathbb{R}\ni x\to\lambda} \frac{1}{\phi(x)(\lambda-x)}\right)
(L-R(\lambda))^{-1} (I_\parallel-E_{R(\lambda)}(L)) BE_\lambda(D) P_\perp
\\
&+ \left(P_{\parallel}^*-P_{\perp}^*(D-\lambda)^{-1}C\right) \left(\lim_{\mathbb{R} \ni x\to\lambda}\frac{1}{\phi(x)}\frac{\lambda-x}{R(\lambda)-R(x)}\right)
 E_{R(\lambda)}(L)
\left(P_{\parallel}-B(D-\lambda)^{-1}P_{\perp}\right).
\end{aligned}
\end{equation}
which implies generally that $\mult(M,\lambda) = \mult(D,\lambda) - |V_\parallel| + 2\mult(L,R(\lambda))$.
If none of the corresponding eigenfunctions vanishes on $V_\parallel$, then the first two terms on the RHS of \eqref{eq:eigp6} vanish, and $\mult(M,\lambda)=\mult(L,R(\lambda))$.

\item \label{eq:SD7} If $\lambda\notin\sigma(D)$, $\phi(\lambda)=0$, $\displaystyle \lim_{\mathbb{R}\ni x\to\lambda} [R(x)]^{-1}=0$, and moreover 
$x\mapsto (\lambda-x)[\phi(x)]^{-1}$ is bounded in a neighborhood of $\lambda$,
 then $E_\lambda(M)=0$, \emph{i.e.,} $\mult(M,\lambda)=0$.
\end{enumerate}
\end{lemma}

\begin{proof}
Our starting point is the combination of \eqref{eq:M-x-1} and \eqref{eq:specsim}.
Let us note right away that
\begin{align*}
\lambda \notin \sigma(D) &\quad\text{implies} \quad \lim_{x\to \lambda} P_\perp^*(\lambda-x)(D-x)^{-1} P_\perp =0,\\
\lambda \in \sigma(D) &\quad\text{implies} \quad \lim_{x\to\lambda} P_\perp^* (\lambda-x)(D-x)^{-1} P_\perp = P_\perp^* E_\lambda(D) P_\perp.
\end{align*}
So this reduces our analysis to the second term on the RHS of \eqref{eq:P0}, namely:
\begin{align}
\label{eq:secondterm}
\left(P_\parallel^* - P_\perp^* (D-x)^{-1} C\right) (\lambda-x)[\phi(x)]^{-1}(L-R(x))^{-1} (I_\parallel -E_{R(\lambda)}(L)) \left(P_\parallel - B(D-x)^{-1}P_\perp\right).
\end{align}
As in the proof of Lemma \ref{specdecthm}, terms which stay bounded (resp.\@ vanish) as $\mathbb{R}\ni x\to\lambda$ are highlighted in {\color{blue} blue} (resp.\@ {\color{purple} purple}).  

\ref{eq:SD2}:
By the assumptions, \eqref{eq:secondterm} reads
\begin{align*}
&{\color{blue} \left(P_{\parallel}^*-P_{\perp}^*(D-x)^{-1}C\right)}(\lambda-x)[\phi(x)]^{-1}(L-R(x))^{-1}{\color{blue}\left(P_{\parallel}-B(D-x)^{-1}P_{\perp}\right)}\\
&={\color{blue} \left(P_{\parallel}^*-P_{\perp}^*(D-x)^{-1}C\right)}\frac{(\lambda-x)[\phi(x)]^{-1}}{R(\lambda)-R(x)}(R(\lambda)-R(x))(L-R(x))^{-1}{\color{blue} E_{R(\lambda)}(L)}{\color{blue}\left(P_{\parallel}-B(D-x)^{-1}P_{\perp}\right)}\\
&+{\color{blue} \left(P_{\parallel}^*-P_{\perp}^*(D-x)^{-1}C\right)}(\lambda-x)[\phi(x)]^{-1}{\color{blue}(L-R(x))^{-1}\left(I_\parallel-E_{R(\lambda)}(L)\right)\left(P_{\parallel}-B(D-x)^{-1}P_{\perp}\right)}.
\end{align*}

\ref{eq:SD3}:
By the assumptions, \eqref{eq:secondterm} reads
\begin{align*}
&{\color{blue}\left(P_{\parallel}^*-P_{\perp}^*(D-x)^{-1}C\right)}\frac{(\lambda-x){\color{purple}[\phi(x)]^{-1}}}{R(\lambda)-R(x)}(R(\lambda)-R(x))(L-R(x))^{-1}{\color{blue} E_{R(\lambda)}(L)}
{\color{blue}\left(P_{\parallel}-B(D-x)^{-1}P_{\perp}\right)}\\
&+\left({\color{blue} P_{\parallel}^*}-P_{\perp}^*(D-x)^{-1}C\right)(\lambda-x){\color{purple}[\phi(x)]^{-1}}{\color{blue}(L-R(x))^{-1}\left(I_\parallel-E_{R(\lambda)}(L)\right)}\left({\color{blue}P_{\parallel}}-B(D-x)^{-1}P_{\perp}\right).
\end{align*}

For the first term, the boundedness of $(D-\lambda)^{-1}$ follows from Lemma \ref{lem:mod}-\eqref{phineq0}. We further note that $\lim_{x\to\lambda} (R(\lambda)-R(x))(L-R(x))^{-1} = E_{R(\lambda)}(L)$ by the functional calculus, and $[\phi(\lambda)]^{-1}=0$ by assumption.
In fact, we would like to show that $\lim_{x\to\lambda} \frac{\lambda-x}{R(\lambda)-R(x)} [\phi(x)]^{-1}=0$, and it suffices to have $\left|\frac{\lambda-x}{R(\lambda)-R(x)}\right|$ to be bounded in a neighborhood of $\lambda$. Consequently the first term vanishes in the limit.

The second term requires more care, as we do not know \emph{a priori} that $D-\lambda$ is invertible.
So we expand it as the sum of four terms
\begin{align*}
&{\color{blue}P_{\parallel}^*}{\color{purple}(\lambda-x)[\phi(x)]^{-1}}{\color{blue}(L-R(x))^{-1}
(I_{\parallel}-E_{R(\lambda)}(L))P_{\parallel}}\\
&-{\color{blue}P_{\perp}^*(\lambda-x)(D-x)^{-1}C}{\color{purple}[\phi(x)]^{-1}}{\color{blue}(L-R(x))^{-1}(I_{\parallel}-E_{R(\lambda)}(L))P_{\parallel}}\\
&-{\color{blue}P_{\parallel}^*}{\color{purple}[\phi(x)]^{-1}}{\color{blue}(L-R(x))^{-1}(I_{\parallel}-E_{R(\lambda)}(L))B(\lambda-x)(D-x)^{-1} P_{\perp}}\\
&+{\color{blue}P_{\perp}^*(\lambda-x)(D-x)^{-1} C}[\phi(x)(\lambda-x)]^{-1}{\color{blue}(L-R(x))^{-1}(I_{\parallel}-E_{R(\lambda)}(L))B(\lambda-x)(D-x)^{-1}P_{\perp}}.
\end{align*}
It can be seen readily that the first three lines vanish in the limit, whereas the fourth line converges to
\[
P_\perp^* E_\lambda(D) C \left(\lim_{\mathbb{R}\ni x\to\lambda} \frac{1}{\phi(x)(\lambda-x)}\right) (L-R(\lambda))^{-1} (I_\parallel-E_{R(\lambda)}(L)) B E_\lambda(D) P_\perp
\]
given the assumptions. The eigenprojector formula \eqref{eq:eigp3} follows.

Observe that the image of $E_\lambda(M)$ is contained in the image of $P_\perp^* E_\lambda(D) P_\perp$.
More specifically,
\[
{\rm rank}(P_{\perp}^*E_{\lambda}(D)P_{\perp})-{\rm rank}(E_\lambda(M))={\rm rank}(I_{\parallel}-E_{R(\lambda)}(L)),
\]
from which the multiplicity formula follows.

\ref{eq:SD4}:
By the assumptions and Lemma \ref{lem:mod}-\eqref{phineq0}, \eqref{eq:secondterm} reads
\begin{align*}
&{\color{blue}\left(P_{\parallel}^*-P_{\perp}^*(D-x)^{-1}C\right)}\frac{(\lambda-x){\color{blue}[\phi(x)]^{-1}}}{R(\lambda)-R(x)}(R(\lambda)-R(x))(L-R(x))^{-1}{\color{blue} E_{R(\lambda)}(L)}
{\color{blue}\left(P_{\parallel}-B(D-x)^{-1}P_{\perp}\right)}\\
&+\left({\color{blue} P_{\parallel}^*}-P_{\perp}^*(D-x)^{-1}C\right)(\lambda-x){\color{blue}[\phi(x)]^{-1}}{\color{blue}(L-R(x))^{-1}\left(I_\parallel-E_{R(\lambda)}(L)\right)}\left({\color{blue}P_{\parallel}}-B(D-x)^{-1}P_{\perp}\right).
\end{align*}

The first term tends to
\begin{align*}
\left(P_{\parallel}^*-P_{\perp}^*(D-\lambda)^{-1}C\right) \left(\lim_{\mathbb{R} \ni x\to\lambda}\frac{1}{\phi(x)}\frac{\lambda-x}{R(\lambda)-R(x)}\right)
 E_{R(\lambda)}(L)
\left(P_{\parallel}-B(D-\lambda)^{-1}P_{\perp}\right).
\end{align*}
The second term is again trickier, being the sum of
\begin{align*}
&{\color{blue}P_{\parallel}^*}{\color{purple}(\lambda-x)}{\color{blue}[\phi(x)]^{-1}}{\color{blue}(L-R(x))^{-1}
(I_{\parallel}-E_{R(\lambda)}(L))P_{\parallel}}\\
&-{\color{blue}P_{\perp}^*}{\color{purple}(\lambda-x)(D-x)^{-1}C}{\color{blue}[\phi(x)]^{-1}}{\color{blue}(L-R(x))^{-1}(I_{\parallel}-E_{R(\lambda)}(L))P_{\parallel}}\\
&-{\color{blue}P_{\parallel}^*}{\color{blue}[\phi(x)]^{-1}}{\color{blue}(L-R(x))^{-1}(I_{\parallel}-E_{R(\lambda)}(L))}{\color{purple}B(\lambda-x)(D-x)^{-1}}{\color{blue} P_{\perp}}\\
&+{\color{blue}P_{\perp}^*}(D-x)^{-1} C {\color{blue}[\phi(x)]^{-1}(L-R(x))^{-1}(I_{\parallel}-E_{R(\lambda)}(L))}{\color{purple}B(\lambda-x)(D-x)^{-1}}{\color{blue}P_{\perp}},
\end{align*}
where the vanishing {\color{purple} purple} terms in the last 3 lines are due to Lemma \ref{lem:mod}-\eqref{ED0}.
Altogether the entire sum vanishes in the limit.
This proves \eqref{eq:eigp4}.
Observe that the two terms on the RHS of \eqref{eq:eigp4} are mutually orthogonal, from which the remaining claims follow.

\ref{eq:SD6}:
This is a straightforward extension of \ref{eq:SD3}.
In particular, if none of the corresponding eigenfunctions vanishes on $V_\parallel$, then by \ref{eq:SD3}, the first two terms on the RHS of \eqref{eq:eigp6} vanishes.

\ref{eq:SD7}:
Since the spectrum of an operator is compact, $(L-R(x))^{-1}$ remains bounded---in fact tends to $0$---as $R(x)\to R(\lambda)=\infty$.
Thus \eqref{eq:secondterm} reads
\begin{align*}
{\color{blue}(P_{\parallel}^*-P_{\perp}^*(D-x)^{-1}C)}
{\color{blue}(\lambda-x)[\phi(x)]^{-1}}{\color{purple}  (L-R(x))^{-1}}
{\color{blue}(P_{\parallel}-B(D-x)^{-1}P_{\perp})}
\end{align*}
which vanishes in the limit.
\end{proof}


\section{Spectral self-similarity of the magnetic Laplacian} \label{sec:ssSG}

\subsection{Schur complement computation}

Let $G_N= (V_N, E_N)$ be the level-$N$ Sierpinski gasket graph.
Following \eqref{eq:SGML}, the magnetic Laplacian $\mathcal{L}^\omega_N$ on $G_N$ endowed with $U(1)$ connection $\omega$ is an operator on $\ell^2(V_N)$, and can be represented in the standard basis by the $|V_N|$-by-$|V_N|$ matrix
\begin{align}
\label{eq:LN}
\mathcal{L}^{\omega}_N(x,y) =
\left\{
\begin{array}{ll}
1, & \text{if } x=y,\\
-\frac{1}{2} \omega_{xy}, & \text{if } x\in V_0,~y\sim x,\\
-\frac{1}{4} \omega_{xy}, & \text{if } x\in V_N\setminus V_0, ~y\sim x,\\
0, & \text{else}.
\end{array}
\right.
\end{align}
Recall that $\mathcal{L}^\omega_N$ is self-adjoint on $L^2(V_N, \deg_{G_N})$.

We express the resolvent in block matrix form
\begin{align}
\label{eq:magLap}
\mathcal{L}^{\omega}_N - \lambda I =
\begin{bmatrix}
A-\lambda I & B \\
C & D-\lambda I
\end{bmatrix},
\quad
\lambda\in \mathbb{C},
\end{align}
where the rows and columns are arranged such that
\[
\begin{aligned}
A: \ell^2(V_{N-1}) \to \ell^2(V_{N-1}),& \quad
B: \ell^2(V_N\setminus V_{N-1}) \to \ell^2(V_{N-1}),\\
C: \ell^2(V_{N-1}) \to \ell^2(V_N\setminus V_{N-1}),& \quad
D: \ell^2(V_N\setminus V_{N-1}) \to \ell^2(V_N\setminus V_{N-1}),
\end{aligned}
\]
where $I$ is the identity matrix of an appropriate size, and $\ell^2(S) =\mathbb{C}^S$.

Assuming that $D-\lambda I$ is invertible for the moment, we define the \textbf{Schur complement} of $\mathcal{L}_N^{\omega} -\lambda I$ with respect to the minor $D-\lambda I$ as
\begin{align}
\label{eq:Schur}
S_N^{\omega}(\lambda):= (A-\lambda I) - B(D-\lambda I)^{-1} C,
\end{align}
which acts on $\ell^2(V_{N-1})$.
To find the entries of $S_N^{\omega}(\lambda)$, we label the vertices in $V_{N-1}$ by $a_i$, and vertices in $V_N\setminus V_{N-1}$ by $b_i$.
Then for $a_i, a_j\in V_{N-1}$ we have
\begin{align}
\label{eq:Schur1}
S_N^{\omega}(\lambda)(a_i, a_j) = (A-\lambda I)(a_i, a_j) - \sum_{b_k, b_l \in V_N\setminus V_{N-1}} B(a_i, b_k) (D-\lambda I)^{-1} (b_k, b_l) C(b_l, a_j).
\end{align}

Recall \eqref{eq:LN}. Observe that $(A-\lambda I)(a_i, a_j)= (1-\lambda) \delta_{a_i a_j}$; $B(a_i, b_k) = -\frac{1}{2}\omega_{a_i b_k}$ if  $a_i\in V_0$ and $a_i \sim b_k$, $-\frac{1}{4}\omega_{a_i b_k}$ if $a_i \in V_N\setminus V_0$ and $a_i\sim b_k$, and $0$ otherwise; $C(b_l, a_j)= -\frac{1}{4}\omega_{b_l a_j}$ if $b_l\sim a_j$, and $0$ otherwise; and $(D-\lambda I)^{-1}$ is zero whenever $b_k \not\sim b_l$.
By the nested structure of $SG$, $(D-\lambda I)^{-1}$ is a block diagonal matrix consisting of 3-by-3 Hermitian matrices,
each of which is supported on the inner vertices of a level-$(N-1)$ cell, and has the same structure.
To be concrete, we denote the cell by $\Lambda$, and its three inner vertices by $b_0, b_1, b_2$. Then
\begin{align}
\label{eq:defofD}
\left.(D-\lambda I)\right|_{\Lambda}(b_i, b_j) = 
\left\{
\begin{array}{ll}
1-\lambda, &\text{if } b_i=b_j,\\
-\frac{1}{4} \omega_{b_i b_j}, &\text{if } b_i\neq b_j.
\end{array}
\right.
\end{align}
Using Cramer's formula for the matrix inverse, we get
\begin{align}
\left.(D-\lambda I)\right|_{\Lambda}^{-1}(b_i, b_j)= \frac{1}{\det(\left.( D-\lambda I)\right|_{\Lambda})}{\rm adj}(\left.(D-\lambda I)\right|_{\Lambda})
\end{align}
where
\begin{align}
\det(\left.(D-\lambda I) \right|_\Lambda) &= (1-\lambda)^3 - \frac{3}{16}(1-\lambda) - \frac{1}{32} {\rm Re}(\omega_{b_0 b_1} \omega_{b_1 b_2} \omega_{b_2 b_0}),\\
\label{eq:adjd}{\rm adj}(\left. (D-\lambda I)\right|_\Lambda)(b_i, b_i) &= (1-\lambda)^2-\frac{1}{16},\quad i\in \{0,1,2\},\\ 
\label{eq:adjod}{\rm adj}(\left. (D-\lambda I)\right|_\Lambda)(b_i, b_j) &= \frac{1}{4}(1-\lambda)\omega_{b_i b_j} + \frac{1}{16} \omega_{b_i b_k} \omega_{b_k b_j}, \text{ ~if } i\neq j,
\end{align}
and $k=k(i,j)$ is the third index in $\{0,1,2\} \setminus \{i,j\}$.

In light of the difference between the diagonal and off-diagonal entries of the adjugate matrix, \eqref{eq:adjd} and \eqref{eq:adjod}, we shall rewrite the second term on the RHS of \eqref{eq:Schur1} by splitting the case $b_l=b_k$ and the case $b_l\neq b_k$; namely, if $a_i \in V_N\setminus V_0$, we have
\begin{equation}
\label{eq:Schur2}
\begin{aligned}
\sum_{b_k, b_l}& B(a_i, b_k) (D-\lambda I)^{-1} (b_k, b_l) C(b_l, a_j)
\\
 &=\frac{1}{16}\sum_{b_k \sim \{a_i,a_j\}} \frac{1}{\det(\left.(D-\lambda I)\right|_{\Lambda(b_k)})} \left((1-\lambda)^2-\frac{1}{16}\right)\omega_{a_i b_k}  \omega_{b_k  a_j}\\
&\quad+\frac{1}{16}\sum_{b_k \sim a_i} \sum_{\substack{b_l \sim a_j\\ b_l\neq b_k}}  \frac{1}{\det(\left.(D-\lambda I)\right|_{\Lambda(b_k,b_l)})}\omega_{a_i b_k} \left(\frac{1}{4}(1-\lambda) \omega_{b_k b_l} + \frac{1}{16} \omega_{b_k b_m} \omega_{b_m b_l} \right)\omega_{b_l  a_j},
\end{aligned}
\end{equation}
where $\Lambda(b_1, b_2,\cdots)$ denotes the level-$(N-1)$ cell which contains the vertices $b_1, b_2, \cdots$, and $\{b_k, b_l, b_m\} \subset V_N\setminus V_{N-1}$ form the 3 inner vertices of $\Lambda(b_k, b_l)$.
If $a_i \in V_0$, replace the prefactor $\frac{1}{16}$ in the formula \eqref{eq:Schur2} by $\frac{1}{8}$.

If \underline{$a_i=a_j \in V_N\setminus V_0$:} 
We have
\begin{equation}
\label{eq:Schur3}
\begin{aligned}
\sum_{b_k, b_l}& B(a_i, b_k) (D-\lambda I)^{-1} (b_k, b_l) C(b_l, a_i)
\\ 
&=\frac{1}{16}\sum_{b_k \sim a_i} \frac{1}{\det(\left.(D-\lambda I)\right|_{\Lambda(b_k)})} \left((1-\lambda)^2-\frac{1}{16}\right)\\
&\quad+\frac{1}{16}\sum_{b_k \sim a_i} \sum_{\substack{b_l \sim a_i\\ b_l\neq b_k}}  \frac{1}{\det(\left.(D-\lambda I)\right|_{\Lambda(b_k,b_l)})}\omega_{a_i b_k} \left(\frac{1}{4}(1-\lambda) \omega_{b_k b_l} + \frac{1}{16} \omega_{b_k b_m} \omega_{b_m b_l} \right)\omega_{b_l  a_i}.
\end{aligned}
\end{equation}
Observe that if $a_i$ is contained in two level-$(N-1)$ cells.
We need to pick $\{b_k, b_l\} \sim a_i$ from the same cell to produce a nonzero summand in the second sum.

If \underline{$a_i =a_j \in V_0$:} 
The formula \eqref{eq:Schur3} holds with the prefactor $\frac{1}{16}$ replaced by $\frac{1}{8}$.
Also, $a_i$ is contained in a unique level-$(N-1)$ cell.

If \underline{$a_i\neq a_j$:}
In \eqref{eq:Schur2} note that $a_i, a_j, b_k, b_l$ must belong to the same level-$(N-1)$ cell to produce a nonzero summand. Therefore once we fix $a_i$ and $a_j$, both sums are localized to the cell $\Lambda(a_i, a_j)$.

\subsection{Diagrammatic analysis} \label{sec:diagramanalysis}

To make the results \eqref{eq:Schur2} and \eqref{eq:Schur3} more transparent, we introduce a diagrammatic bookkeeping device.
Given a path $P=\{x_0, x_1,\cdots, x_m\}$, we represent the product of the parallel transports along $P$, $\omega_{x_0 x_1} \omega_{x_1 x_2}\cdots \omega_{x_{m-1} x_m}=:\omega(P)$, by the diagram
\begin{center}
\begin{tikzpicture}[scale=0.3]
\begin{scope}[thick, every node/.style={sloped,allow upside down}]
\draw [thick] (0,0) -- node {\midarrow} (3,0) -- node {\midarrow} (4.3,2.6)-- node {\midarrow} (3,5.2);
\draw [thick, dotted] (3,5.2) -- (1.5,5.2);
\draw [thick] (1.5,5.2) -- node {\midarrow} (-0.2,5.2);
\end{scope}
\tkzDefPoint(0,0){x_0}
\tkzDefPoint(3,0){x_1}
\tkzDefPoint(4.3,2.6){x_2}
\tkzDefPoint(3,5.2){x_3}
\tkzDefPoint(-0.2,5.2){x_m}
\tkzLabelPoint[below](x_0){$x_0$}
\tkzLabelPoint[below](x_1){$x_1$}
\tkzLabelPoint[right](x_2){$x_2$}
\tkzLabelPoint[above](x_3){$x_3$}
\tkzLabelPoint[left](x_m){$x_m$}
\foreach \n in {x_0,x_1,x_2,x_3,x_m}
  \node at (\n)[circle,fill,inner sep=1.5pt]{};
\end{tikzpicture}
\end{center}

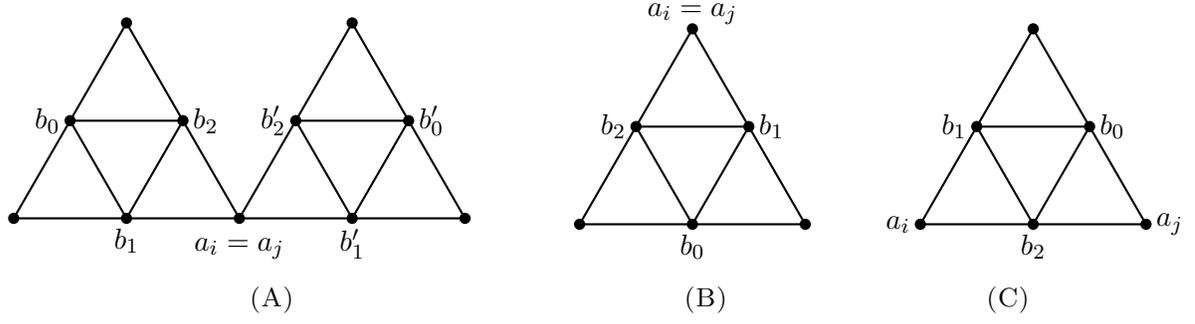
\begin{figure}
\centering
\begin{subfigure}[b]{0.4\textwidth}
\begin{tikzpicture}
\draw [thick] (0,0) -- (60:3) -- (3,0) -- cycle;
\draw [thick, rotate around={-60:(1.5,0)}] (1.5,0) -- (60:1.5) -- (0,0) -- cycle;
\tkzDefPoint(1.5,0){b_1}
\tkzDefPoint(0.75,1.3){b_0}
\tkzDefPoint(2.25,1.3){b_2}
\tkzDefPoint(3,0){a}
\tkzDefPoint(0,0){p1}
\tkzDefPoint(1.5,2.6){p2}
\tkzLabelPoint[below](b_1){$b_1$}
\tkzLabelPoint[left](b_0){$b_0$}
\tkzLabelPoint[right](b_2){$b_2$}
\foreach \n in {b_1,b_2,b_0,a,p1,p2}
  \node at (\n)[circle,fill,inner sep=1.5pt]{};
\draw (3,-0.4) node {$a_i=a_j$};

\begin{scope}[shift={(3,0)}]
\draw [thick] (0,0) -- (60:3) -- (3,0) -- cycle;
\draw [thick, rotate around={-60:(1.5,0)}] (1.5,0) -- (60:1.5) -- (0,0) -- cycle;
\tkzDefPoint(1.5,0){b_1'}
\tkzDefPoint(0.75,1.3){b_2'}
\tkzDefPoint(2.25,1.3){b_0'}
\tkzDefPoint(3,0){p3}
\tkzDefPoint(1.5,2.6){p4}
\tkzLabelPoint[below](b_1'){$b_1'$}
\tkzLabelPoint[right](b_0'){$b_0'$}
\tkzLabelPoint[left](b_2'){$b_2'$}
\foreach \n in {b_1',b_2',b_0',p4,p3}
  \node at (\n)[circle,fill,inner sep=1.5pt]{};
\end{scope}
\end{tikzpicture}
\caption{}
\label{subfig1}
\end{subfigure}
\quad
\begin{subfigure}[b]{0.2\textwidth}
\begin{tikzpicture}
\draw [thick] (0,0) -- (60:3) -- (3,0) -- cycle;
\draw [thick, rotate around={-60:(1.5,0)}] (1.5,0) -- (60:1.5) -- (0,0) -- cycle;
\tkzDefPoint(0,0){p1}
\tkzDefPoint(1.5,0){b_0}
\tkzDefPoint(3,0){p2}
\tkzDefPoint(0.75,1.3){b_2}
\tkzDefPoint(2.25,1.3){b_1}
\tkzDefPoint(1.5,2.6){p3}
\tkzLabelPoint[below](b_0){$b_0$}
\tkzLabelPoint[left](b_2){$b_2$}
\tkzLabelPoint[right](b_1){$b_1$}
\foreach \n in {b_0,b_1,b_2,p1,p2,p3} \node at (\n)[circle,fill,inner sep=1.5pt]{};
\draw (1.5,2.8) node {$a_i=a_j$};
\end{tikzpicture}
\caption{}
\label{subfig2}
\end{subfigure}
\quad
\begin{subfigure}[b]{0.2\textwidth}
\begin{tikzpicture}
\label{fig:sd}
\draw [thick] (0,0) -- (60:3) -- (3,0) -- cycle;
\draw [thick, rotate around={-60:(1.5,0)}] (1.5,0) -- (60:1.5) -- (0,0) -- cycle;
\tkzDefPoint(0,0){a_i}
\tkzDefPoint(1.5,0){b_2}
\tkzDefPoint(3,0){a_j}
\tkzDefPoint(0.75,1.3){b_1}
\tkzDefPoint(2.25,1.3){b_0}
\tkzDefPoint(1.5,2.6){p}
\tkzLabelPoint[right](b_0){$b_0$}
\tkzLabelPoint[below](b_2){$b_2$}
\tkzLabelPoint[left](b_1){$b_1$}
\tkzLabelPoint[left](a_i){$a_i$}
\tkzLabelPoint[right](a_j){$a_j$}
\foreach \n in {b_0,b_1,b_2,a_i,a_j,p} \node at (\n)[circle,fill,inner sep=1.5pt]{};
\end{tikzpicture}
\caption{}
\label{subfig3}
\end{subfigure}
\caption{The unit cells for diagrammatic analysis used in \S\ref{sec:diagramanalysis}.}
\label{fig:startingposition}
\end{figure}

If \underline{$a_i=a_j\in V_N \setminus V_0$}: Consider \eqref{eq:Schur3} and the diagram in Figure \ref{fig:startingposition}(\subref{subfig1}).
We find that there are 4 identical terms in the first summand because $\deg(a_i)=4$, and there are 8 terms in the second summand. A diagrammatic representation of \eqref{eq:Schur3} becomes 
\begin{equation}
\label{eq:Schurdiag1}
\begin{aligned}
\sum_{b_k, b_l}& B(a_i, b_k) (D-\lambda I)^{-1} (b_k, b_l) C(b_l, a_i)
\\
& =\frac{1}{\mathcal{D}(\beta,\lambda)}\cdot\frac{4}{16}\left((1-\lambda)^2-\frac{1}{16}\right)
\\
& \quad +\frac{1}{16}\cdot\frac{1}{\mathcal{D}(\beta,\lambda)}\left(\frac{1}{4}(1-\lambda)\cdot
\begin{tikzpicture}[thick, every node/.style={sloped, align=center, allow upside down}, baseline=($(b_1.base)!.5!(b_2.base)$)]
\draw (0,0) -- pic[pos=0.6]{arrow=latex} (1.5,0) -- pic[pos=0.6]{arrow=latex} (0.75,1.3) -- pic[pos=0.6]{arrow=latex} (0,0);
\tkzDefPoint(0,0){b_1}
\tkzDefPoint(1.5,0){a_i}
\tkzDefPoint(0.75,1.3){b_2}
\tkzLabelPoint[left](b_1){$b_1$}
\tkzLabelPoint[right](a_i){$a_i$}
\tkzLabelPoint[above](b_2){$b_2$}
\foreach \n in {b_1,b_2,a_i}
  \node at (\n)[circle,fill,inner sep=1.5pt]{};
\end{tikzpicture}
+\frac{1}{16}\cdot
\begin{tikzpicture}[thick, every node/.style={sloped, align=center, allow upside down}, baseline=($(b_1.base)!.5!(b_2.base)$)]
\draw (1.5,0) -- pic[pos=0.6]{arrow=latex} (0.75,1.3) -- pic[pos=0.6]{arrow=latex} (-0.75,1.3) -- pic[pos=0.6]{arrow=latex} (0,0) -- pic[pos=0.6]{arrow=latex} (1.5,0);
\tkzDefPoint(0,0){b_1}
\tkzDefPoint(1.5,0){a_i}
\tkzDefPoint(0.75,1.3){b_2}
\tkzDefPoint(-0.75,1.3){b_0}
\tkzLabelPoint[left](b_1){$b_1$}
\tkzLabelPoint[right](a_i){$a_i$}
\tkzLabelPoint[right](b_2){$b_2$}
\tkzLabelPoint[left](b_0){$b_0$}
\foreach \n in {b_1,b_2,a_i,b_0}
  \node at (\n)[circle,fill,inner sep=1.5pt]{};
\end{tikzpicture}
\right)
\\
& \quad +\frac{1}{16}\cdot\frac{1}{\mathcal{D}(\beta,\lambda)}\left(\frac{1}{4}(1-\lambda)\cdot
\begin{tikzpicture}[thick, every node/.style={sloped, align=center, allow upside down}, baseline=($(b_1.base)!.5!(b_2.base)$)]
\draw (0,0) -- pic[pos=0.6]{arrow=latex} (-1.5,0) -- pic[pos=0.6]{arrow=latex} (-0.75,1.3) -- pic[pos=0.6]{arrow=latex} (0,0);
\tkzDefPoint(-1.5,0){b_1}
\tkzDefPoint(0,0){a_i}
\tkzDefPoint(-0.75,1.3){b_2}
\tkzLabelPoint[left](b_1){$b_1$}
\tkzLabelPoint[right](a_i){$a_i$}
\tkzLabelPoint[above](b_2){$b_2$}
\foreach \n in {b_1,b_2,a_i}
  \node at (\n)[circle,fill,inner sep=1.5pt]{};
\end{tikzpicture}
+\frac{1}{16}\cdot
\begin{tikzpicture}[thick, every node/.style={sloped, align=center, allow upside down}, baseline=($(b_1.base)!.5!(b_2.base)$)]
\draw (1.5,0) -- pic[pos=0.6]{arrow=latex} (0,0) -- pic[pos=0.6]{arrow=latex} (-0.75,1.3) -- pic[pos=0.6]{arrow=latex} (0.75,1.3) -- pic[pos=0.6]{arrow=latex} (1.5,0);
\tkzDefPoint(0,0){b_1}
\tkzDefPoint(1.5,0){a_i}
\tkzDefPoint(0.75,1.3){b_2}
\tkzDefPoint(-0.75,1.3){b_0}
\tkzLabelPoint[left](b_1){$b_1$}
\tkzLabelPoint[right](a_i){$a_i$}
\tkzLabelPoint[right](b_2){$b_2$}
\tkzLabelPoint[left](b_0){$b_0$}
\foreach \n in {b_1,b_2,a_i,b_0}
  \node at (\n)[circle,fill,inner sep=1.5pt]{};
\end{tikzpicture}
\right)
\\
& \quad + \frac{1}{16}\cdot\frac{1}{\mathcal{D}(\beta,\lambda)}\left(\frac{1}{4}(1-\lambda)\cdot
\begin{tikzpicture}
[thick, every node/.style={sloped, align=center, allow upside down}, baseline=($(b_1'.base)!.5!(b_2'.base)$)]
\draw (0,0) -- pic[pos=0.6]{arrow=latex} (1.5,0) -- pic[pos=0.6]{arrow=latex} (0.75,1.3) -- pic[pos=0.6]{arrow=latex} (0,0);
\tkzDefPoint(1.5,0){b_1'}
\tkzDefPoint(0,0){a_i}
\tkzDefPoint(0.75,1.3){b_2'}
\tkzLabelPoint[right](b_1'){$b_1'$}
\tkzLabelPoint[left](a_i){$a_i$}
\tkzLabelPoint[above](b_2'){$b_2'$}
\foreach \n in {b_1',b_2',a_i}
  \node at (\n)[circle,fill,inner sep=1.5pt]{};
\end{tikzpicture}
+\frac{1}{16}\cdot
\begin{tikzpicture}[thick, every node/.style={sloped, align=center, allow upside down}, baseline=($(b_1'.base)!.5!(b_2'.base)$)]
\draw (0,0) -- pic[pos=0.6]{arrow=latex} (1.5,0) -- pic[pos=0.6]{arrow=latex} (2.25,1.3) -- pic[pos=0.6]{arrow=latex} (0.75,1.3) -- pic[pos=0.6]{arrow=latex} (0,0);
\tkzDefPoint(0,0){a_i}
\tkzDefPoint(1.5,0){b_1'}
\tkzDefPoint(0.75,1.3){b_2'}
\tkzDefPoint(2.25,1.3){b_0'}
\tkzLabelPoint[right](b_1'){$b_1'$}
\tkzLabelPoint[left](a_i){$a_i$}
\tkzLabelPoint[left](b_2'){$b_2'$}
\tkzLabelPoint[right](b_0'){$b_0'$}
\foreach \n in {b_1',b_2',a_i,b_0'}
  \node at (\n)[circle,fill,inner sep=1.5pt]{};
\end{tikzpicture}
\right)
\\
& \quad + \frac{1}{16}\cdot\frac{1}{\mathcal{D}(\beta,\lambda)}\left(\frac{1}{4}(1-\lambda)\cdot
\begin{tikzpicture}
[thick, every node/.style={sloped, align=center, allow upside down}, baseline=($(b_1'.base)!.5!(b_2'.base)$)]
\draw (0,0) -- pic[pos=0.6]{arrow=latex} (0.75,1.3) -- pic[pos=0.6]{arrow=latex} (1.5,0) -- pic[pos=0.6]{arrow=latex} (0,0);
\tkzDefPoint(1.5,0){b_1'}
\tkzDefPoint(0,0){a_i}
\tkzDefPoint(0.75,1.3){b_2'}
\tkzLabelPoint[right](b_1'){$b_1'$}
\tkzLabelPoint[left](a_i){$a_i$}
\tkzLabelPoint[above](b_2'){$b_2'$}
\foreach \n in {b_1',b_2',a_i}
  \node at (\n)[circle,fill,inner sep=1.5pt]{};
\end{tikzpicture}
+\frac{1}{16}\cdot
\begin{tikzpicture}[thick, every node/.style={sloped, align=center, allow upside down}, baseline=($(b_1'.base)!.5!(b_2'.base)$)]
\draw (0,0) -- pic[pos=0.6]{arrow=latex} (0.75,1.3) -- pic[pos=0.6]{arrow=latex} (2.25,1.3) -- pic[pos=0.6]{arrow=latex} (1.5,0) -- pic[pos=0.6]{arrow=latex} (0,0);
\tkzDefPoint(0,0){a_i}
\tkzDefPoint(1.5,0){b_1'}
\tkzDefPoint(0.75,1.3){b_2'}
\tkzDefPoint(2.25,1.3){b_0'}
\tkzLabelPoint[right](b_1'){$b_1'$}
\tkzLabelPoint[left](a_i){$a_i$}
\tkzLabelPoint[left](b_2'){$b_2'$}
\tkzLabelPoint[right](b_0'){$b_0'$}
\foreach \n in {b_1',b_2',a_i,b_0'}
  \node at (\n)[circle,fill,inner sep=1.5pt]{};
\end{tikzpicture}
\right)
\end{aligned}
\end{equation}
where
\begin{equation}
\label{eq:detofD}
\begin{aligned}
\mathcal{D}(\beta,\lambda)&=\det\left((D-\lambda I)|_{\Lambda(b_0^{(')},b_1^{(')},b_2^{(')})}\right)
\\
& =(1-\lambda)^3-\frac{3}{16}(1-\lambda)-\frac{1}{64}\left(
\begin{tikzpicture}[thick, every node/.style={sloped, align=center, allow upside down}, baseline=($(b_1^{(')}.base)!.5!(b_2^{(')}.base)$)]
\draw (-0.75,0.65) -- pic[pos=0.6]{arrow=latex} (0.75,0.65) -- pic[pos=0.6]{arrow=latex} (0,-0.65) -- pic[pos=0.6]{arrow=latex} (-0.75,0.65);
\tkzDefPoint(0,-0.65){b_1^{(')}}
\tkzDefPoint(-0.75,0.65){b_2^{(')}}
\tkzDefPoint(0.75,0.65){b_0^{(')}}
\tkzLabelPoint[below]({b_1^{(')}}){${b_1^{(')}}$}
\tkzLabelPoint[left]({b_2^{(')}}){${b_2^{(')}}$}
\tkzLabelPoint[right]({b_0^{(')}}){${b_0^{(')}}$}
\foreach \n in {b_0^{(')},b_1^{(')},b_2^{(')}}
  \node at (\n)[circle,fill,inner sep=1.5pt]{};
\end{tikzpicture}
+\begin{tikzpicture}[thick, every node/.style={sloped, align=center, allow upside down}, baseline=($(b_1^{(')}.base)!.5!(b_2^{(')}.base)$)]
\draw (-0.75,0.65) -- pic[pos=0.6]{arrow=latex} (0,-0.65) -- pic[pos=0.6]{arrow=latex} (0.75,0.65) -- pic[pos=0.6]{arrow=latex} cycle;
\tkzDefPoint(0,-0.65){b_1^{(')}}
\tkzDefPoint(-0.75,0.65){b_2^{(')}}
\tkzDefPoint(0.75,0.65){b_0^{(')}}
\tkzLabelPoint[below]({b_1^{(')}}){${b_1^{(')}}$}
\tkzLabelPoint[left]({b_2^{(')}}){${b_2^{(')}}$}
\tkzLabelPoint[right]({b_0^{(')}}){${b_0^{(')}}$}
\foreach \n in {b_0^{(')},b_1^{(')},b_2^{(')}}
  \node at (\n)[circle,fill,inner sep=1.5pt]{};
\end{tikzpicture}
\right).
\end{aligned}
\end{equation}

If \underline{$a_i=a_j\in V_0$}: See Figure \ref{fig:startingposition}(\subref{subfig2}).
Formula \eqref{eq:Schur3} becomes
\begin{equation}
\label{eq:Schurdiag2}
\begin{aligned}
\sum_{b_k, b_l}& B(a_i, b_k) (D-\lambda I)^{-1} (b_k, b_l) C(b_l, a_i)
\\
& =\frac{2}{8}\cdot\frac{1}{\mathcal{D}(\beta,\lambda)} \left((1-\lambda)^2-\frac{1}{16}\right)
\\
& +\frac{1}{8}\cdot\frac{1}{\mathcal{D}(\beta,\lambda)}\cdot \left(\frac{1}{4}\cdot(1-\lambda)\cdot
\begin{tikzpicture}[scale=0.7,thick, every node/.style={sloped, allow upside down}, baseline=($(b_1.base)!.5!(a_i.base)$)]
\draw (0.75,1.3) -- pic[pos=0.6]{arrow=latex} (0,0) -- pic[pos=0.6]{arrow=latex} (1.5,0) -- pic[pos=0.6]{arrow=latex} (0.75,1.3);
\tkzDefPoint(1.5,0){b_1}
\tkzDefPoint(0,0){b_2}
\tkzDefPoint(0.75,1.3){a_i}
\tkzLabelPoint[right](b_1){$b_1$}
\tkzLabelPoint[above](a_i){$a_i$}
\tkzLabelPoint[left](b_2){$b_2$}
\foreach \n in {b_1,b_2,a_i}
  \node at (\n)[circle,fill,inner sep=1.5pt]{};
\end{tikzpicture}
+\frac{1}{16}\cdot
\begin{tikzpicture}[scale=0.6, thick, every node/.style={sloped, align=center, allow upside down}, baseline=($(b_0.base)!.5!(a_i.base)$)]
\draw (0.75,1.3) -- pic[pos=0.6]{arrow=latex} (0,0) -- pic[pos=0.6]{arrow=latex} (0.75,-1.3) -- pic[pos=0.6]{arrow=latex} (1.5,0) -- pic[pos=0.6]{arrow=latex} (0.75,1.3);
\tkzDefPoint(0,0){b_2}
\tkzDefPoint(0.75,-1.3){b_0}
\tkzDefPoint(1.5,0){b_1}
\tkzDefPoint(0.75,1.3){a_i}
\tkzLabelPoint[left](b_2){$b_2$}
\tkzLabelPoint[below](b_0){$b_0$}
\tkzLabelPoint[right](b_1){$b_1$}
\tkzLabelPoint[above](a_i){$a_i$}
\foreach \n in {b_1,b_2,a_i,b_0}
  \node at (\n)[circle,fill,inner sep=1.5pt]{};
\end{tikzpicture}
\right)
\\
& +\frac{1}{8}\cdot\frac{1}{\mathcal{D}(\beta,\lambda)}\cdot \left(\frac{1}{4}\cdot(1-\lambda)\cdot
\begin{tikzpicture}[scale=0.7, thick, every node/.style={sloped, align=center, allow upside down}, baseline=($(b_1.base)!.5!(a_i.base)$)]
\draw (0.75,1.3) -- pic[pos=0.6]{arrow=latex} (1.5,0) -- pic[pos=0.6]{arrow=latex} (0,0) -- pic[pos=0.6]{arrow=latex} (0.75,1.3);
\tkzDefPoint(1.5,0){b_1}
\tkzDefPoint(0,0){b_2}
\tkzDefPoint(0.75,1.3){a_i}
\tkzLabelPoint[right](b_1){$b_1$}
\tkzLabelPoint[above](a_i){$a_i$}
\tkzLabelPoint[left](b_2){$b_2$}
\foreach \n in {b_1,b_2,a_i}
  \node at (\n)[circle,fill,inner sep=1.5pt]{};
\end{tikzpicture}
+\frac{1}{16}\cdot
\begin{tikzpicture}[scale=0.6, thick, every node/.style={sloped, align=center, allow upside down}, baseline=($(b_0.base)!.5!(a_i.base)$)]
\draw (0.75,1.3) -- pic[pos=0.7]{arrow=latex} (1.5,0) -- pic[pos=0.6]{arrow=latex} (0.75,-1.3) -- pic[pos=0.7]{arrow=latex} (0,0) -- pic[pos=0.8]{arrow=latex} cycle;
\tkzDefPoint(0,0){b_2}
\tkzDefPoint(0.75,-1.3){b_0}
\tkzDefPoint(1.5,0){b_1}
\tkzDefPoint(0.75,1.3){a_i}
\tkzLabelPoint[left](b_2){$b_2$}
\tkzLabelPoint[below](b_0){$b_0$}
\tkzLabelPoint[right](b_1){$b_1$}
\tkzLabelPoint[above](a_i){$a_i$}
\foreach \n in {b_1,b_2,a_i,b_0}
  \node at (\n)[circle,fill,inner sep=1.5pt]{};
\end{tikzpicture}
\right).
\end{aligned}
\end{equation}

If \underline{$a_i\neq a_j$ and $a_i\in V_N\setminus V_0$}: See Figure \ref{fig:startingposition}(\subref{subfig3}). Formula \eqref{eq:Schur2} becomes
\begin{equation}
\label{eq:Schuroffdiag}
\begin{aligned}
\sum_{b_k, b_l}& B(a_i, b_k) (D-\lambda I)^{-1} (b_k, b_l) C(b_l, a_j)
\\
& =\frac{1}{16}\cdot\frac{1}{\mathcal{D}(\beta,\lambda)}\left((1-\lambda)^2-\frac{1}{16}\right)\cdot\left(
\begin{tikzpicture}[thick, every node/.style={sloped, align=center, allow upside down}, baseline=($(a_i.base)!.5!(a_j.base)$)]
\draw (0,0) -- pic[pos=0.6]{arrow=latex} (0.75,0) -- pic[pos=0.6]{arrow=latex} (1.5,0);
\tkzDefPoint(0,0){a_i}
\tkzDefPoint(0.75,0){b_2}
\tkzDefPoint(1.5,0){a_j}
\tkzLabelPoint[below](a_i){$a_i$}
\tkzLabelPoint[below](b_2){$b_2$}
\tkzLabelPoint[below](a_j){$a_j$}
\foreach \n in {a_i,b_2,a_j}
  \node at (\n)[circle,fill,inner sep=1.5pt]{};
\end{tikzpicture}
\right)
\\
& +\frac{1}{16}\cdot\frac{1}{\mathcal{D}(\beta,\lambda)}\cdot\frac{1}{4}(1-\lambda)\cdot\left(
\begin{tikzpicture}[scale=0.5, thick, every node/.style={sloped, align=center, allow upside down}, baseline=($(b_1.base)!.5!(a_i.base)$)]
\draw (0,0) -- pic[pos=0.6]{arrow=latex} (0.75,1.3) -- pic[pos=0.6]{arrow=latex} (1.5,0) -- pic[pos=0.6]{arrow=latex} (3,0);
\tkzDefPoint(0,0){a_i}
\tkzDefPoint(1.5,0){b_2}
\tkzDefPoint(3,0){a_j}
\tkzDefPoint(0.75,1.3){b_1}
\tkzLabelPoint[left](a_i){$a_i$}
\tkzLabelPoint[below](b_2){$b_2$}
\tkzLabelPoint[right](a_j){$a_j$}
\tkzLabelPoint[above](b_1){$b_1$}
\foreach \n in {a_i,b_1,b_2,a_j}
  \node at (\n)[circle,fill,inner sep=1.5pt]{};
\end{tikzpicture}
+
\begin{tikzpicture}[scale=0.5, thick, every node/.style={sloped, align=center, allow upside down}, baseline=($(b_1.base)!.5!(a_i.base)$)]
\draw (0,0) -- pic[pos=0.6]{arrow=latex} (0.75,1.3) -- pic[pos=0.6]{arrow=latex} (2.25,1.3) -- pic[pos=0.6]{arrow=latex} (3,0);
\tkzDefPoint(0,0){a_i}
\tkzDefPoint(2.25,1.3){b_0}
\tkzDefPoint(3,0){a_j}
\tkzDefPoint(0.75,1.3){b_1}
\tkzLabelPoint[left](a_i){$a_i$}
\tkzLabelPoint[above](b_0){$b_0$}
\tkzLabelPoint[right](a_j){$a_j$}
\tkzLabelPoint[above](b_1){$b_1$}
\foreach \n in {a_i,b_1,b_0,a_j}
  \node at (\n)[circle,fill,inner sep=1.5pt]{};
\end{tikzpicture}
+
\begin{tikzpicture}[scale=0.5, thick, every node/.style={sloped, align=center, allow upside down}, baseline=($(b_0.base)!.5!(a_i.base)$)]
\draw (0,0) -- pic[pos=0.6]{arrow=latex} (1.5,0) -- pic[pos=0.6]{arrow=latex} (2.25,1.3) -- pic[pos=0.6]{arrow=latex} (3,0);
\tkzDefPoint(0,0){a_i}
\tkzDefPoint(1.5,0){b_2}
\tkzDefPoint(3,0){a_j}
\tkzDefPoint(2.25,1.3){b_0}
\tkzLabelPoint[left](a_i){$a_i$}
\tkzLabelPoint[above](b_0){$b_0$}
\tkzLabelPoint[right](a_j){$a_j$}
\tkzLabelPoint[below](b_2){$b_2$}
\foreach \n in {a_i,b_2,b_0,a_j}
  \node at (\n)[circle,fill,inner sep=1.5pt]{};
\end{tikzpicture}
\right)
\\
& +\frac{1}{16}\cdot\frac{1}{\mathcal{D}(\beta,\lambda)}\cdot\frac{1}{16}\cdot\left(
\begin{tikzpicture}[scale=0.5, thick, every node/.style={sloped, align=center, allow upside down}, baseline=($(b_0.base)!.5!(a_i.base)$)]
\draw (0,0) -- pic[pos=0.6]{arrow=latex} (0.75,1.3) -- pic[pos=0.6]{arrow=latex} (1.5,0) -- pic[pos=0.6]{arrow=latex} (2.25,1.3) -- pic[pos=0.6]{arrow=latex} (3,0);
\tkzDefPoint(0,0){a_i}
\tkzDefPoint(0.75,1.3){b_1}
\tkzDefPoint(1.5,0){b_2}
\tkzDefPoint(3,0){a_j}
\tkzDefPoint(2.25,1.3){b_0}
\tkzLabelPoint[left](a_i){$a_i$}
\tkzLabelPoint[above](b_0){$b_0$}
\tkzLabelPoint[right](a_j){$a_j$}
\tkzLabelPoint[below](b_2){$b_2$}
\tkzLabelPoint[above](b_1){$b_1$}
\foreach \n in {a_i,b_2,b_0,a_j,b_1}
  \node at (\n)[circle,fill,inner sep=1.5pt]{};
\end{tikzpicture}
+
\begin{tikzpicture}[scale=0.5, thick, every node/.style={sloped, align=center, allow upside down}, baseline=($(b_0.base)!.5!(a_i.base)$)]
\draw (0,0) -- pic[pos=0.6]{arrow=latex} (0.75,1.3) -- pic[pos=0.6]{arrow=latex} (2.25,1.3) -- pic[pos=0.6]{arrow=latex} (1.5,0) -- pic[pos=0.6]{arrow=latex} (3,0);
\tkzDefPoint(0,0){a_i}
\tkzDefPoint(0.75,1.3){b_1}
\tkzDefPoint(1.5,0){b_2}
\tkzDefPoint(3,0){a_j}
\tkzDefPoint(2.25,1.3){b_0}
\tkzLabelPoint[left](a_i){$a_i$}
\tkzLabelPoint[above](b_0){$b_0$}
\tkzLabelPoint[right](a_j){$a_j$}
\tkzLabelPoint[below](b_2){$b_2$}
\tkzLabelPoint[above](b_1){$b_1$}
\foreach \n in {a_i,b_2,b_0,a_j,b_1}
  \node at (\n)[circle,fill,inner sep=1.5pt]{};
\end{tikzpicture}
+
\begin{tikzpicture}[scale=0.5, thick, every node/.style={sloped, align=center, allow upside down}, baseline=($(b_0.base)!.5!(a_i.base)$)]
\draw (0,0) -- pic[pos=0.6]{arrow=latex} (1.5,0) -- pic[pos=0.6]{arrow=latex} (0.75,1.3) -- pic[pos=0.6]{arrow=latex} (2.25,1.3) -- pic[pos=0.6]{arrow=latex} (3,0);
\tkzDefPoint(0,0){a_i}
\tkzDefPoint(0.75,1.3){b_1}
\tkzDefPoint(1.5,0){b_2}
\tkzDefPoint(3,0){a_j}
\tkzDefPoint(2.25,1.3){b_0}
\tkzLabelPoint[left](a_i){$a_i$}
\tkzLabelPoint[above](b_0){$b_0$}
\tkzLabelPoint[right](a_j){$a_j$}
\tkzLabelPoint[below](b_2){$b_2$}
\tkzLabelPoint[above](b_1){$b_1$}
\foreach \n in {a_i,b_2,b_0,a_j,b_1}
  \node at (\n)[circle,fill,inner sep=1.5pt]{};
\end{tikzpicture}
\right)
\\
&=\frac{1}{16}\cdot\frac{1}{\mathcal{D}(\beta,\lambda)}\cdot\left(
\begin{tikzpicture}
[thick, every node/.style={sloped, align=center, allow upside down}, baseline=($(a_i.base)!.5!(a_j.base)$)]
\draw (0,0) -- pic[pos=0.6]{arrow=latex} (0.75,0) -- pic[pos=0.6]{arrow=latex} (1.5,0);
\tkzDefPoint(0,0){a_i}
\tkzDefPoint(0.75,0){b_2}
\tkzDefPoint(1.5,0){a_j}
\tkzLabelPoint[below](a_i){$a_i$}
\tkzLabelPoint[below](b_2){$b_2$}
\tkzLabelPoint[below](a_j){$a_j$}
\foreach \n in {a_i,b_2,a_j}
  \node at (\n)[circle,fill,inner sep=1.5pt]{};
\end{tikzpicture}
\right)\cdot\left((1-\lambda)^2-\frac{1}{16}\right)
\\
& +\frac{1-\lambda}{64\mathcal{D}(\beta,\lambda)}\cdot\left(
\begin{tikzpicture}
[thick, every node/.style={sloped, align=center, allow upside down}, baseline=($(a_i.base)!.5!(a_j.base)$)]
\draw (0,0) -- pic[pos=0.6]{arrow=latex} (0.75,0) -- pic[pos=0.6]{arrow=latex} (1.5,0);
\tkzDefPoint(0,0){a_i}
\tkzDefPoint(0.75,0){b_2}
\tkzDefPoint(1.5,0){a_j}
\tkzLabelPoint[below](a_i){$a_i$}
\tkzLabelPoint[below](b_2){$b_2$}
\tkzLabelPoint[below](a_j){$a_j$}
\foreach \n in {a_i,b_2,a_j}
  \node at (\n)[circle,fill,inner sep=1.5pt]{};
\end{tikzpicture}
\right)\cdot\left(
\begin{tikzpicture}[scale=0.5, thick, every node/.style={sloped, align=center, allow upside down}, baseline=($(a_i.base)!.5!(b_1.base)$)]
\draw (0,0) -- pic[pos=0.6]{arrow=latex} (0.75,1.3) -- pic[pos=0.6]{arrow=latex} (1.5,0) -- pic[pos=0.6]{arrow=latex} (0,0);
\tkzDefPoint(0,0){a_i}
\tkzDefPoint(1.5,0){b_2}
\tkzDefPoint(0.75,1.3){b_1}
\tkzLabelPoint[left](a_i){$a_i$}
\tkzLabelPoint[right](b_2){$b_2$}
\tkzLabelPoint[above](b_1){$b_1$}
\foreach \n in {a_i,b_2,b_1}
  \node at (\n)[circle,fill,inner sep=1.5pt]{};
\end{tikzpicture}
+
\begin{tikzpicture}[scale=0.5, thick, every node/.style={sloped, align=center, allow upside down}, baseline=($(a_i.base)!.5!(b_1.base)$)]
\draw (0,0) -- pic[pos=0.6]{arrow=latex} (0.75,1.3) -- pic[pos=0.6]{arrow=latex} (2.25,1.3) -- pic[pos=0.6]{arrow=latex} (3,0) -- pic[pos=0.6]{arrow=latex} (1.5,0) -- pic[pos=0.6]{arrow=latex} (0,0);
\tkzDefPoint(0,0){a_i}
\tkzDefPoint(1.5,0){b_2}
\tkzDefPoint(0.75,1.3){b_1}
\tkzDefPoint(2.25,1.3){b_0}
\tkzDefPoint(3,0){a_j}
\tkzLabelPoint[left](a_i){$a_i$}
\tkzLabelPoint[below](b_2){$b_2$}
\tkzLabelPoint[above](b_1){$b_1$}
\tkzLabelPoint[above](b_0){$b_0$}
\tkzLabelPoint[right](a_j){$a_j$}
\foreach \n in {a_i,b_2,b_1,b_0,a_j}
  \node at (\n)[circle,fill,inner sep=1.5pt]{};
\end{tikzpicture}
+
\begin{tikzpicture}[scale=0.5, thick, every node/.style={sloped, align=center, allow upside down}, baseline=($(a_j.base)!.5!(b_0.base)$)]
\draw (0,0) -- pic[pos=0.6]{arrow=latex} (0.75,1.3) -- pic[pos=0.6]{arrow=latex} (1.5,0) -- pic[pos=0.6]{arrow=latex} (0,0);
\tkzDefPoint(0,0){b_2}
\tkzDefPoint(1.5,0){a_j}
\tkzDefPoint(0.75,1.3){b_0}
\tkzLabelPoint[left](b_2){$b_2$}
\tkzLabelPoint[right](a_j){$a_j$}
\tkzLabelPoint[above](b_0){$b_0$}
\foreach \n in {a_j,b_2,b_0}
  \node at (\n)[circle,fill,inner sep=1.5pt]{};
\end{tikzpicture}
\right)
\\
& +\frac{1}{256\mathcal{D}(\beta,\lambda)}\cdot\left(
\begin{tikzpicture}
[thick, every node/.style={sloped, align=center, allow upside down}, baseline=($(a_i.base)!.5!(a_j.base)$)]
\draw (0,0) -- pic[pos=0.6]{arrow=latex} (0.75,0) -- pic[pos=0.6]{arrow=latex} (1.5,0);
\tkzDefPoint(0,0){a_i}
\tkzDefPoint(0.75,0){b_2}
\tkzDefPoint(1.5,0){a_j}
\tkzLabelPoint[below](a_i){$a_i$}
\tkzLabelPoint[below](b_2){$b_2$}
\tkzLabelPoint[below](a_j){$a_j$}
\foreach \n in {a_i,b_2,a_j}
  \node at (\n)[circle,fill,inner sep=1.5pt]{};
\end{tikzpicture}
\right)\cdot\left(
\begin{tikzpicture}[scale=0.5, thick, every node/.style={sloped, align=center, allow upside down}, baseline=($(a_i.base)!.5!(b_1.base)$)]
\draw (0,0) -- pic[pos=0.6]{arrow=latex} (0.75,1.3) -- pic[pos=0.6]{arrow=latex} (1.5,0) -- pic[pos=0.6]{arrow=latex} (0,0);
\draw (1.5,0) -- pic[pos=0.6]{arrow=latex} (2.25,1.3) -- pic[pos=0.6]{arrow=latex} (3,0) -- pic[pos=0.6]{arrow=latex} (1.5,0);
\tkzDefPoint(0,0){a_i}
\tkzDefPoint(1.5,0){b_2}
\tkzDefPoint(0.75,1.3){b_1}
\tkzDefPoint(2.25,1.3){b_0}
\tkzDefPoint(3,0){a_j}
\tkzLabelPoint[left](a_i){$a_i$}
\tkzLabelPoint[below](b_2){$b_2$}
\tkzLabelPoint[above](b_1){$b_1$}
\tkzLabelPoint[above](b_0){$b_0$}
\tkzLabelPoint[right](a_j){$a_j$}
\foreach \n in {a_i,b_2,b_1,b_0,a_j}
  \node at (\n)[circle,fill,inner sep=1.5pt]{};
\end{tikzpicture}
+
\begin{tikzpicture}[scale=0.5, thick, every node/.style={sloped, align=center, allow upside down}, baseline=($(a_i.base)!.5!(b_1.base)$)]
\draw (0,0) -- pic[pos=0.6]{arrow=latex} (0.75,1.3) -- pic[pos=0.6]{arrow=latex} (2.25,1.3) -- pic[pos=0.6]{arrow=latex} (1.5,0) -- pic[pos=0.6]{arrow=latex} (0,0);
\tkzDefPoint(0,0){a_i}
\tkzDefPoint(1.5,0){b_2}
\tkzDefPoint(0.75,1.3){b_1}
\tkzDefPoint(2.25,1.3){b_0}
\tkzLabelPoint[left](a_i){$a_i$}
\tkzLabelPoint[right](b_2){$b_2$}
\tkzLabelPoint[above](b_1){$b_1$}
\tkzLabelPoint[above](b_0){$b_0$}
\foreach \n in {a_i,b_2,b_1,b_0}
  \node at (\n)[circle,fill,inner sep=1.5pt]{};
\end{tikzpicture}
+
\begin{tikzpicture}[scale=0.5, thick, every node/.style={sloped, align=center, allow upside down}, baseline=($(a_j.base)!.5!(b_1.base)$)]
\draw (1.5,0) -- pic[pos=0.6]{arrow=latex} (0.75,1.3) -- pic[pos=0.6]{arrow=latex} (2.25,1.3) -- pic[pos=0.6]{arrow=latex} (3,0) -- pic[pos=0.6]{arrow=latex} (1.5,0);
\tkzDefPoint(3,0){a_j}
\tkzDefPoint(1.5,0){b_2}
\tkzDefPoint(0.75,1.3){b_1}
\tkzDefPoint(2.25,1.3){b_0}
\tkzLabelPoint[right](a_j){$a_j$}
\tkzLabelPoint[below](b_2){$b_2$}
\tkzLabelPoint[above](b_1){$b_1$}
\tkzLabelPoint[above](b_0){$b_0$}
\foreach \n in {a_j,b_2,b_1,b_0}
  \node at (\n)[circle,fill,inner sep=1.5pt]{};
\end{tikzpicture}
\right).
\end{aligned}
\end{equation}

If \underline{$a_i\neq a_j$ and $a_i\in V_0$:} $\sum_{b_k, b_l}B(a_i, b_k) (D-\lambda I)^{-1} (b_k, b_l) C(b_l, a_i)$ is half of equation \eqref{eq:Schuroffdiag}.

\subsection{Establishing spectral self-similarity}

Using Definition \ref{def:flux2} we simplify the expressions for the Schur complements. 
Note the following equivalent holonomy diagrams.
\begin{center}
\begin{tikzpicture}[thick, every node/.style={sloped, allow upside down}]
\begin{scope}[yshift=-0.4cm]
\draw (1.5,0) -- pic[pos=0.6]{arrow=latex} (0.75,1.3) -- pic[pos=0.6]{arrow=latex} (-0.75,1.3) -- pic[pos=0.6]{arrow=latex} (0,0) -- pic[pos=0.6]{arrow=latex} (1.5,0);
\tkzDefPoint(0,0){a_j}
\tkzDefPoint(1.5,0){b_2}
\tkzDefPoint(0.75,1.3){b_1}
\tkzDefPoint(-0.75,1.3){b_0}
\foreach \n in {a_j,b_2,b_1,b_0}
  \node at (\n)[circle,fill,inner sep=1.5pt]{};
\draw (1.9,0.6) node {${\Large\equiv}$};
\begin{scope}[xshift=3cm, thick, every node/.style={sloped, allow upside down}]
\draw (1.5,0) -- pic[pos=0.6]{arrow=latex} (0.75,1.3) -- pic[pos=0.65]{arrow=latex} (0,0) -- pic[pos=0.6]{arrow=latex} (1.5,0);
\draw (0,0) -- pic[pos=0.65]{arrow=latex} (0.75,1.3) -- pic[pos=0.6]{arrow=latex} (-0.75,1.3) -- pic[pos=0.6]{arrow=latex} (0,0);
\tkzDefPoint(0,0){a_j}
\tkzDefPoint(1.5,0){b_2}
\tkzDefPoint(0.75,1.3){b_1}
\tkzDefPoint(-0.75,1.3){b_0}
\foreach \n in {a_j,b_2,b_1,b_0}
  \node at (\n)[circle,fill,inner sep=1.5pt]{};
\end{scope}
\end{scope}
\begin{scope}[xshift=7cm, thick, every node/.style={sloped, allow upside down}]
\draw (0.75,1.3) -- pic[pos=0.6]{arrow=latex} (1.5,0) -- pic[pos=0.6]{arrow=latex} (0.75,-1.3) -- pic[pos=0.6]{arrow=latex} (0,0) -- pic[pos=0.6]{arrow=latex} (0.75,1.3);
\tkzDefPoint(0,0){b_2}
\tkzDefPoint(0.75,-1.3){b_0}
\tkzDefPoint(1.5,0){b_1}
\tkzDefPoint(0.75,1.3){a_i}
\foreach \n in {b_1,b_2,a_i,b_0}
  \node at (\n)[circle,fill,inner sep=1.5pt]{};
\draw (2.2,0) node {${\Large\equiv}$};
\end{scope}
\begin{scope}[xshift=10cm, thick, every node/.style={sloped, allow upside down}]
\draw (0.75,1.3) -- pic[pos=0.6]{arrow=latex} (1.5,0) -- pic[pos=0.65]{arrow=latex} (0,0) -- pic[pos=0.6]{arrow=latex} (0.75,1.3);
\draw (0,0) -- pic[pos=0.65]{arrow=latex} (1.5,0) -- pic[pos=0.6]{arrow=latex} (0.75,-1.3) -- pic[pos=0.6]{arrow=latex} (0,0);
\tkzDefPoint(0,0){b_2}
\tkzDefPoint(0.75,-1.3){b_0}
\tkzDefPoint(1.5,0){b_1}
\tkzDefPoint(0.75,1.3){a_i}
\foreach \n in {b_1,b_2,a_i,b_0}
  \node at (\n)[circle,fill,inner sep=1.5pt]{};
\end{scope}
\end{tikzpicture}
\end{center}
Therefore we can reexpress \eqref{eq:Schurdiag1}, \eqref{eq:detofD}, and \eqref{eq:Schurdiag2} to get
\begin{equation}
\label{eq:symSchurdiag1}
S_N^{\omega}(\alpha,\beta,\lambda)(a_i, a_i)=1-\lambda-\frac{A(\alpha,\beta,\lambda)}{64\mathcal{D}(\beta,\lambda)}
\end{equation}
where
\begin{align}
\label{eq:A}
A(\alpha,\beta,\lambda)&=16\lambda^2-(32+4\cos(2\pi\alpha))\lambda+15+4\cos(2\pi\alpha)+\cos(2\pi(\alpha+\beta)),\\
\label{eq:D}
\mathcal{D}(\beta,\lambda)&=-\lambda^3+3\lambda^2-\frac{45}{16}\lambda+\frac{13}{16}-\frac{1}{32}\cos(2\pi\beta).
\end{align}

Similarly, due to the equivalent holonomy diagrams
\begin{center}
\begin{tikzpicture}[thick, every node/.style={sloped, align=center, allow upside down}]
\draw (0,0) -- pic[pos=0.6]{arrow=latex} (0.75,1.3) -- pic[pos=0.6]{arrow=latex} (2.25,1.3) -- pic[pos=0.6]{arrow=latex} (3,0) -- pic[pos=0.6]{arrow=latex} (1.5,0) -- pic[pos=0.6]{arrow=latex} (0,0);
\tkzDefPoint(0,0){a_i}
\tkzDefPoint(1.5,0){b_2}
\tkzDefPoint(0.75,1.3){b_1}
\tkzDefPoint(2.25,1.3){b_0}
\tkzDefPoint(3,0){a_j}
\foreach \n in {a_i,b_2,b_1,b_0,a_j}
  \node at (\n)[circle,fill,inner sep=1.5pt]{};
\draw (3.8,0.7) node {${\Large\equiv}$};
\begin{scope}[xshift=4.5cm]
\draw (0,0) -- pic[pos=0.6]{arrow=latex} (0.75,1.3) -- pic[pos=0.65]{arrow=latex} (1.5,0) -- pic[pos=0.6]{arrow=latex} (0,0);
\draw (1.5,0) -- pic[pos=0.65]{arrow=latex} (0.75,1.3) -- pic[pos=0.6]{arrow=latex} (2.25,1.3) -- pic[pos=0.65]{arrow=latex} (1.5,0);
\draw (1.5,0) -- pic[pos=0.65]{arrow=latex} (2.25,1.3) -- pic[pos=0.6]{arrow=latex} (3,0) -- pic[pos=0.6]{arrow=latex} (1.5,0);
\tkzDefPoint(0,0){a_i}
\tkzDefPoint(1.5,0){b_2}
\tkzDefPoint(0.75,1.3){b_1}
\tkzDefPoint(2.25,1.3){b_0}
\tkzDefPoint(3,0){a_j}
\foreach \n in {a_i,b_2,b_1,b_0,a_j}
  \node at (\n)[circle,fill,inner sep=1.5pt]{};
\end{scope}
\begin{scope}[xshift=9.5cm]
\draw (0,0) -- pic[pos=0.6]{arrow=latex} (0.75,1.3) -- pic[pos=0.6]{arrow=latex} (2.25,1.3) -- pic[pos=0.6]{arrow=latex} (1.5,0) -- pic[pos=0.6]{arrow=latex} (0,0);
\tkzDefPoint(0,0){a_i}
\tkzDefPoint(1.5,0){b_2}
\tkzDefPoint(0.75,1.3){b_1}
\tkzDefPoint(2.25,1.3){b_0}
\foreach \n in {a_i,b_2,b_1,b_0}
  \node at (\n)[circle,fill,inner sep=1.5pt]{};
\end{scope}
\draw (12.2,0.7) node {${\Large\equiv}$};
\begin{scope}[xshift=12.5cm]
\draw (0,0) -- pic[pos=0.6]{arrow=latex} (0.75,1.3) -- pic[pos=0.65]{arrow=latex} (1.5,0) -- pic[pos=0.6]{arrow=latex} (0,0);
\draw (1.5,0) -- pic[pos=0.65]{arrow=latex} (0.75,1.3) -- pic[pos=0.6]{arrow=latex} (2.25,1.3) -- pic[pos=0.6]{arrow=latex} (1.5,0);
\tkzDefPoint(0,0){a_i}
\tkzDefPoint(1.5,0){b_2}
\tkzDefPoint(0.75,1.3){b_1}
\tkzDefPoint(2.25,1.3){b_0}
\foreach \n in {a_i,b_2,b_1,b_0}
  \node at (\n)[circle,fill,inner sep=1.5pt]{};
\end{scope}
\end{tikzpicture}
\end{center}
we can rewrite \eqref{eq:Schuroffdiag} to get
\begin{equation}
\label{eq:symSchuroffdiag}
S_N^{\omega}(\alpha,\beta,\lambda)(a_i, a_j)=-\frac{\Psi(\alpha,\beta,\lambda)\omega_{a_i b_2}\omega_{b_2 a_j}}{16\mathcal{D}(\beta,\lambda)},
\end{equation}
where
\begin{equation}
\label{eq:psi}
\Psi(\alpha,\beta,\lambda)=(1-\lambda)^2-\frac{1}{16}+\frac{1-\lambda}{4}(2e^{-2\pi i\alpha}+e^{-2\pi i(2\alpha+\beta)})+\frac{1}{16}(e^{-4\pi i \alpha}+2e^{-2\pi i (\alpha+\beta)}).
\end{equation}
Note that the exponents of \eqref{eq:psi} all carry a negative sign since the orientation of the edge $a_ia_j$ is counterclockwise, while the diagrams in  \eqref{eq:Schuroffdiag} have clockwise orientation. If the orientation of $a_ia_j$ is clockwise, replace all the exponents in $S_N^{\omega}(a_i,a_j)$ with a positive sign.

We summarize the preceding arguments as follows.

\begin{proposition}
\label{prop:Schur}
Let $\Delta_{a_ia_j}$ be the upright triangle that the edge $a_ia_j$ belongs to, and $b_2$ be the midpoint of $a_i$ and $a_j$. 
We have
\begin{equation}
S_N^{\omega}(\alpha,\beta,\lambda)(a_i,a_j)=
\begin{cases}
\vspace{4pt}
\displaystyle 1-\lambda-\frac{A(\alpha,\beta,\lambda)}{64\mathcal{D}(\beta,\lambda)} & \text{if $a_i=a_j$}, \\
\vspace{4pt}
\displaystyle -\frac{\Psi(\alpha,\beta,\lambda)\omega_{a_ib_2}\omega_{b_2a_j}}{16\mathcal{D}(\beta,\lambda)} & \text{if $a_i\neq a_j$, $a_i\in V_N\setminus V_0$, and $\Delta_{a_ia_j}$ is traversed CCW}, \\
\vspace{4pt}
\displaystyle -\frac{\Psi(\alpha,\beta,\lambda)\omega_{a_ib_2}\omega_{b_2a_j}}{8\mathcal{D}(\beta,\lambda)} & \text{if $a_i\neq a_j$, $a_i\in V_0$, and $\Delta_{a_ia_j}$ is traversed CCW}, \\
\vspace{4pt}
\displaystyle -\frac{\overline{\Psi(\alpha,\beta,\lambda)}\omega_{a_ib_2}\omega_{b_2a_j}}{16\mathcal{D}(\beta,\lambda)} & \text{if $a_i\neq a_j$, $a_i\in V_N\setminus V_0$, and $\Delta_{a_ia_j}$ is traversed CW}, \\
\vspace{4pt}
\displaystyle -\frac{\overline{\Psi(\alpha,\beta,\lambda)}\omega_{a_ib_2}\omega_{b_2a_j}}{8\mathcal{D}(\beta,\lambda)} & \text{if $a_i\neq a_j$, $a_i\in V_0$, and $\Delta_{a_ia_j}$ is traversed CW}, \\
\end{cases}
\end{equation}
where $A(\alpha,\beta,\lambda)$, $\mathcal{D}(\beta,\lambda)$, and $\Psi(\alpha,\beta,\lambda)$ were defined respectively in  \eqref{eq:A}, \eqref{eq:D}, and \eqref{eq:psi}.
\end{proposition}

\begin{corollary}[Spectral decimation identity]
\label{cor:specdec}
The Schur complement in Proposition \ref{prop:Schur} can be reexpressed as
\begin{equation}
\label{eq:specdec}
S_N^{\omega}(\alpha,\beta,\lambda)=\phi(\alpha,\beta,\lambda)(\mathcal{L}_{N-1}^{\Omega}-R(\alpha,\beta,\lambda))\qquad (\lambda\in \mathbb{R})
\end{equation}
where $\mathcal{L}^\Omega_{N-1}$ is the magnetic Laplacian on $V_{N-1}$ with $U(1)$ connection $\Omega$, a self-adjoint operator on $L^2(V_{N-1}, \deg_{G_{N-1}})$,
and $R(\alpha,\beta,\lambda)$ is the spectral decimation function.
Specifically:
\begin{enumerate}[start=1, label={(Case \arabic*)}, wide]
\item \label{item:R}: If $\mathbb{R} \ni \lambda\mapsto \Psi(\alpha,\beta,\lambda)$ is $\mathbb{R}$-valued, then
\begin{align}
\label{eq:phi}
\phi(\alpha,\beta,\lambda)&=\frac{\Psi(\alpha,\beta,\lambda)}{4\mathcal{D}(\beta,\lambda)},\\
\label{eq:R}
R(\alpha,\beta,\lambda)&=1+\frac{A(\alpha,\beta,\lambda)-64\mathcal{D}(\beta,\lambda)(1-\lambda)}{16\Psi(\alpha,\beta,\lambda)},\\
\label{eq:Omega}
\Omega_{ab}(\alpha,\beta)&=\omega_{ac}\omega_{cb}.
\end{align}
\item \label{item:C}: If $\mathbb{R}\ni \lambda\mapsto \Psi(\alpha,\beta,\lambda)$ is $\mathbb{C}$-valued, then
\begin{align}
\label{eq:abs_phi}
\phi(\alpha,\beta,\lambda)&=\frac{|\Psi(\alpha,\beta,\lambda)|}{4\mathcal{D}(\beta,\lambda)},\\
\label{eq:abs_R}
R(\alpha,\beta,\lambda)&=1+\frac{A(\alpha,\beta,\lambda)-64\mathcal{D}(\beta,\lambda)(1-\lambda)}{16|\Psi(\alpha,\beta,\lambda)|},\\
\label{eq:abs_theta}
\theta(\alpha,\beta,\lambda)&=\frac{\textup{arg} \hspace{0.1cm} \Psi(\alpha,\beta,\lambda)}{2\pi} \qquad ({\rm arg}: \mathbb{C} \to [0,2\pi)),\\
\label{eq:abs_Omega}
\Omega_{ab}(\alpha,\beta,\lambda)&=\omega_{ac}\omega_{cb}e^{2\pi i\theta(\alpha,\beta,\lambda)}.
\end{align}
\end{enumerate}
In both cases, $a\sim c\sim b$, and the upright cell to which the edge $ab$ belongs is traversed counterclockwise.
\end{corollary}

Two important remarks are in order.
First, $A(\alpha,\beta,\lambda)$, $\mathcal{D}(\beta,\lambda)$, and $\Psi(\alpha,\beta,\lambda)$ are all independent of the level $N$, and therefore so is $R(\alpha,\beta,\lambda)$. This is the essence of spectral self-similarity and what allows us to characterize the spectrum recursively.
Second, in Corollary \ref{cor:specdec}-\ref{item:R}, the connection $\Omega$ is manifestly independent of $\lambda$, whereas in Corollary \ref{cor:specdec}-\ref{item:C}, $\Omega$ receives an extra ``twist'' by a unit complex number $e^{2\pi i\theta}$, which depends on $\lambda$ in general.
\emph{There does not seem to be an easy way to eliminate this twist via gauge transformations.}

The following was first noted by \cite{A84} and invoked later in \cites{Ghez, quasicrystals}.

\begin{proposition}[Evolution of the magnetic flux under spectral decimation]
\label{prop:fluxevolve}
Let the magnetic flux going through every upright triangle on level $N$ be $\alpha_N$, and downright triangle, $\beta_N$. Then
\begin{align}
\label{eq:alchange}
\alpha_{N-1}= \alpha_\downarrow(\alpha_N, \beta_N, \lambda)&=3\alpha_N+\beta_N+3\theta(\alpha_N,\beta_N,\lambda),\\
\label{eq:bechange}
\beta_{N-1}=\beta_\downarrow(\alpha_N, \beta_N, \lambda)&=3\beta_N+\alpha_N-3\theta(\alpha_N,\beta_N,\lambda),
\end{align}
so $\alpha_{N-1}+\beta_{N-1} = 4(\alpha_N+\beta_N)$.
Specifically, in the setting of Corollary \ref{cor:specdec}-\ref{item:R}, $\theta \equiv 0$.
\end{proposition}

\begin{proof}
By \eqref{eq:Omega} or \eqref{eq:abs_Omega},
$
\Omega_{a_1a_2}(\alpha,\beta,\lambda)=\omega_{a_1b_0}\omega_{b_0a_2}e^{2\pi i\theta(\alpha,\beta,\lambda)}
$; see the diagram below.
\begin{center}
\begin{tikzpicture}[scale=0.8]
\draw [thick] (0,0) -- (1.5,0) -- (3,0);
\draw (3,0) -- (2.25,1.3) -- (1.5,2.6) -- (0.75,1.3) -- (0,0);
\draw (1.5,0) -- (2.25,1.3) -- (0.75,1.3) -- (1.5,0);
\draw (3,0) -- (3.75,1.3) -- (4.5,2.6) -- (3,2.6) -- (1.5,2.6);
\draw (2.25,1.3) -- (3.75,1.3) -- (3,2.6) -- (2.25,1.3);
\draw (0.75,0.5) node {$\alpha_N$};
\draw (2.25,0.5) node {$\alpha_N$};
\draw (1.5,1.8) node {$\alpha_N$};
\draw (1.5,0.8) node {$\beta_N$};
\draw (3,1.8) node {$\alpha_N$};
\draw (3,0.8) node {$\beta_N$};
\draw (2.25,2.1) node {$\beta_N$};
\draw (3.75,2.1) node {$\beta_N$};
\tkzDefPoint(0,0){a_1}
\tkzDefPoint(1.5,0){b_0}
\tkzDefPoint(3,0){a_2}
\tkzDefPoint(1.5,2.6){a_0}
\tkzDefPoint(2.25,1.3){b_1}
\tkzDefPoint(0.75,1.3){b_2}
\tkzDefPoint(3.75,1.3){b_3}
\tkzDefPoint(3,2.6){b_4}
\tkzDefPoint(4.5,2.6){a_3}
\tkzLabelPoint[below](a_1){$a_1$}
\tkzLabelPoint[below](a_2){$a_2$}
\tkzLabelPoint[below](b_0){$b_0$}
\tkzLabelPoint[left](b_2){$b_2$}
\tkzLabelPoint[below, scale=0.8](b_1){$b_1$}
\tkzLabelPoint[above](a_0){$a_0$}
\tkzLabelPoint[above](a_3){$a_3$}
\tkzLabelPoint[right](b_3){$b_3$}
\tkzLabelPoint[above](b_4){$b_4$}
\foreach \n in {a_0,a_1,a_2,a_3,b_0,b_1,b_2,b_3,b_4}
  \node at (\n)[circle,fill,inner sep=1pt]{};
\draw [thick,->] (5,1.3) -- (6,1.3);
\begin{scope}[xshift=6.2cm]
\draw (3,0) -- (3.75,1.3) -- (4.5,2.6) -- (3,2.6) -- (1.5,2.6);
\draw [thick] (0,0) -- (3,0);
\draw (3,0) -- (1.5,2.6) -- (0,0);
\draw (1.5,1.1) node {$\alpha_{N-1}$};
\draw (3,1.8) node {$\beta_{N-1}$};
\tkzDefPoint(3,0){a_2}
\tkzDefPoint(1.5,2.6){a_0}
\tkzDefPoint(0,0){a_1}
\tkzDefPoint(4.5,2.6){a_3}
\tkzLabelPoint[above](a_0){$a_0$}
\tkzLabelPoint[below](a_1){$a_1$}
\tkzLabelPoint[below](a_2){$a_2$}
\tkzLabelPoint[above](a_3){$a_3$}
\foreach \n in {a_0,a_1,a_2,a_3}
  \node at (\n)[circle,fill,inner sep=1pt]{};
\end{scope}
\end{tikzpicture}
\end{center}
By Definition \ref{def:flux2},
\[
\begin{aligned}
e^{2\pi i\alpha_{N-1}} & = \Omega_{a_1a_2}\Omega_{a_2a_0}\Omega_{a_0a_1}  =\omega_{a_1b_0}\omega_{b_0a_2}\omega_{a_2b_1}\omega_{b_1a_0}\omega_{a_0b_2}\omega_{b_2a_1}e^{2\pi i (3\theta(\alpha,\beta,\lambda))} \\
& = e^{2\pi i (3\alpha_N+\beta_N)}e^{2\pi i (3\theta(\alpha_N,\beta_N,\lambda))} 
 = e^{2\pi i (3\alpha_N+\beta_N+3\theta(\alpha_N,\beta_N,\lambda))},
\end{aligned}
\]
and similarly for $e^{2\pi i \beta_{N-1}}$. This implies \eqref{eq:alchange} and \eqref{eq:bechange}.
\end{proof}


\section{Recursive characterization of the magnetic spectrum} \label{sec:solution}

In this section we explicitly characterize the spectrum $\sigma(\mathcal{L}^\omega_N)$ under Definition \ref{def:flux2}, thereby proving Theorems \ref{thm:1} and \ref{thm:2}.
Our approach is to specialize the results from \S\ref{sec:specdec} to
\begin{align*}
V= V_N, \quad V_\parallel= V_{N-1}, \quad V_\perp = V\setminus V_\parallel, \quad
M=\mathcal{L}^\omega_N, \quad L=\mathcal{L}^\Omega_{N-1},
\end{align*}
and involve all the functions referenced in Corollary \ref{cor:specdec}.

As a first step, we distinguish the case where the fluxes $\alpha, \beta \in \{0,\frac{1}{2}\}$ from the other cases.
This is made not just for convenience, but actually reflects the dichotomy between \ref{item:R} and \ref{item:C} in Corollary \ref{cor:specdec}.


\begin{proposition}
\label{prop:Rvalued}
The function $\mathbb{R}\ni \lambda\mapsto \Psi(\alpha,\beta,\lambda)$ is $\mathbb{R}$-valued if and only if $\alpha, \beta\in\{0,\frac{1}{2}\}$.
\end{proposition}
\begin{proof}
From \eqref{eq:psi} we have
\[
{\rm Im}(\Psi(\alpha,\beta,\lambda))=(2\sin(2\pi\alpha)+\sin(2\pi(2\alpha+\beta)))\frac{1-\lambda}{4}+\frac{1}{16}(\sin(4\pi\alpha)+2\sin(2\pi(\alpha+\beta))),
\]
which is identically zero for all $\lambda\in \mathbb{R}$ if and only if
\begin{align*}
    2\sin(2\pi\alpha)+\sin(2\pi(2\alpha+\beta))=0
    \quad
    \text{and}
    \quad
     \sin(4\pi\alpha)+2\sin(2\pi(\alpha+\beta))=0.
\end{align*}
Using the shorthands $X=\cos(2\pi\alpha)$ and $Y=\cos(2\pi\beta)$, and applying several trig identities (double-angle formula, sum-to-product formula), we rewrite the last condition as
\begin{align}
\label{eq:im1sim}
(2+2XY)\sqrt{1-X^2}&=(1-2X^2)\sqrt{1-Y^2},\\
\label{eq:im2sim}
\text{and} \qquad (X+Y)\sqrt{1-X^2}&=-X\sqrt{1-Y^2}.
\end{align}
Now square both sides of \eqref{eq:im1sim} and \eqref{eq:im2sim} and simplify to get
\begin{align}
\label{eq:im1sim2}
(4X^4+8X^3Y-8XY)-Y^2-4X^2(Y^2-1)+4X(Y^2-1)-3&=0,\\
\label{eq:im2sim2}
X^4+2X^3Y-2XY&=Y^2.
\end{align}
Using \eqref{eq:im2sim2} we replace $4X^4+8X^3Y-8XY$ by $4Y^2$ in \eqref{eq:im1sim2}, which can then be simplified to yield
$
(4X^2-4X+3)(Y^2-1)=0
$.
So $X=\cos(2\pi\alpha)=-\frac{1}{2}$ or $\frac{3}{2}$ (the latter is impossible), or $Y=\cos(2\pi\beta)=\pm 1$. 
In addition, by substituting $-\frac{1}{2}$ for $X$ in \eqref{eq:im1sim}, we see that $Y=\frac{3\pm i\sqrt{2}}{2}$ is $\mathbb{C}$-valued, so $X$ cannot be $-\frac{1}{2}$. 
Thus it must be that $\cos(2\pi\beta)=\pm 1$, so $\beta=0$ or $\frac{1}{2}$. In turn $\cos(2\pi\alpha)$ has to be $\pm 1$, \emph{i.e.,} $\alpha=0$ or $\frac{1}{2}$.
\end{proof}

We shall refer to the case $\alpha,\beta\in \{0,\frac{1}{2}\}$ as \underline{Case I}.


\subsection{Case-by-case analysis of the exceptional set}

In this subsection we systematically identify the exceptional set for spectral decimation of $\mathcal{L}^\omega_N$, \emph{cf.\@} \eqref{eq:eset}.
In fact, since $\sigma(\mathcal{L}^\omega_N) \subset \mathbb{R}$, it suffices to only consider real numbers in this set, namely: 
\begin{align}
\mathcal{E}(\alpha,\beta)=\{x\in \mathbb{R} : \text{$\mathcal{D}(\beta,x)=0$ or $\phi(\alpha,\beta,x)=0$}\}.
\end{align}
Recalling the cubic polynomial \eqref{eq:D}, which is the characteristic polynomial of a Hermitian $3\times 3$ matrix, we see that the three zeros of $\mathcal{D}(\beta,\cdot)$ (which does not depend on $\alpha$) belong to $\mathcal{E}(\alpha,\beta)$.
For reasons to be made clear later, we shall determine if any of the zeros appears multiple times.

\begin{lemma}
\label{lem:multzero}
The cubic polynomial $\mathcal{D}(\beta,\cdot)$, \eqref{eq:D}, has a multiple zero only if:
\begin{itemize}
\item $\beta=0$, in which case the zeros are $\frac{5}{4}$ (double) and $\frac{1}{2}$;
\item $\beta=\frac{1}{2}$, in which case the zeros are $\frac{3}{4}$ (double) and $\frac{3}{2}$.
\end{itemize}
\end{lemma}
\begin{proof}
It is easy to see that for any $\beta$, $\mathcal{D}(\beta,\cdot)$ cannot have a triple zero, since there is no $c\in \mathbb{R}$ such that $\mathcal{D}(\beta,x) = -(x-c)^3$.
To exhibit the double zeros, we find $c,c'\in \mathbb{R}$ such that $\mathcal{D}(\beta,x)= -(x-c)^2(x-c')$.
The RHS can be expanded to give $-x^3 +(c'+2c)x^2-c(2c'+c)x + c'c^2$.
Equating the coefficients on both sides leads to the system of equations $c'+2c=3$, $c(2c'+c)=\frac{45}{16}$, and $c'c^2=\frac{13}{16}-\frac{1}{32}\cos(2\pi\beta)$.
The claim follows.
\end{proof}

So it remains to identify the $\mathbb{R}$-valued zeros of $\phi(\alpha,\beta,\cdot)$, \emph{cf.\@} \eqref{eq:phi} or \eqref{eq:abs_phi}.
Actually we shall identify the $\mathbb{R}$-valued zeros of $\Psi(\alpha,\beta,\cdot)$, and check if any of them happens also to be a zero of $\mathcal{D}(\beta,\cdot)$.

In \underline{Case I}, we indicate in Table \ref{table:CaseA} the quadratic polynomial $\Psi$ and its $\mathbb{R}$-valued zeros.

\renewcommand*{\arraystretch}{1.3}

\begin{table}
\begin{center}
\caption{\underline{Case I}. The {\color{red}red} number is a double zero of $\mathcal{D}(\beta,\cdot)$.}
\label{table:CaseA}
\begin{tabular}{|c|c|c|c|}
\hline
$\alpha$ & $\beta$ & $\Psi(\alpha,\beta,x)$ & \begin{tabular}{@{}c@{}}$\mathbb{R}$-valued zeros\\ of $\Psi(\alpha,\beta,\cdot)$\end{tabular} \\
\hline
0 & 0 & $(1-x)^2+\frac{3}{4}(1-x)+\frac{1}{8}$ & {\color{red} $\frac{5}{4}$},$\frac{3}{2}$ \\
\hline
0 & $\frac{1}{2}$ & $(1-x)^2+\frac{1}{4}(1-x)-\frac{1}{8}$ & {\color{red}$\frac{3}{4}$},$\frac{3}{2}$ \\
\hline
$\frac{1}{2}$ & 0 & $(1-x)^2-\frac{1}{4}(1-x)-\frac{1}{8}$ & $\frac{1}{2}$,{\color{red} $\frac{5}{4}$} \\
\hline
$\frac{1}{2}$ & $\frac{1}{2}$ & $(1-x)^2-\frac{3}{4}(1-x)+\frac{1}{8}$ & $\frac{1}{2}$,{\color{red}$\frac{3}{4}$} \\
\hline
\end{tabular}
\end{center}
\end{table}

Beyond \underline{Case I} we must apply Corollary \ref{cor:specdec}-\ref{item:C}.
The next natural scenario is when exactly one of $\alpha$ and $\beta$ belongs to $\{0,\frac{1}{2}\}$.
We call this \underline{Case II}.
In this case there is only one $\mathbb{R}$-valued zero of $\Psi(\alpha,\beta, \cdot)$, see Table \ref{tab:freepara}.

\begin{table}
    \begin{center}
     \caption{\underline{Case II}. The {\color{red}red} number is a double zero of $\mathcal{D}(\beta,\cdot)$.}
    \label{tab:freepara}
    \begin{tabular}{|c|c|c|c|}
    \hline
    $\alpha$ & $\beta$ & $\Psi(\alpha,\beta,x)$ & \begin{tabular}{@{}c@{}}$\mathbb{R}$-valued zero \\ of $\Psi(\alpha,\beta,\cdot)$\end{tabular} \\
    \hline
    $0$ & $\notin \{0,\frac{1}{2}\}$ & $\left((1-x)+\frac{1}{2}\right)\left((1-x)+\frac{1}{4}e^{-2\pi i\beta}\right)$ & $\frac{3}{2}$ \\
    \hline
     $\notin \{0,\frac{1}{2}\}$ & $\frac{1}{2}$ & $\left((1-x)-\frac{1}{4}\right)\left((1-x)-\frac{1}{4}e^{-4\pi i\alpha}+\frac{1}{2}e^{-2\pi i\alpha}+\frac{1}{4}\right)$ & ${\color{red} \frac{3}{4}}$ \\
    \hline
    $\frac{1}{2}$ &  $\notin \{0,\frac{1}{2}\}$ & $\left((1-x)-\frac{1}{2}\right)\left((1-x)+\frac{1}{4}e^{-2\pi i\beta}\right)$ & $\frac{1}{2}$\\
    \hline
     $\notin \{0,\frac{1}{2}\}$ & $0$ & $\left((1-x)+\frac{1}{4}\right)\left((1-x)+\frac{1}{4}e^{-4\pi i\alpha}+\frac{1}{2}e^{-2\pi i\alpha}-\frac{1}{4}\right)$ & ${\color{red} \frac{5}{4}}$\\
    \hline
    \end{tabular}
    \end{center}
\end{table}

Now we consider $\alpha,\beta \notin \{0,\frac{1}{2}\}$.
It turns out that there is a line in the $(\alpha,\beta)$-parameter space on which $\Psi(\alpha,\beta,\cdot)$ has an $\mathbb{R}$-valued zero.
This line corresponds to having half-integer fluxes through all the upright triangles of side length $2$ in the graph distance.

\begin{proposition}
\label{exceptionalset2}
Suppose $\alpha,\beta\notin\left\{0,\frac{1}{2}\right\}$. 
Then $\Psi(\alpha,\beta,\cdot)$ has an $\mathbb{R}$-valued zero if and only if $3\alpha+\beta=\frac{1}{2} \pmod 1$.
If so, this zero is unique, equals $1+\frac{1}{2}\cos(2\pi\alpha)$, and is not a zero of $\mathcal{D}(\beta,\cdot)$.
\end{proposition}

\begin{proof}
We solve ${\rm {Re}}(\Psi(\alpha,\beta,x)) =0 $ and ${\rm {Im}}(\Psi(\alpha,\beta,x))=0$ simultaneously.
Let us mention that the assumption $\alpha,\beta\notin \{0,\frac{1}{2}\}$ implies that $\sin(2\pi\alpha)\neq 0$ and $\sin(2\pi\beta)\neq 0$.
Denoting $\eta=1-x$, we have
\begin{align*}
{\rm Re}(\Psi(\alpha,\beta,\lambda))&=\eta^2+\frac{\eta}{4}\left(2\cos(2\pi\alpha)+\cos(2\pi(2\alpha+\beta))\right) + \frac{1}{16}\left(\cos(4\pi\alpha)+2\cos(2\pi(\alpha+\beta))-1\right),\\
{\rm Im}(\Psi(\alpha,\beta,\lambda))&=\frac{\eta}{4}(2\sin(2\pi\alpha)+\sin(2\pi(2\alpha+\beta)))+\frac{1}{16}(\sin(4\pi\alpha)+2\sin(2\pi(\alpha+\beta))).
\end{align*}
Then ${\rm Im}(\Psi)=0$ is equivalent to
\begin{equation}
\label{eq:lsubs}
\eta=-\frac{1}{4}\cdot\frac{\sin(4\pi\alpha)+2\sin(2\pi(\alpha+\beta))}{2\sin(2\pi\alpha)+\sin(2\pi(2\alpha+\beta))}.
\end{equation}
Substitute this into ${\rm Re}(\Psi)=0$ to get
\[
\begin{aligned}
& \frac{(\sin(4\pi\alpha)+2\sin(2\pi(\alpha+\beta)))^2}{(2\sin(2\pi\alpha)+\sin(2\pi(2\alpha+\beta)))^2}-\frac{(2\cos(2\pi\alpha)+\cos(2\pi(2\alpha+\beta)))(\sin(4\pi\alpha)+2\sin(2\pi(\alpha+\beta)))}{2\sin(2\pi\alpha)+\sin(2\pi(2\alpha+\beta))} \\
& +(\cos(4\pi\alpha)+2\cos(2\pi(\alpha+\beta))-1)=0.
\end{aligned}
\]
Now multiply both sides by $[2\sin(2\pi\alpha)+\sin(2\pi(2\alpha+\beta))]^2$, a nonzero quantity, to get
\[
\begin{aligned}
& [\sin(4\pi\alpha)+2\sin(2\pi(\alpha+\beta))]^2-[2\sin(2\pi\alpha)+\sin(2\pi(2\alpha+\beta))]^2\\
&-(2\sin(2\pi\alpha)+\sin(2\pi(2\alpha+\beta))) \times [(2\cos(2\pi\alpha)+\cos(2\pi(2\alpha+\beta)))((\sin(4\pi\alpha)+2\sin(2\pi(\alpha+\beta)))\\
& \qquad-(\cos(4\pi\alpha)+2\cos(2\pi(\alpha+\beta)))(2\sin(2\pi\alpha)+\sin(2\pi(2\alpha+\beta)))]=0.
\end{aligned}
\]
Combining the appropriate terms in the last square bracket and using the sum-to-difference formulas for sine, we can simplify the last equation to
\begin{equation}
\label{eq:lsubs4}
\begin{aligned}
& [\sin(4\pi\alpha)+\sin(2\pi(2\alpha+\beta))+2\sin(2\pi(\alpha+\beta))+2\sin(2\pi\alpha)]\\
& \quad\times[\sin(4\pi\alpha)-\sin(2\pi(2\alpha+\beta))+2\sin(2\pi(\alpha+\beta))-2\sin(2\pi\alpha)]\\
& -3\sin(2\pi\beta)(2\sin(2\pi\alpha)+\sin(2\pi(2\alpha+\beta)))=0.
\end{aligned}
\end{equation}
Another application of the sum-to-product formulas on the sine functions reduces the expressions in the square brackets of \eqref{eq:lsubs4}, giving rise to
\[
\begin{aligned}
& 2\cos(\pi\beta)\cdot[2\sin(\pi(2\alpha+\beta))+\sin(\pi(4\alpha+\beta))]\times 2\sin(\pi\beta)\cdot[2\cos(\pi(2\alpha+\beta))-\cos(\pi(4\alpha+\beta))]\\
& \hspace{3cm} -3\sin(2\pi\beta)(2\sin(2\pi\alpha)+\sin(2\pi(2\alpha+\beta)))=0.
\end{aligned}
\]
We then divide both sides by $2\sin(2\pi \beta) = 2\cos(\pi \beta)\sin(\pi\beta)$  to get
\begin{equation}
\label{eq:lsubs6}
\begin{aligned}
& [2\sin(\pi(2\alpha+\beta))+\sin(\pi(4\alpha+\beta))][2\cos(\pi(2\alpha+\beta))-\cos(\pi(4\alpha+\beta))] \\
& \hspace{6cm} -\frac{3}{2}(2\sin(2\pi\alpha)+\sin(2\pi(2\alpha+\beta)))=0.
\end{aligned}
\end{equation}
Now we expand the first term on the LHS of \eqref{eq:lsubs6}, and use the double-angle formula and the sum-to-difference formulas for sine to simplify \eqref{eq:lsubs6} to
\begin{equation}
\label{eq:lsubs7}
\frac{1}{2}(\sin(2\pi(2\alpha+\beta))-\sin(2\pi(4\alpha+\beta)))-\sin(2\pi\alpha)=0.
\end{equation}
Applying again the sum-to-product formula to the first term on the LHS of \eqref{eq:lsubs7}, we get
\[
\sin(2\pi\alpha)(\cos(2\pi(3\alpha+\beta))+1)=0.
\]
Since $\sin(2\pi\alpha)\neq 0$, it must be that
$
\cos(2\pi(3\alpha+\beta))+1=0
$,
\emph{i.e.,\@} $3\alpha+\beta=\frac{1}{2} \pmod 1$. 
Finally, substitute this back into \eqref{eq:lsubs} leads to the conclusion that $1+\frac{1}{2}\cos(2\pi\alpha)$ is the only zero of $\Psi(\alpha,\beta,\cdot)$.
\end{proof}

In a nutshell, we have established four cases from which the exceptional set for spectral decimation is analyzed. They are:
\begin{enumerate}[label=\underline{Case \Roman*:}, wide]
\item $\alpha,\beta\in\{0,\frac{1}{2}\}$.
\item Only one of $\alpha$ and $\beta$ is in $\{0,\frac{1}{2}\}$. There is only one $\mathbb{R}$-valued zero of $\Psi(\alpha, \beta,\cdot)$, which may or may not be a (double) zero of $\mathcal{D}(\beta,\cdot)$.
\item $3\alpha+\beta =\frac{1}{2} \pmod 1$, excluding flux values already discussed in Cases I \& II. There is only one $\mathbb{R}$-valued zero of $\Psi(\alpha,\beta,\cdot)$, and it is not a zero of $\mathcal{D}(\beta,\cdot)$.
\item The remaining case. There are no $\mathbb{R}$-valued zeros of $\Psi(\alpha,\beta,\cdot)$.
\end{enumerate}
These are indicated in the flux parameter space in Figure \ref{fig:param}, and we summarize our main findings as follows.

\begin{proposition}[Exceptional set for spectral decimation of $\mathcal{L}^\omega_N$]
\label{prop:exceptionalsetlist}
The exceptional set $\mathcal{E}(\alpha,\beta)$ consists of:
\begin{itemize}
    \item The three zeros of $\mathcal{D}(\beta,\cdot)$; and
    \item The corresponding values $x$ in Table \ref{tab:exceptionalset} if any of the conditions in the first column is met.
\end{itemize}

\begin{table}[h!]
    \centering
        \caption{}
    \label{tab:exceptionalset}
    \begin{tabular}{|c|c|}
    \hline
    Condition & Value to be added to $\mathcal{E}(\alpha,\beta)$\\
    \hline
    $\alpha=0$ & $\frac{3}{2}$ \\
    \hline
    $\alpha=\frac{1}{2}$ & $\frac{1}{2}$ \\
    \hline
    $3\alpha+\beta=\frac{1}{2} \pmod 1$ & $1+\frac{1}{2}\cos(2\pi\alpha)$ \\
    \hline
    \end{tabular}
\end{table}
\end{proposition}

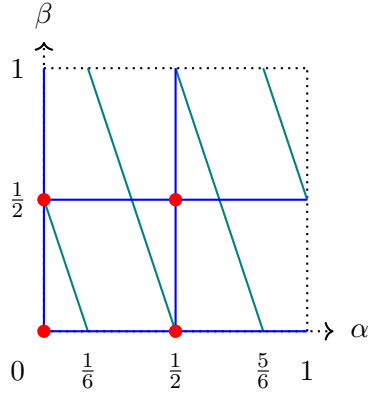
\begin{figure}
\centering
\begin{tikzpicture}[thick, scale=3.5]
\draw [dotted,->] (0,0) -- (1.1,0);
\draw [dotted,->] (0,0) -- (0,1.1);
\draw [dotted] (0,1) -- (1,1);
\draw [dotted] (1,0) -- (1,1);
\draw [teal] (0,0.5) -- (0.167,0);
\draw [teal] (0.167,1) -- (0.5,0);
\draw [teal] (0.5,1) -- (0.833,0);
\draw [teal] (0.833,1) -- (1,0.5);
\draw [thick,blue] (0,0) -- (1,0);
\draw [thick,blue] (0,0.5) -- (1,0.5);
\draw [thick,blue] (0,0) -- (0,1);
\draw [thick,blue] (0.5,0) -- (0.5,1);
\fill (0,0) node[circle, fill=red, draw=red, inner sep=1.5pt] {};
\fill (0.5,0) node[circle, fill=red, draw=red, inner sep=1.5pt] {};
\fill (0,0.5) node[circle, fill=red, draw=red, inner sep=1.5pt] {};
\fill (0.5,0.5) node[circle, fill=red, draw=red, inner sep=1.5pt] {};
\draw (0.5,-0.15) node {$\frac{1}{2}$};
\draw (1.2,0) node {$\alpha$};
\draw (0,1.2) node {$\beta$};
\draw (-0.1,0.5) node {$\frac{1}{2}$};
\draw (-0.1, 1) node {$1$};
\draw (-0.1,-0.15) node {$0$};
\draw (0.167,-0.15) node {$\frac{1}{6}$};
\draw (0.833,-0.15) node {$\frac{5}{6}$};
\draw (1,-0.15) node {$1$};
\end{tikzpicture}
\caption{The four cases in the analysis of the exceptional set for spectral decimation of $\mathcal{L}_N^\omega$, indicated in the flux parameter space $(\alpha,\beta) \in [0,1)^2$: {\color{red} Case I}, {\color{blue} Case II}, {\color{teal} Case III}, and Case IV (the white region in $[0,1)^2$ that is not colored).}
\label{fig:param}
\end{figure}


\subsection{Spectrum under fluxes $\alpha,\beta\in\{0,\frac{1}{2}\}$} 
\label{sec:halfflux}

In this subsection we solve $\sigma(\mathcal{L}^\omega_N)$ in \underline{Case I}, thereby proving Theorem \ref{thm:1}.
Here we use Corollary \ref{cor:specdec}-\ref{item:R}, where all the functions are explicit polynomial or rational functions.
This allows us to carry out spectral decimation of $\mathcal{L}^\omega_N$ all the way to $\mathcal{L}^{(0,0)}_0$, as Figure \ref{fig:SDdiagram} indicates.

\subsubsection{$\alpha=\beta=0$}
\label{sec:zeroflux}

This case corresponds to spectral decimation of the usual graph Laplacian, which has been solved in \cite[Section 3]{Teplyaev} (see also \cite[Section 5]{3n-gasket}). We include the full analysis here since it will be referred to in the other cases below.
The relevant functions are
\begin{align*}
&\Psi(0,0,\lambda)=\left(\lambda-\frac{3}{2}\right)\left(\lambda-\frac{5}{4}\right), \quad
\mathcal{D}(0,\lambda)=-\frac{1}{32}(2\lambda-1)(4\lambda-5)^2,\\
&\phi(0,0,\lambda)=-\frac{\lambda-\frac{3}{2}}{(\lambda-\frac{1}{2})(4\lambda-5)},\quad
R(0,0,\lambda)=\lambda(5-4\lambda),\quad
\mathcal{E}(0,0)=\left\{\frac{1}{2},\frac{5}{4},\frac{3}{2}\right\}.
\end{align*}

The enumeration of $\sigma(\mathcal{L}_N^{(0,0)})$ proceeds in 5 steps.
The first three steps are obtained by first observing that 1 is an eigenvalue on level 0, and then using Lemma \ref{specdecthm} and the appropriate cases in Lemma \ref{multiplicity}. For the last two steps, we observe that the eigenvalues $\frac{3}{4}$ and $\frac{5}{4}$ first appear on level 1 and level 2, respectively, which then lead to their corresponding $\frac{3}{4}$-series and $\frac{5}{4}$-series. 
To wit:

\begin{enumerate}[wide]
    \item $\textup{mult}\left(\mathcal{L}_N^{(0,0)},0\right)=1$ by induction on $N$ and Lemma \ref{specdecthm}. This result is consistent with the Perron-Frobenius theorem.
    \item $\frac{3}{2} \in \mathcal{E}(0,0)$, and $\displaystyle \textup{mult}\left(\mathcal{L}_N^{(0,0)},\frac{3}{2}\right)=\frac{3^N+3}{2}$ by Lemma \ref{multiplicity}-\ref{eq:SD2}.
    \item $\frac{1}{2}\in \mathcal{E}(0,0)$, and $\displaystyle \textup{mult}\left(\mathcal{L}_N^{(0,0)},\frac{1}{2}\right)=0$ by Lemma \ref{multiplicity}-\ref{eq:SD3}.
    \item $\frac{3}{4}$-series, which contains $\frac{3}{4}$ and any number of preimages of $\frac{3}{4}$ under $R(0,0,\cdot)$:
    
    For $k\in \{0,1,\cdots, N-1\}$, $\displaystyle \textup{mult}\left(\mathcal{L}_N^{(0,0)},\left(R(0,0,\cdot)\right)^{-k}\left(\frac{3}{4}\right)\right)=\frac{3^{N-k-1}+3}{2}$ by Lemma \ref{specdecthm} and Lemma \ref{multiplicity}-\ref{eq:SD3}.
    \item $\frac{5}{4}$-series, which contains $\frac{5}{4}$ and any number of preimages of $\frac{5}{4}$ under $R(0,0,\cdot)$:
    
    For $k\in \{0,1,\cdots, N-2\}$, $\displaystyle \textup{mult}\left(\mathcal{L}_N^{(0,0)},\left(R(0,0,\cdot)\right)^{-k}\left(\frac{5}{4}\right)\right)=\frac{3^{N-k-1}-1}{2}$ by Lemma \ref{specdecthm} and Lemma \ref{multiplicity}-\ref{eq:SD3}.
    Note that $\displaystyle \mult\left(\mathcal{L}_N^{(0,0)}, \left(R(0,0,\cdot)\right)^{-(N-1)}\left(\frac{5}{4}\right)\right)=0$.
\end{enumerate}

As an aside, note that the preimage of $\mathbb{R}_+$ under $R(0,0,\cdot)$ is $[0,\frac{5}{4}]$. 
It is then direct to see that $\sigma(\mathcal{L}_N^{(0,0)})$ is contained in $[0,\frac{5}{4}]\cup\{\frac{3}{2}\}$. 

\subsubsection{$\alpha=\beta=\frac{1}{2}$}
\label{sec:halfhalf}

The relevant functions are
\begin{align*}
&\Psi\left(\frac{1}{2},\frac{1}{2},\lambda\right)=\left(\lambda-\frac{1}{2}\right)\left(\lambda-\frac{3}{4}\right), \quad
\mathcal{D}\left(\frac{1}{2},\lambda\right)=-\frac{1}{32}(2\lambda-3)(4\lambda-3)^2,\\
&\phi\left(\frac{1}{2},\frac{1}{2},\lambda\right)=-\frac{\lambda-\frac{1}{2}}{(\lambda-\frac{3}{2})(4\lambda-3)}, \quad
R\left(\frac{1}{2},\frac{1}{2},\lambda\right)=-(\lambda-2)(4\lambda-3),\quad
\mathcal{E}\left(\frac{1}{2},\frac{1}{2}\right)=\left\{\frac{1}{2},\frac{3}{4},\frac{3}{2}\right\}.
\end{align*}
Since $\alpha_\downarrow(\frac{1}{2},\frac{1}{2},\lambda)=0$ and $\beta_\downarrow(\frac{1}{2},\frac{1}{2},\lambda)=0$, many of the eigenvalues in $\sigma(\mathcal{L}_N^{(\frac{1}{2}, \frac{1}{2})})$ are preimages of $\sigma(\mathcal{L}_{N_1}^{(0,0)})$ under $R(\frac{1}{2},\frac{1}{2},\cdot)$.
The enumeration of $\sigma(\mathcal{L}^{(\frac{1}{2},\frac{1}{2})}_N)$ proceeds in 6 steps.
\begin{enumerate}[wide]
    \item Since $0\in \sigma(\mathcal{L}_{N-1}^{(0,0)})$, we consider its preimages under $R(\frac{1}{2},\frac{1}{2},\cdot)$, which are $2$ and $\frac{3}{4}$ (exceptional). Since $\phi(\frac{1}{2},\frac{1}{2},2)\neq 0$ and $\mathcal{D}(\frac{1}{2}, 2)\neq 0$, we can apply Lemma \ref{specdecthm} to get
    \[
    \textup{mult}\left(\mathcal{L}_N^{\left(\frac{1}{2},\frac{1}{2}\right)},2\right)=\textup{mult}\left(\mathcal{L}_{N-1}^{(0,0)},0\right)=1.
    \]
\end{enumerate}
Next we treat each of the three exceptional values.
\begin{enumerate}[resume,wide]
    \item We see that $\mathcal{D}(\frac{1}{2}, \frac{1}{2})\neq 0$ and $\phi\left(\frac{1}{2},\frac{1}{2},\frac{1}{2}\right)=0$. So Lemma \ref{multiplicity}-\ref{eq:SD2} applies and
    \[
     \textup{mult}\left(\mathcal{L}_N^{\left(\frac{1}{2},\frac{1}{2}\right)},\frac{1}{2}\right)=\frac{3^{N}+3}{2}.
    \]
    \item Since $\mathcal{D}(\frac{1}{2}, \frac{3}{4})=0$, $\phi(\frac{1}{2}, \frac{1}{2},\cdot)$ has a pole at $\frac{3}{4}$, and $R\left(\frac{1}{2},\frac{1}{2},\frac{3}{4}\right)=0$, Lemma \ref{multiplicity}-\ref{eq:SD3} applies and
    \[
    \textup{mult}\left(\mathcal{L}_N^{\left(\frac{1}{2},\frac{1}{2}\right)},\frac{3}{4}\right)=3^{N-1}\cdot 2-\frac{3^N+3}{2}+\textup{mult}\left(\mathcal{L}_{N-1}^{(0,0)},0\right)=\frac{3^{N-1}-1}{2}.
    \]
    \item The last exceptional value $\frac{3}{2}$ satisfies the same conditions as $\frac{3}{4}$, and $R\left(\frac{1}{2},\frac{1}{2},\frac{3}{2}\right)=\frac{3}{2}$, so Lemma \ref{multiplicity}-\ref{eq:SD3} applies and
    \[
    \textup{mult}\left(\mathcal{L}_N^{\left(\frac{1}{2},\frac{1}{2}\right)},\frac{3}{2}\right)=3^{N-1}-\frac{3^N+3}{2}+\textup{mult}\left(\mathcal{L}_{N-1}^{(0,0)},\frac{3}{2}\right)=0.
    \]
    \item We need to consider the other preimage of $\frac{3}{2}$ under $R\left(\frac{1}{2},\frac{1}{2},\cdot\right)$, which is $\frac{5}{4}$. 
    Given that $\mathcal{D}(\frac{1}{2}, \frac{5}{4})\neq 0$ and $\phi(\frac{1}{2}, \frac{1}{2}, \frac{5}{4})\neq 0$, Lemma \ref{specdecthm} applies and
    \[
    \textup{mult}\left(\mathcal{L}_N^{\left(\frac{1}{2},\frac{1}{2}\right)},\frac{5}{4}\right)=\textup{mult}\left(\mathcal{L}_{N-1}^{(0,0)},\frac{3}{2}\right)=\frac{3^{N-1}+3}{2}.
    \]
    \item From the previous case $\alpha=\beta=0$, we saw that all the eigenvalues in $\sigma\left(\mathcal{L}_{N-1}^{(0,0)}\right) \setminus \{0,\frac{3}{2}\}$ belong to the $\frac{3}{4}$-series or the $\frac{5}{4}$-series, and lie in $\left[0,\frac{3}{2}\right]$. 
Since $R(\frac{1}{2},\frac{1}{2},\lambda)=-4(\lambda-\frac{11}{8})^2+\frac{25}{16}$, each eigenvalue of the $\frac{3}{4}$-series and the $\frac{5}{4}$-series has two positive real preimages under $R(\frac{1}{2},\frac{1}{2},\cdot)$. So by Lemma \ref{specdecthm} we have
\begin{align*}
\textup{mult}\left(\mathcal{L}_N^{(\frac{1}{2},\frac{1}{2})},\left(R\left(\frac{1}{2},\frac{1}{2},\cdot\right)\right)^{-1} \circ \left(R(0,0,\cdot)\right)^{-k}\left(\frac{3}{4}\right)\right)=\frac{3^{N-k-2}+3}{2},&\quad  k\in \{0,1,\cdots, N-2\},\\
\textup{mult}\left(\mathcal{L}_N^{(\frac{1}{2},\frac{1}{2})},\left(R\left(\frac{1}{2},\frac{1}{2},\cdot\right)\right)^{-1} \circ \left(R(0,0,\cdot)\right)^{-k}\left(\frac{5}{4}\right)\right)=\frac{3^{N-k-2}-1}{2},&\quad k\in \{0,1,\cdots, N-3\}.
\end{align*}
These give rise to
    \[
    2\left(\dim_{N-1}-\textup{mult}\left(\mathcal{L}_{N-1}^{(0,0)},\frac{3}{2}\right)-\textup{mult}\left(\mathcal{L}_{N-1}^{(0,0)},0\right)\right)=2\left(\dim_{N-1}-\frac{3^{N-1}+3}{2}-1\right)=2(3^{N-1}-1)
    \]
  many eigenvalues.
\end{enumerate}

The total count of eigenvalues from the 6 steps above is
\[
1+\frac{3^N+3}{2}+\frac{3^{N-1}-1}{2}+0+\frac{3^{N-1}+3}{2}+2 (3^{N-1}-1)=\frac{3^{N+1}+3}{2}=\dim_N,
\]
as desired. 

As an aside, note that the eigenvalues in the first 5 items do not fall into the interval $(\frac{1}{2},\frac{3}{4})$. Also, the eigenvalues in Step (6), which are preimages of $R(\frac{1}{2},\frac{1}{2},\cdot)$, are in $[\frac{3}{4},2]$, also outside of $(\frac{1}{2},\frac{3}{4})$. This is consistent with the gap $(\frac{1}{2},\frac{3}{4})$ in the butterfly spectrum (Figure \ref{fig:newspectrum}).

\subsubsection{$\alpha=\frac{1}{2}, \beta=0$}
\label{sec:halfzero}

The relevant functions are
\begin{align*}
&\Psi\left(\frac{1}{2},0,\lambda\right)=\left(\lambda-\frac{1}{2}\right)\left(\lambda-\frac{5}{4}\right),\quad
\mathcal{D}(0,\lambda)=-\frac{1}{32}(2\lambda-1)(4\lambda-5)^2,
\\
&
\phi\left(\frac{1}{2},0,\lambda\right)=-\frac{1}{4\lambda-5},\quad
R\left(\frac{1}{2},0,\lambda\right)=-4\lambda^2+9\lambda-3, \quad
\mathcal{E}\left(\frac{1}{2},0\right)=\left\{\frac{1}{2},\frac{5}{4}\right\}.
\end{align*}
Since $\alpha_\downarrow(\frac{1}{2},0,\lambda)=\frac{1}{2}$ and $\beta_\downarrow(\frac{1}{2},0,\lambda)=\frac{1}{2}$, many of the eigenvalues in $\sigma(\mathcal{L}^{(\frac{1}{2},0)}_N)$ are preimages of $\sigma(\mathcal{L}^{(\frac{1}{2}, \frac{1}{2})}_{N-1})$ under $R(\frac{1}{2},0,\cdot)$.
We enumerate $\sigma(\mathcal{L}_N^{(\frac{1}{2},0)})$ in five steps.

First consider the two values in the exceptional set.
\begin{enumerate}[wide]
    \item The first exceptional value $\frac{1}{2}$ satisfies $\mathcal{D}(0,\frac{1}{2})=0$ and $\phi(\frac{1}{2},0,\frac{1}{2})\neq 0$. Moreover, both $\phi$ and $\phi R$ are bounded in a neighborhood of $\frac{1}{2}$, and $\partial_\lambda R(\frac{1}{2},0,\frac{1}{2}) \neq 0$. 
    Also $R(\frac{1}{2},0,\frac{1}{2})=\frac{1}{2}$.
    Therefore Lemma \ref{multiplicity}-\ref{eq:SD4} applies and
    \[
    \textup{mult}\left(\mathcal{L}_N^{\left(\frac{1}{2},0\right)},\frac{1}{2}\right)=3^{N-1}+ \textup{mult}\left(\mathcal{L}_{N-1}^{\left(\frac{1}{2},\frac{1}{2}\right)},\frac{1}{2}\right)=3^{N-1}+\frac{3^{N-1}+3}{2}=\frac{3^N+3}{2}.
    \]
    \item The second exceptional value $\frac{5}{4}$ satisfies all the conditions in Lemma \ref{multiplicity}-\ref{eq:SD3}. We also know that $R\left(\frac{1}{2},0,\frac{5}{4}\right)=2$. So we get
    \[
    \textup{mult}\left(\mathcal{L}_N^{\left(\frac{1}{2},0\right)},\frac{5}{4}\right)=3^{N-1}\cdot 2-\frac{3^N+3}{2}+\textup{mult}\left(\mathcal{L}_{N-1}^{\left(\frac{1}{2},\frac{1}{2}\right)},2\right)=\frac{3^{N-1}-1}{2}.
    \]
\end{enumerate}

The next two steps deal with the other preimages of $\frac{1}{2}$ and $2$, which appeared in the first two steps and belong to $\sigma(\mathcal{L}_{N-1}^{\left(\frac{1}{2},\frac{1}{2}\right)})$. 
\begin{enumerate}[resume, wide]
    \item $\frac{7}{4}$ is the other preimage of $\frac{1}{2}$ under $R\left(\frac{1}{2},0,\cdot\right)$. Lemma \ref{specdecthm} applies and
    \[
    \textup{mult}\left(\mathcal{L}_N^{\left(\frac{1}{2},0\right)},\frac{7}{4}\right)=\textup{mult}\left(\mathcal{L}_{N-1}^{\left(\frac{1}{2},\frac{1}{2}\right)},\frac{1}{2}\right)=\frac{3^{N-1}+3}{2}.
    \]
    \item $1$ is the other preimage of $2$ under $R\left(\frac{1}{2},0,\cdot\right)$. Lemma \ref{specdecthm} applies and
    \[
    \textup{mult}\left(\mathcal{L}_N^{\left(\frac{1}{2},0\right)},1\right)=\textup{mult}\left(\mathcal{L}_{N-1}^{\left(\frac{1}{2},\frac{1}{2}\right)},2\right)=1.
    \]
    \item Finally, we apply Lemma \ref{specdecthm} to all eigenvalues of $\sigma(\mathcal{L}^{(\frac{1}{2}, \frac{1}{2})}_{N-1})\setminus\{\frac{1}{2},2\}$. First, we would like to investigate $\frac{3}{4}$ and $\frac{3}{2}$, which are the remaining two values in $\mathcal{E}(\frac{1}{2}, \frac{1}{2})$. Since $\frac{3}{2} \notin \sigma(\mathcal{L}^{(\frac{1}{2}, \frac{1}{2})}_{N-1})$, we can ignore it. As for $\frac{3}{4}$, its preimages do not lie in $\mathcal{E}\left(\frac{1}{2},0\right)$, so Lemma \ref{specdecthm} applies. Moreover, all eigenvalues in $\sigma(\mathcal{L}^{(\frac{1}{2}, \frac{1}{2})}_{N-1})\setminus\{\frac{1}{2},2\}$ are in $[\frac{3}{4},2]$ and $R(\frac{1}{2},0,\lambda)=-4(\lambda-\frac{9}{8})^2+\frac{33}{16}$, so Lemma \ref{specdecthm} applies and they all have two positive real preimages under $R(\frac{1}{2},0,\cdot)$. So in total they contribute to $\sigma(\mathcal{L}_N^{(\frac{1}{2},0)})$
\[2\left(\dim_{N-1}-\textup{mult}\left(\mathcal{L}_{N-1}^{\left(\frac{1}{2},\frac{1}{2}\right)},2\right)-\textup{mult}\left(\mathcal{L}_{N-1}^{\left(\frac{1}{2},\frac{1}{2}\right)},\frac{1}{2}\right)\right)=\frac{2\cdot 3^N-2\cdot 3^{N-1}-4}{2}
\]
eigenvalues.
\end{enumerate}

The total count of eigenvalues is indeed
\begin{equation*}
\frac{3^N+3}{2}+\frac{3^{N-1}-1}{2}+\frac{3^{N-1}+3}{2}+1+\frac{2\cdot 3^N-2\cdot 3^{N-1}-4}{2}=\frac{3^{N+1}+3}{2}=\dim_N.
\end{equation*}

\subsubsection{$\alpha=0, \beta=\frac{1}{2}$}
\label{sec:zerohalf}

The relevant functions are
\begin{align*}
&\Psi\left(0,\frac{1}{2},\lambda\right)=\left(\lambda-\frac{3}{4}\right)\left(\lambda-\frac{3}{2}\right),\quad
\mathcal{D}\left(\frac{1}{2},\lambda\right)=-\frac{1}{32}(2\lambda-3)(4\lambda-3)^2,\\
&\phi\left(0,\frac{1}{2},\lambda\right)=-\frac{1}{4\lambda-3}, \quad
R\left(0,\frac{1}{2},\lambda\right)=-4\lambda^2+7\lambda-1,\quad
\mathcal{E}\left(0,\frac{1}{2}\right)=\left\{\frac{3}{4},\frac{3}{2}\right\}.
\end{align*}
Since $\alpha_\downarrow(0,\frac{1}{2},\lambda)=\frac{1}{2}$ and $\beta_\downarrow(0,\frac{1}{2},\lambda)=\frac{1}{2}$, most of the eigenvalues in $\sigma(\mathcal{L}^{(0,\frac{1}{2})}_N)$ are preimages of $\sigma(\mathcal{L}^{(\frac{1}{2}, \frac{1}{2})}_{N-1})$ under $R(0,\frac{1}{2},\cdot)$.
The following five-step analysis is similar to that of the last case $\alpha=\frac{1}{2}, \beta=0$.

We consider the two values in the exceptional set first.
\begin{enumerate}[wide]
    \item The first value $\frac{3}{2}$ is in $\sigma(D)$, and $\phi\left(0,\frac{1}{2},\lambda\right)$ is neither 0 at $\frac{3}{2}$ nor has a pole at $\frac{3}{2}$. We also know that $R\left(0,\frac{1}{2},\frac{3}{2}\right)=\frac{1}{2}$. Therefore, we shall apply Lemma \ref{multiplicity}-\ref{eq:SD4} to get
    \[
    \textup{mult}\left(\mathcal{L}_N^{\left(0,\frac{1}{2}\right)},\frac{3}{2}\right)=3^{N-1}+ \textup{mult}\left(\mathcal{L}_{N-1}^{\left(\frac{1}{2},\frac{1}{2}\right)},\frac{1}{2}\right)=3^{N-1}+\frac{3^{N-1}+3}{2}=\frac{3^N+3}{2}.
    \]
    \item The second value in the exceptional set $\frac{3}{4}$ is also in $\sigma(D)$, and satisfies the rest of the conditions in Lemma \ref{multiplicity}-\ref{eq:SD3}. We also know that $R\left(0,\frac{1}{2},\frac{3}{4}\right)=2$. So we get
    \[
    \textup{mult}\left(\mathcal{L}_N^{\left(0,\frac{1}{2}\right)},\frac{3}{4}\right)=3^{N-1}\cdot 2-\frac{3^N+3}{2}+\textup{mult}\left(\mathcal{L}_{N-1}^{\left(\frac{1}{2},\frac{1}{2}\right)},2\right)=\frac{3^{N-1}-1}{2}.
    \]
\end{enumerate}

The next two items deal with the other preimages of $\frac{1}{2}$ and $2$ under $R\left(0,\frac{1}{2},\cdot\right)$.

\begin{enumerate}[resume,wide]
    \item $\frac{1}{4}$ is the other preimage of $\frac{1}{2}$ under $R\left(\frac{1}{2},0,\cdot\right)$, so we use Lemma \ref{specdecthm} to get
    \[
    \textup{mult}\left(\mathcal{L}_N^{\left(0,\frac{1}{2}\right)},\frac{1}{4}\right)=\textup{mult}\left(\mathcal{L}_{N-1}^{\left(\frac{1}{2},\frac{1}{2}\right)},\frac{1}{2}\right)=\frac{3^{N-1}+3}{2}.
    \]
    \item $1$ is the other preimage of $2$ under $R\left(\frac{1}{2},0,\cdot\right)$, so we use Lemma \ref{specdecthm} to get
    \[
    \textup{mult}\left(\mathcal{L}_N^{\left(\frac{1}{2},0\right)},1\right)=\textup{mult}\left(\mathcal{L}_{N-1}^{\left(\frac{1}{2},\frac{1}{2}\right)},2\right)=1.
    \]
    \item Finally, by the same argument in step (5) of the previous case, all eigenvalues in $\sigma(\mathcal{L}^{(\frac{1}{2}, \frac{1}{2})}_{N-1})\setminus\{\frac{1}{2},2\}$ are in $[\frac{3}{4},2]$ and $R(0,\frac{1}{2},\lambda)=-4(\lambda-\frac{7}{8})^2+\frac{33}{16}$, so Lemma \ref{specdecthm} applies and they all have two positive real preimages under $R(0,\frac{1}{2},\cdot)$. So in total they contribute to $\sigma(\mathcal{L}_N^{(0,\frac{1}{2})})$
\[
\frac{2\cdot 3^N-2\cdot 3^{N-1}-4}{2}
\]
eigenvalues.
\end{enumerate}

The total count of eigenvalues is indeed
\begin{equation*}
\frac{3^N+3}{2}+\frac{3^{N-1}-1}{2}+\frac{3^{N-1}+3}{2}+1+\frac{2\cdot 3^N-2\cdot 3^{N-1}-4}{2}=\frac{3^{N+1}+3}{2}=\dim_N.
\end{equation*}


\subsection{Spectrum when not both of the fluxes $\alpha,\beta$ are in $\{0,\frac{1}{2}\}$} \label{sec:nonhalfintflux}

In this subsection we characterize $\sigma(\mathcal{L}^\omega_N)$ in \underline{Cases II, III, and IV}, thereby proving Theorem \ref{thm:2}.
Recall from Proposition \ref{prop:Rvalued} that $\mathbb{R}\ni\lambda\mapsto \Psi(\alpha,\beta,\lambda)$ is $\mathbb{C}$-valued, so we use Corollary \ref{cor:specdec}-\ref{item:C}.
In particular the reduced magnetic Laplacian $\mathcal{L}^\Omega_{N-1}$ receives a ``twist'' in the form of a multiplier $e^{2\pi i \theta(\alpha,\beta,\lambda)}$, $\theta(\alpha,\beta,\lambda)=(2\pi)^{-1} {\rm arg}\Psi(\alpha,\beta,\lambda)$.
The decimation diagram takes the form
\[
\begin{tikzcd}[column sep=huge]
\mathcal{L}_N^{(\alpha_N, \beta_N)} \arrow[r, "{R(\alpha_N, \beta_N, \cdot)}"] & 
\mathcal{L}_{N-1}^{(\alpha_{N-1}, \beta_{N-1})} 
\longrightarrow \cdots \longrightarrow
\mathcal{L}_0^{(\alpha_0, \beta_0)},
\end{tikzcd}
\]
where for each $n\in \{1,2,\cdots, N\}$, $\alpha_{n-1}(\lambda) = \alpha_\downarrow(\alpha_n, \beta_n,\lambda)$ and $\beta_{n-1}(\lambda)=\beta_\downarrow(\alpha_n,\beta_n,\lambda)$, \emph{cf.\@} Proposition \ref{prop:fluxevolve}.
We emphasize again the dependence of the magnetic Laplacians and fluxes on the spectral parameter $\lambda$ under decimation.
That said, to avoid an overcharged notation, we will suppress the flux symbols $\alpha_N$ and $\beta_N$ in this subsection unless the context requires their presence.

The order of our analysis starts with the case $\Psi(\lambda)\neq 0$, followed by the case $\Psi(\lambda)=0$.

\subsubsection{$\Psi(\lambda)\neq 0$}


\begin{proposition}
\label{prop:31}
In any of \underline{Cases II, III, and IV}, suppose $\Psi(\lambda)\neq 0$.
\begin{enumerate}[wide, label=(G\arabic*)]
\item \label{G1} If $\mathcal{D}(\lambda)\neq 0$, then
\begin{align}
\label{eq:G1}
E_\lambda(\mathcal{L}_N^\omega) = \left(\lim_{\mathbb{R} \ni x\to\lambda} \phi(x) \frac{R(\lambda)-R(x)}{\lambda-x}\right)^{-1} \left(P^*_\parallel - P^*_\perp(D-\lambda)^{-1}C\right) E_{R(\lambda)}(\mathcal{L}^\Omega_{N-1}) \left(P_\parallel - B(D-\lambda)^{-1} P_\perp\right).
\end{align}
In particular,
$
{\rm mult}(\mathcal{L}^\omega_N, \lambda) = {\rm mult}\left(\mathcal{L}^\Omega_{N-1}, R(\lambda)\right)
$
.
\item \label{G2} If $\mathcal{D}(\lambda)=0$, then $\lambda$ is a simple zero of $\mathcal{D}$.

On the one hand, suppose $\displaystyle \lim_{\mathbb{R}\ni x\to\lambda}\frac{1}{|\Psi(x)|} \frac{\mathcal{D}(x)(\lambda-x)}{R(\lambda)-R(x)} =0$.
Then
\begin{equation*}
\begin{aligned}
E&_\lambda(\mathcal{L}_N^\omega) = 
P^*_\perp E_\lambda(D) P_\perp \\
&+ 
\left(\lim_{\mathbb{R}\ni x\to \lambda}\frac{4\mathcal{D}(x)}{(\lambda-x)|\Psi(x)|}\right)P_\perp^* E_\lambda(D) C 
\left(\mathcal{L}^\Omega_{N-1}- R(\lambda)\right)^{-1} 
\left(I_{\parallel} - E_{R(\lambda)}(\mathcal{L}^\Omega_{N-1})\right) 
BE_\lambda(D) P_\perp.
\end{aligned}
\end{equation*}
In particular, $E_\lambda(\mathcal{L}_N^\omega) (P^*_\perp E_\lambda(D) P_\perp) = E_\lambda(\mathcal{L}_N^\omega)$,
$$
{\rm mult}(\mathcal{L}^\omega_N, \lambda) =
3^{N-1} - {\rm dim}_{N-1} + {\rm mult}\left(\mathcal{L}^\Omega_{N-1}, R(\lambda)\right) = -\frac{3^{N-1}+3}{2}+ {\rm mult}\left(\mathcal{L}^\Omega_{N-1}, R(\lambda)\right),
$$
and the corresponding eigenfunction vanishes on $V_{N-1}$.

On the other hand, suppose $\displaystyle \lim_{\mathbb{R}\ni x\to\lambda}\frac{1}{|\Psi(x)|} \frac{\mathcal{D}(x)(\lambda-x)}{R(\lambda)-R(x)} \neq 0$.
Then
\begin{equation}
\label{eq:G2eigp2}
\begin{aligned}
E&_\lambda(\mathcal{L}_N^\omega) = 
P^*_\perp E_\lambda(D) P_\perp \\
&+\left(\lim_{\mathbb{R}\ni x\to \lambda}\frac{4\mathcal{D}(x)}{(\lambda-x)|\Psi(x)|}\right)P_\perp^* E_\lambda(D) C 
\left(\mathcal{L}^\Omega_{N-1}- R(\lambda)\right)^{-1} 
\left(I_{\parallel} - E_{R(\lambda)}(\mathcal{L}^\Omega_{N-1})\right) 
BE_\lambda(D) P_\perp\\
&+ \left(P_{\parallel}^*-P_{\perp}^*(D-\lambda)^{-1}C\right) \left(\lim_{\mathbb{R} \ni x\to\lambda}\frac{1}{\phi(x)}\frac{\lambda-x}{R(\lambda)-R(x)}\right)
 E_{R(\lambda)}(\mathcal{L}^\Omega_{N-1})
\left(P_{\parallel}-B(D-\lambda)^{-1}P_{\perp}\right),
\end{aligned}
\end{equation}
and in general 
$$\mult(\mathcal{L}^\omega_N,\lambda) = 3^{N-1} - {\rm dim}_{N-1} + 2\mult(\mathcal{L}^\Omega_{N-1}, R(\lambda)) = -\frac{3^{N-1}+3}{2} + 2\mult(\mathcal{L}^\Omega_{N-1}, R(\lambda)).
$$
That said, if none of the corresponding eigenfunctions vanishes on $V_{N-1}$, then the first two terms on the RHS of \eqref{eq:G2eigp2} vanish, and $\mult(\mathcal{L}^\omega_N,\lambda) = \mult(\mathcal{L}^\Omega_{N-1}, R(\lambda))$.
\end{enumerate}
\end{proposition}

\begin{proof}
\ref{G1}: We are in the setting of Lemma \ref{specdecthm}.
The only thing to justify is the existence of the limit on the RHS of \eqref{eq:G1}.
Indeed, from \eqref{eq:abs_phi} and \eqref{eq:abs_R} we get
\begin{align*}
\phi(x)& \frac{R(\lambda)-R(x)}{\lambda-x}= \frac{|\Psi(x)|}{4\mathcal{D}(x)}\frac{1}{\lambda-x} \left(\frac{A(\lambda)-64\mathcal{D}(\lambda)(1-\lambda)}{16|\Psi(\lambda)|}-\frac{A(x)-64\mathcal{D}(x)(1-x)}{16|\Psi(x)|}\right)\\
&= \frac{1}{64\mathcal{D}(x)}\frac{1}{\lambda-x}\left(\frac{|\Psi(x)|}{|\Psi(\lambda)|}\left(A(\lambda)-64\mathcal{D}(\lambda)(1-\lambda)\right) -  \left(A(x)-64\mathcal{D}(x)(1-x)\right)\right)\\
&= \frac{1}{64\mathcal{D}(x)}\left( \frac{F(\lambda)-F(x)}{\lambda-x} + \frac{|\Psi(x)|-|\Psi(\lambda)|}{\lambda-x}\cdot \frac{F(\lambda)}{|\Psi(\lambda)|}\right) 
\end{align*}
where $F(x) = A(x)-64\mathcal{D}(x)(1-x)$, which is analytic.
On the other hand, $|\Psi(x)|^2$ is a quadratic polynomial in $x\in \mathbb{R}$. Therefore $|\Psi(x)| = \sqrt{|\Psi(x)|^2}$ is differentiable on $\mathbb{R}$ so long as $\Psi(x) \neq 0$.
The claim now follows from the given assumptions.

\ref{G2}: By Lemma \ref{lem:multzero}, a multiple zero $\lambda$ of $\mathcal{D}$ occurs only if $\beta=0$ or $\beta=\frac{1}{2}$, \emph{i.e.,} under \underline{Case II}. 
Moreover, by Table \ref{tab:freepara}, $\lambda$ is also a zero of $\Psi$.
Therefore under the stated conditions $\lambda$ can only be a simple zero of $\mathcal{D}$.

Since $\Psi$ is continuous, $\Psi(\lambda)\neq 0$, and $\mathcal{D}(\lambda)= 0$, it follows that $\lim_{\mathbb{R}\ni x\to\lambda}[\phi(x)]^{-1}=0$.
Thus we are in the setting of either Lemma \ref{multiplicity}-\ref{eq:SD3} or Lemma \ref{multiplicity}-\ref{eq:SD6}, provided that the following two limits exist:
\begin{align}
\label{eq:twolimits}
\lim_{\mathbb{R}\ni x\to\lambda} \frac{1}{\phi(x)(\lambda-x)}
 \quad \text{ and } \quad
\lim_{\mathbb{R}\ni x\to\lambda} \frac{1}{\phi(x)}\frac{\lambda-x}{R(\lambda)-R(x)}.
\end{align}
For the first ratio in \eqref{eq:twolimits}, since $\lambda$ is a simple zero of $\mathcal{D}$, and $\Psi(\lambda)\neq 0$,
\[
\frac{1}{\phi(x)(\lambda-x)} = \frac{4\mathcal{D}(x)}{|\Psi(x)|(\lambda-x)} = -\frac{4}{|\Psi(x)|} (x-r_1)(x-r_2) \quad \text{for some } r_1, r_2 \neq \lambda.
\]
This has a well-defined nonzero limit as $x\to\lambda$.
For the second ratio in \eqref{eq:twolimits},
\begin{align}
\label{eq:secondratio}
\frac{1}{\phi(x)}\frac{\lambda-x}{R(\lambda)-R(x)}
= \frac{4}{|\Psi(x)|} \frac{\mathcal{D}(x)(\lambda-x)}{R(\lambda)-R(x)},
\end{align}
the existence of the limit as $x\to\lambda$ is clear. 
If this limit is zero (resp.\@ nonzero), Lemma \ref{multiplicity}-\ref{eq:SD3} (resp.\@ Lemma \ref{multiplicity}-\ref{eq:SD6}) applies.
\end{proof}

Given our knowledge of the functions $\mathcal{D}$, $\Psi$, and $A$, it would be more satisfying to give concrete criteria for whether the limit of \eqref{eq:secondratio} is zero.
Below is our best attempt using elementary analysis.

By assumption we may write $\mathcal{D}(x) = -(x-\lambda)(x-a)(x-b)$, $a,b \neq \lambda$ being the two other zeros of $\mathcal{D}$.
Also, since $\mathcal{D}(\lambda)=0$ and $\Psi(\lambda)\neq 0$, 
\begin{align*}
R&(\lambda)-R(x) = \frac{A(\lambda)}{16|\Psi(\lambda)|} - \frac{A(x)-64\mathcal{D}(x)(1-x)}{16|\Psi(x)|}\\
&=\frac{1}{16|\Psi(x)|}\left(A(\lambda)-A(x) +64\mathcal{D}(x)(1-x)\right) + \frac{A(\lambda)}{16}\left(\frac{1}{|\Psi(\lambda)|}-\frac{1}{|\Psi(x)|}\right).
\end{align*}
Therefore \eqref{eq:secondratio} rewrites as
\begin{align}
\label{eq:secondratio2}
 64(x-a)(x-b) \left[ \frac{(\lambda-x)^2}{A(\lambda)-A(x)+64\mathcal{D}(x)(1-x)} +
  \frac{|\Psi(\lambda)|}{A(\lambda)} \frac{(\lambda-x)^2}{|\Psi(x)|- |\Psi(\lambda)|}  \right].
\end{align}
assuming $A(\lambda)\neq 0$. 
(If $A(\lambda)=0$, then only the first term inside the square bracket in \eqref{eq:secondratio2} survives.)

Let us note the elementary identity
\[
\frac{1}{|\Psi(x)|-|\Psi(\lambda)|} = \frac{|\Psi(x)|+|\Psi(\lambda)|}{|\Psi(x)|^2-|\Psi(\lambda)|^2}.
\]
By Taylor approximation, $|\Psi(\cdot)|^2- |\Psi(\lambda)|^2$ has a multiple zero at $\lambda$ if and only if $\frac{d}{dx}|\Psi(x)|^2|_{x=\lambda}=0$.
Therefore the second term in the square bracket in \eqref{eq:secondratio2} converges to $0$ as $\mathbb{R}\ni x\to\lambda$ if and only if $\frac{d}{dx}|\Psi(x)|^2|_{x=\lambda}\neq 0$.

The same reasoning applies to the first term in the square bracket in \eqref{eq:secondratio2}.
By construction, the polynomial in the denominator must contain at least one factor of $(x-\lambda)$. 
If it contains multiple factors of $(x-\lambda)$, then the first term converges to a nonzero limit. 
Luckily we can derive an explicit criterion.

\begin{lemma}
\label{lem:zero}
Assume $\lambda$ is a zero of $\mathcal{D}(\beta,\cdot)$.
Set 
\[
\mathcal{H}(\alpha,\beta,x) := A(\alpha,\beta,\lambda) - A(\alpha,\beta,x) + 64\mathcal{D}(\beta,x)(1-x).
\]
Then $\lambda$ is a multiple zero of $\mathcal{H}(\alpha,\beta,\cdot)$ if and only if
\begin{align}
\label{eq:zero}
8(\lambda-1)\left(1-2\left(\lambda^2 - 2\lambda (3-\lambda) +\frac{45}{16}\right)\right) = \cos(2\pi\alpha).
\end{align}
\end{lemma}

\begin{remark}
Lemma \ref{lem:zero} is not vacuous. It is easy to see that \eqref{eq:zero} holds with $\alpha \in \{\frac{1}{4}, \frac{3}{4}\}$, $\beta\in \{\frac{1}{4},\frac{3}{4}\}$, and $\lambda=1$.
More generally, observe that the LHS of \eqref{eq:zero} depends only on $\beta$, whereas the RHS depends only on $\alpha$.
So long as the LHS has modulus $\leq 1$, there exists $\alpha$ for which \eqref{eq:zero} holds.
\end{remark}

\begin{proof}[Proof of Lemma \ref{lem:zero}]
Using \eqref{eq:A} for $A(\alpha,\beta,\cdot)$, as well as the factorization $\mathcal{D}(\beta,x) = -(x-\lambda)(x-a)(x-b)$, we get
\begin{align*}
\mathcal{H}(\alpha,\beta,x) &= 16(\lambda^2-x^2) - (32+4\cos(2\pi\alpha))(\lambda-x) +64(\lambda-x)(x-a)(x-b)(1-x)\\
&= 4(\lambda-x) \left[4(\lambda+x) - 8-\cos(2\pi\alpha) +16(x-a)(x-b)(1-x)\right].
\end{align*}
Thus $\lambda$ is a multiple zero of $\mathcal{H}(\alpha,\beta,\cdot)$ if and only if the expression in the square bracket vanishes when $x=\lambda$, \emph{i.e.,}
\[
8(\lambda-1)\left(1-2(\lambda-a)(\lambda-b)\right) = \cos(2\pi\alpha).
\]
We can then replace $a+b$ and $ab$ in terms of $\lambda$ and coefficients of the cubic polynomial $\mathcal{D}(\beta,\cdot)$ to obtain \eqref{eq:zero}.
\end{proof}

We summarize the above discussions in the following Table \ref{tab:TF}.

\begin{table}[h!]
\centering
\caption{Criterion table for Proposition \ref{prop:31}-\ref{G2}}
\label{tab:TF}
\begin{tabular}{ccc}
$\frac{d}{dx}|\Psi(x)|^2|_{x=\lambda}  =0$
&
\eqref{eq:zero} holds
&
$\displaystyle \lim_{\mathbb{R}\ni x\to\lambda}\frac{1}{|\Psi(x)|}\frac{\mathcal{D}(x)(\lambda-x)}{R(\lambda)-R(x)}$
\\ \hline
F & F & $0$\\
T & F & nonzero\\
F & T & nonzero\\
T & T & $0$ (if cancellation occurs) or nonzero
\end{tabular}
\end{table}

\subsubsection{$\Psi(\lambda)= 0$}

Recall that in \underline{Case IV}, $\Psi(\lambda)\neq 0$ for any $\lambda\in\mathbb{R}$, so Proposition \ref{prop:31} settles the spectral decimation problem in this case.
It remains to treat the exceptional values $\lambda$ in \underline{Cases II and III} where $\Psi(\lambda)=0$.
These are established in the next two propositions.
We begin with the (much) more straightfoward case.


\begin{proposition}
\label{prop:33}
In \underline{Case III}, if $\lambda=1+\frac{1}{2}\cos(2\pi \alpha_N)$, then ${\rm mult}(\mathcal{L}^\omega_N, \lambda)=0$.
\end{proposition}
\begin{proof}
In this case $\Psi(\lambda)=0$ and $\mathcal{D}(\lambda)\neq 0$, so $\phi(\lambda)=0$, $\lim_{\mathbb{R}\ni x\to\lambda}[R(x)]^{-1}=0$, and $(\lambda-x)[\phi(x)]^{-1}$ stays bounded as $x\to\lambda$.
Thus we are in the setting of Lemma \ref{multiplicity}-\ref{eq:SD7}. 
\end{proof}

Now we come to the subtler case.
Since $\Psi(x) = (x-\lambda)(x-a)$ for some $a\in \mathbb{C}$, $a\neq \lambda$, it follows that
\[
\lim_{x\uparrow \lambda} e^{i({\rm arg} \Psi(x))} = 
-\lim_{x\downarrow \lambda} e^{i({\rm arg} \Psi(x))},
\]
and this implies that the connection $\Omega(x)$ in the reduced Laplacian $\mathcal{L}^\Omega_{N-1}$ differs by an overall sign when $x\uparrow \lambda$ compared to when $x\downarrow \lambda$.
Precisely we have the identity
\begin{align}
\label{eq:limitagree}
\lim_{x\uparrow \lambda} {\rm sgn}(\lambda-x)(\mathcal{L}^\Omega_{N-1}-1) = \lim_{x\downarrow \lambda} {\rm sgn}(\lambda-x)(\mathcal{L}^\Omega_{N-1}-1).
\end{align}
With this in mind we can now state and prove the last remaining case.


\begin{proposition}
\label{prop:32}
In \underline{Case II}:
\begin{enumerate}[wide, label=(II.\arabic*)]
\item \label{II1} If \underline{either} $\alpha_N=0$, $\beta_N \notin \{0,\frac{1}{2}\}$, and $\lambda=\frac{3}{2}$, 
\underline{or} $\alpha_N=\frac{1}{2}$, $\beta_N\notin \{0,\frac{1}{2}\}$, and $\lambda=\frac{1}{2}$,
then
\begin{align*}
\label{eq:II1}
E_\lambda(\mathcal{L}^\omega_N) =  \left(P^*_\parallel - P^*_\perp (D-\lambda)^{-1} C\right)
\left(\lim_{\mathbb{R}\ni x\to \lambda}\frac{1}{\chi_0(x)}(\mathcal{L}^\Omega_{N-1}-R(x))^{-1} \right)
\left(P_\parallel - B(D-\lambda)^{-1} P_\perp\right),
\end{align*}
where $\displaystyle \chi_0(x) = \frac{4(\lambda-x) \mathcal{D}(x)}{|\Psi(x)|}$.
In particular,
$
{\rm mult}(\mathcal{L}^\omega_N, \lambda) = {\rm dim}_{N-1} = \frac{3^N+3}{2}
$.
\item \label{II2} If $\lambda$ is a double zero of $\mathcal{D}$---that is, \underline{either} $\beta_N=0$, $\alpha_N\notin \{0,\frac{1}{2}\}$, and $\lambda =\frac{5}{4}$, 
\underline{or} $\beta_N=\frac{1}{2}$, $\alpha_N\notin \{0,\frac{1}{2}\}$, and $\lambda=\frac{3}{4}$
---
then the following dichotomy holds.

On the one hand, if $\alpha_N \in \{\frac{1}{6},\frac{5}{6}\}$ in the case $\beta_N=0$ and $\lambda=\frac{5}{4}$, or $\alpha_N\in \{\frac{1}{3}, \frac{2}{3}\}$ in the case $\beta_N=\frac{1}{2}$ and $\lambda=\frac{3}{4}$, then
\begin{equation}
\label{eq:II2A}
\begin{aligned}
E&_\lambda(\mathcal{L}_N^\omega) = 
P^*_\perp E_\lambda(D) P_\perp \\
&+P_\perp^* E_\lambda(D) C 
\left(\lim_{\mathbb{R}\ni x\to \lambda} \frac{1}{\psi_0(x)}\left(\mathcal{L}^\Omega_{N-1}- R(x)\right)^{-1} \right)
\left(I_{\parallel} - E_{R(\lambda)}(\mathcal{L}^\Omega_{N-1})\right) 
BE_\lambda(D) P_\perp\\
&+ \left(P_{\parallel}^*-P_{\perp}^*(D-\lambda)^{-1}C\right) \left(\lim_{\mathbb{R} \ni x\to\lambda}\frac{1}{\phi(x)}\frac{\lambda-x}{R(\lambda)-R(x)}
 E_{R(x)}(\mathcal{L}^\Omega_{N-1})\right)
\left(P_{\parallel}-B(D-\lambda)^{-1}P_{\perp}\right),
\end{aligned}
\end{equation}
where $\displaystyle \psi_0(x)= \frac{(\lambda-x)|\Psi(x)|}{4\mathcal{D}(x)}$.
In general,
$$\mult(\mathcal{L}^\omega_N,\lambda) = 2\cdot 3^{N-1} - {\rm dim}_{N-1} + 2\mult(\mathcal{L}^\Omega_{N-1}, R(\lambda)) = \frac{3^{N-1}-3}{2} + 2\mult(\mathcal{L}^\Omega_{N-1}, R(\lambda)).
$$
That said, if none of the corresponding eigenfunctions vanishes on $V_{N-1}$, then the first two terms on the RHS of \eqref{eq:II2A} vanish, and $\mult(\mathcal{L}^\omega_N,\lambda) = \mult(\mathcal{L}^\Omega_{N-1}, R(\lambda))$.

On the other hand, for all other scenarios
\begin{equation*}
\label{eq:II2}
\begin{aligned}
E_\lambda&(\mathcal{L}_N^\omega) = 
P^*_\perp E_\lambda(D) P_\perp \\
&+ 
P_\perp^* E_\lambda(D) C 
\left(\lim_{\mathbb{R}\ni x\to \lambda} \frac{1}{\psi_0(x)}\left(\mathcal{L}^\Omega_{N-1}- R(x)\right)^{-1} \right)
\left(I_{\parallel} - E_{R(\lambda)}(\mathcal{L}^\Omega_{N-1})\right) 
BE_\lambda(D) P_\perp,
\end{aligned}
\end{equation*}
In particular, $E_\lambda(\mathcal{L}_N^\omega) (P^*_\perp E_\lambda(D) P_\perp) = E_\lambda(\mathcal{L}_N^\omega)$,
$$
{\rm mult}(\mathcal{L}^\omega_N, \lambda) = 
2 \cdot 3^{N-1} - {\rm dim}_{N-1} + {\rm mult}\left(\mathcal{L}^\Omega_{N-1}, R(\lambda)\right) = \frac{3^{N-1}-3}{2} +{\rm mult}\left(\mathcal{L}^\Omega_{N-1}, R(\lambda)\right),
$$
and the corresponding eigenfunction vanishes on $V_{N-1}$.
\end{enumerate}
\end{proposition}
\begin{proof}
\ref{II1}: We have $\Psi(\lambda)=0$, $\mathcal{D}(\lambda)\neq 0$, and thus $\phi(\lambda)=0$.
Nominally this would fall under the scenario of Lemma \ref{multiplicity}-\ref{eq:SD2}, but we need to address the connection sign change at $\lambda$.
First we carry out a tedious but elementary computation to get
\[
\begin{aligned}
R(0,\beta,x) = 1 + \frac{x-\frac{3}{2}}{|x-\frac{3}{2}|}\frac{-32 x^3 + 80x^2 - 58x + 11 - \cos(2\pi\beta)}{2 |4x-4-e^{-2\pi i\beta}|}&, \text{  if } \alpha=0,~ \beta\notin \left\{0,\frac{1}{2}\right\},~ \lambda=\frac{3}{2};\\
R\left(\frac{1}{2},\beta, x\right) = 1 + \frac{x-\frac{1}{2}}{|x-\frac{1}{2}|} \frac{-32x^3+112 x^2-122x + 41 - \cos(2\pi\beta)}{2|4x-4-e^{-2\pi i \beta}|}&, \text{  if } \alpha=\frac{1}{2},~\beta\notin\left\{0,\frac{1}{2}\right\},~\lambda=\frac{1}{2}.
\end{aligned}
\]
In either case we find
\begin{align}
\label{eq:R-1}
R(x)-1 = \frac{\lambda-x}{|\lambda-x|} \mathfrak{F}(x)
\end{align}
where $\mathfrak{F}$ is bounded in an $\mathbb{R}$-neighborhood of $\lambda$.
Consequently,
\begin{align*}
&\frac{\lambda-x}{\phi(x)}(\mathcal{L}^\Omega_{N-1} - R(x))^{-1} = 
4\mathcal{D}(x)\frac{\lambda-x}{|\Psi(x)|}\left((\mathcal{L}^\Omega_{N-1}-1) - (R(x)-1)\right)^{-1}\\
&= \frac{4\mathcal{D}(x)}{|x-a|}\frac{\lambda-x}{|\lambda-x|}\left((\mathcal{L}^\Omega_{N-1}-1) - \frac{\lambda-x}{|\lambda-x|} \mathfrak{F}(x)\right)^{-1}
= \frac{4\mathcal{D}(x)}{|x-a|}\left({\rm sgn}(\lambda-x)(\mathcal{L}^\Omega_{N-1}-1) - \mathfrak{F}(x)\right)^{-1},
\end{align*} 
which has a well-defined nonzero limit as $\mathbb{R}\ni x\to\lambda$ by \eqref{eq:limitagree}.

\ref{II2}: Since $\Psi(\lambda)=0$, and $\lambda$ is a double zero of $\mathcal{D}$, we have $\lim_{\mathbb{R}\ni x\to\lambda}[\phi(x)]^{-1} =0$.
This suggests that we are in the setting of either Lemma \ref{multiplicity}-\ref{eq:SD3} or Lemma \ref{multiplicity}-\ref{eq:SD6}, though again we need to account for the connection sign change at $\lambda$.
A tedious but elementary computation shows that
\[
\begin{aligned}
R(\alpha,0,x) = 1+ \frac{x-\frac{5}{4}}{|x-\frac{5}{4}|} \frac{(4x-3-\cos(2\pi\alpha)) + 2(2x-1)(4x-5)(1-x)}{|4x-3-e^{-4\pi i \alpha} - 2e^{-2\pi i \alpha}|},& \text{  if } \beta=0,~\alpha\notin\left\{0,\frac{1}{2}\right\},~\lambda=\frac{5}{4};\\
R\left(\alpha, \frac{1}{2},x\right) = 1 + \frac{x-\frac{3}{4}}{|x-\frac{3}{4}|}\frac{(4x-5-\cos(2\pi \alpha)) + 2(2x-3)(4x-3)(1-x)}{|4x-5+e^{-4\pi i \alpha} - 2 e^{-2\pi i \alpha}|}, &\text{  if } \beta=\frac{1}{2},~\alpha\notin \left\{0,\frac{1}{2}\right\},~\lambda=\frac{3}{4}.
\end{aligned}
\]
So once again $R(x)-1$ has the form \eqref{eq:R-1}.
Consequently, the limit  as $\mathbb{R}\ni x\to\lambda$ of
\[
\frac{1}{\phi(x)(\lambda-x)} (\mathcal{L}^\Omega_{N-1}-R(x))^{-1} = 
\frac{4}{|x-a|}\frac{\mathcal{D}(x)}{(\lambda-x)^2} \left({\rm sgn}(\lambda-x) (\mathcal{L}^\Omega_{N-1}-1) -\mathfrak{F}(x)\right)^{-1}
\]
exists and is nonzero on the image of $I_\parallel - E_{R(\lambda)}(\mathcal{L}^\Omega_{N-1})$.
The other quantity to analyze is
\begin{align*}
\frac{\lambda-x}{\phi(x)}&(\mathcal{L}^\Omega_{N-1}-R(x))^{-1}
= \frac{4\mathcal{D}(x)}{|x-a|}
({\rm sgn}(\lambda-x)(\mathcal{L}^\Omega_{N-1}-1) -\mathfrak{F}(x))^{-1}
\\
&=\frac{4(x-r_3)}{|x-a|} \frac{(\lambda-x)^2}{\mathfrak{F}(\lambda)-\mathfrak{F}(x)} (\mathfrak{F}(\lambda)-\mathfrak{F}(x))({\rm sgn}(\lambda-x)(\mathcal{L}^\Omega_{N-1}-1) -\mathfrak{F}(x))^{-1},
\end{align*}
where $r_3$ is the third zero of $\mathcal{D}$.
Since this term acts on the image of $E_{R(\lambda)}(\mathcal{L}^\Omega_{N-1})$, it remains to determine whether
\begin{align}
\label{eq:ratio2}
\lim_{\mathbb{R}\ni x\to\lambda}\frac{(\lambda-x)^2}{\mathfrak{F}(\lambda)-\mathfrak{F}(x)}
\end{align}
is nonzero, \emph{i.e.,} whether $\mathfrak{F}(\lambda)-\mathfrak{F}(\cdot)$ has a multiple zero at $\lambda$.
A direct computation shows that the limit \eqref{eq:ratio2} is nonzero iff: $\alpha \in \left\{\frac{1}{6}, \frac{5}{6}\right\}$ in the case $\beta=0$ and $\lambda=\frac{5}{4}$; or $\alpha \in \left\{\frac{1}{3}, \frac{2}{3}\right\}$ in the case $\beta=\frac{1}{2}$ and $\lambda=\frac{3}{4}$.
In these scenarios we are in the setting of Lemma \ref{multiplicity}-\ref{eq:SD6}; otherwise, Lemma \ref{multiplicity}-\ref{eq:SD3}.
\end{proof}

We now complete the proofs of the main results stated in \S\ref{sec:intro}.

\begin{proof}[Proof of Theorem \ref{thm:2}]
Combine Propositions \ref{prop:31}, \ref{prop:33}, and \ref{prop:32} to obtain \eqref{eq:spec2}.
\end{proof}

\begin{proof}[Proof of Proposition \ref{prop:zerospec}]
Proposition \ref{prop:31}-\ref{G2} and Proposition \ref{prop:32}-\ref{II2} implies Item \eqref{p1} and Item \eqref{p2}, respectively.
\end{proof}

\begin{proof}[Proof of Corollary \ref{cor:supspec}]
By assumption there is no $\lambda$ which is a double zero of $\mathcal{D}(\beta,\cdot)$.
Incorporating Proposition \ref{prop:zerospec} into Theorem \ref{thm:2} yields \eqref{eq:spec2ineq}. 
Now recall from Proposition \ref{exceptionalset2} that if $\alpha, \beta\notin \{0,\frac{1}{2}\}$, then $\Psi(\alpha,\beta,\cdot)$ has a $\mathbb{R}$-valued zero iff $3\alpha+\beta = \frac{1}{2} \pmod 1$. By Proposition \ref{prop:33} this zero is not in the spectrum. This allows us to deduce \eqref{eq:spec2ineqsimple} from \eqref{eq:spec2} and \eqref{eq:spec2ineq}.
\end{proof}

\begin{proof}[Proof of Theorem \ref{thm:specinf}]
Recall that for a normal operator $T$ on a Hilbert space $H$, the spectral radius of $T$ is equal to the operator norm of $T$.
This applies to the self-adjoint operator $\mathcal{L}^\omega_\infty$ on $L^2(V(G_\infty),\deg_{G_\infty})$.
Since $G_\infty$ is a bounded degree graph, and the $\omega_{xy}$ are unit complex numbers, it is direct to verify that the operator norm of $\mathcal{L}^\omega_\infty$ is bounded uniformly for all choices of $\omega$. 

We now prove that $S^\infty_1(\alpha,\beta) \subset \mathcal{K}(\mathcal{U}) \times ((\alpha,\beta)\times \mathbb{R})$.
Suppose $z\in \sigma(\mathcal{L}^{(\alpha,\beta)}_\infty) \setminus \mathcal{E}(\alpha,\beta)$.
By Theorem \ref{thm:2}, we have that $R(\alpha,\beta,z) \in \sigma(\mathcal{L}_\infty^{(\alpha_\downarrow(\alpha,\beta,z), \beta_\downarrow(\alpha,\beta,z))})$, \emph{i.e.,} $\mathcal{U}(\alpha,\beta,z) \in \mathbb{T}^2 \times \bigcup_{(\alpha',\beta')\in \mathbb{T}^2} \sigma(\mathcal{L}_\infty^{(\alpha',\beta')})$.
Iterating this forward, we deduce that $\bigcup_{k=0}^\infty \mathcal{U}^{\circ k}(\alpha,\beta,z) \subset \mathbb{T}^2 \times \bigcup_{(\alpha',\beta')\in \mathbb{T}^2} \sigma(\mathcal{L}_\infty^{(\alpha', \beta')})$.
Since the latter set is a bounded subset of $\mathbb{T}^2 \times \mathbb{C}$, we conclude that $\bigcup_{k=0}^\infty \mathcal{U}^{\circ k}(\alpha,\beta,z)$ is bounded, whence $(\alpha,\beta,z) \in \mathcal{K}(\mathcal{U})$.
\end{proof}


\section{Magnetic Laplacian determinants and asymptotic complexities}
\label{sec:det}

\subsection{Magnetic Laplacian determinants and cycle-rooted spanning forests} 
\label{sec:LapDet}

Using Theorem \ref{thm:1} we can compute the magnetic Laplacian determinant in the case $\alpha,\beta\in \{0,\frac{1}{2}\}$.
Recall that the determinant $\det$ is the product of all eigenvalues.
If $0$ is an eigenvalue, then we define ${\rm det}'$ to be the product of all nonzero eigenvalues.

The classic Kirchhoff's \textbf{matrix-tree theorem} states that on a finite graph $G$, the number of spanning trees $\tau(G)$ on $G$ equals $\frac{1}{|V(G)|}{\rm det}'(\Delta_G)$, or equivalently, the cofactor of $\Delta_G$ obtained by removing any one row and any one column (up to an overall sign).
Since we use the probabilistic Laplacian $\mathcal{L}_G = D_G^{-1}\Delta_G$, it is useful to know that
\begin{align}
\label{eq:matrixtree}
\tau(G) = \psi(G) {\rm det}'(\mathcal{L}_G), \quad\text{ where  } \psi(G) = \frac{\left(\prod_{v\in V(G)} \deg(v)\right)}{\left(\sum_{v\in V(G)} \deg(v)\right)}.
\end{align}
This follows from matching the coefficient of $t$ in the identity $\det(\mathcal{L}_G + tI) = (\det D_G)^{-1} \det(\Delta_G +t D_G)$ using the cofactor expansion, and the aforementioned matrix-tree theorem.

Enumeration of spanning trees on $SG$ has already been studied.
Set
\[
\psi(G_N):=\frac{\left(\prod_{v\in V_N} \deg(v)\right)}{\left(\sum_{v_\in V_N} \deg(v)\right)} =\frac{1}{2}\frac{2^{3^{N+1}}}{3^{N+1}}.
\]
It was shown in \cites{CCY, TeuflWagner} via a combinatorial approach, and in \cite{AnemaTsougkas} via spectral decimation and \eqref{eq:matrixtree}, that
\begin{align}
\label{det1} \tau(G_N) = \psi(G_N) {\rm det}'(\mathcal{L}^{(0,0)}_N)= 2^{\frac{3^N}{2}-\frac{1}{2}} \cdot
3^{\frac{3^{N+1}}{4}+\frac{N}{2}+\frac{1}{4}} \cdot
5^{\frac{3^N}{4}-\frac{N}{2}-\frac{1}{4}}.
\end{align}
By placing a uniform probability measure on the set of all spanning trees, \emph{a.k.a.\@} considering \textbf{uniform spanning trees (USTs)}, we obtain a determinantal point process on the edge set with kernel $K=d{\sf G}d^*$, where ${\sf G}$ is the Green's function for random walks.
The matrix $K$ is known as the \emph{transfer impedance matrix} \cite{BurtonPemantle}.
For more properties of USTs on $SG$ and scaling limit questions, see \cite{STW14}.

Our next theorem gives the determinant formulae for the three magnetic Laplacians.
The normalization prefactor $\psi(G_N)$ is used for the same reason as described in \eqref{eq:matrixtree} above.

\begin{theorem}
\label{thm:3}
We have
\begin{equation}
\begin{aligned}
\det&(\mathcal{L}^{(\frac{1}{2},\frac{1}{2})}_N)= \frac{1}{\psi(G_N)}
\cdot 2^{\frac{3^N}{2}+\frac{3}{2}}
\cdot 3^{\frac{3^{N-1}}{2}-N-\frac{3}{2}}
\cdot 5^{\frac{3^{N-1}}{2}+\frac{3}{2}}\\
&\times
\left[
\prod_{k=0}^{N-2} \left(H(k)+\frac{1}{2}\right)^{\frac{3^{N-k-2}+3}{2}}
\right]
\left[
\prod_{k=0}^{N-3} \left(H(k)+\frac{5}{2}\right)^{\frac{3^{N-k-2}-1}{2}}
\right],\\
\end{aligned}
\end{equation}
where $H(0)=26.5$, and for $k\geq 1$, $H(k)=[H(k-1)]^2 - \frac{15}{4}$;
\begin{equation}
\begin{aligned}
\det&(\mathcal{L}^{(\frac{1}{2},0)}_N)= \frac{1}{\psi(G_N)}\cdot 2^{\frac{13}{6} 3^{N-1}-\frac{5}{2}}\cdot 3^{\frac{3^{N-2}}{2}-N-\frac{3}{2}} \cdot 5^{\frac{5}{2} 3^{N-2}-1}\cdot 7^{\frac{3^{N-1}}{2}+\frac{3}{2}}\cdot 17^{\frac{3^{N-2}}{2} + \frac{3}{2}} \\
&\times
\left[
\prod_{k=0}^{N-3} \left(\tilde{H}(k) +\frac{1}{2}\right)^{\frac{3^{N-k-3}+3}{2}}
\right]
\left[
\prod_{k=0}^{N-4} \left(\tilde{H}(k) +\frac{5}{2}\right)^{\frac{3^{N-k-3}-1}{2}}
\right],
\end{aligned}
\end{equation}
where $\tilde{H}(0)=302.5$, and for $k\geq 1$, $\tilde{H}(k)=[\tilde{H}(k-1)]^2 - \frac{15}{4}$;
and
\begin{equation}
\begin{aligned}
\det&(\mathcal{L}^{(0,\frac{1}{2})}_N)= \frac{1}{\psi(G_N)}\cdot
2^{\frac{13}{6}3^{N-1}-\frac{5}{2}} \cdot
3^{\frac{7}{3}3^{N-1} - N + 3}\cdot
7^{\frac{3^{N-2}}{2} - \frac{1}{2}}
\\
&\times
\left[
\prod_{k=0}^{N-3} \left(\hat{H}(k) +\frac{1}{2}\right)^{\frac{3^{N-k-3}+3}{2}}
\right]
\left[
\prod_{k=0}^{N-4} \left(\hat{H}(k) +\frac{5}{2}\right)^{\frac{3^{N-k-3}-1}{2}}
\right],
\end{aligned}
\end{equation}
where $\hat{H}(0)=86.5$, and for $k\geq 1$, $\hat{H}(k)=[\hat{H}(k-1)]^2 - \frac{15}{4}$.
\end{theorem}

The proof of Theorem \ref{thm:3} is postponed till \S\ref{sec:pfthm3}.


There is an analog of the matrix-tree theorem for the magnetic Laplacian determinant, established by Forman \cite{Forman} and Kenyon \cite{kenyon}. 
To explain this, we recall some definitions from \cites{kenyon, KK17}, and refer the reader to these papers for more details.
A cycle-rooted tree, or \emph{unicycle}, is a tree plus an extra edge to form a single cycle.
A \textbf{cycle-rooted spanning forest (CRSF)} is a spanning forest whose connected components are unicycles.
See Figure \ref{fig:CRSF} for an illustration.

\begin{figure}
\centering
\includegraphics[width=0.3\textwidth]{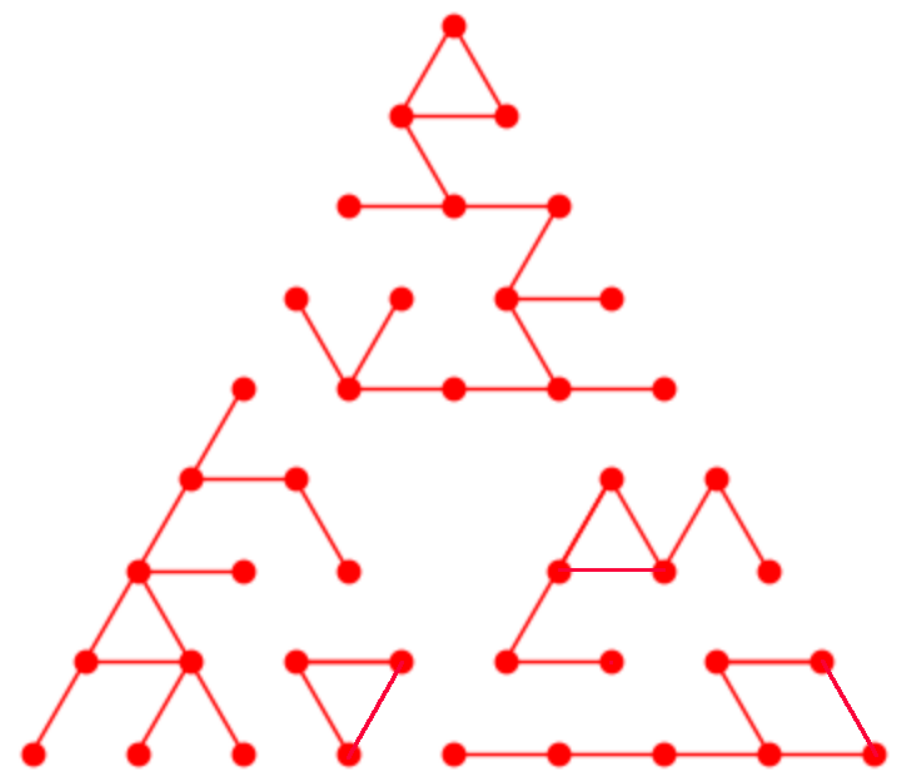}
\caption{An instance of a cycle-rooted spanning forest on the level-3 gasket graph, generated via the sampling algorithm of Kassel and Kenyon \cite{KK17}*{p.\@ 938} based on loop-erased random walks.
Image courtesy of Quan Vu.}
\label{fig:CRSF}
\end{figure}

Fix a connected graph $G$, a directed edge conductance function ${\bf c}$ on $\{\pm E\}$, and a line bundle connection $\omega$ on $G$.
We would like to assign a probability measure on the set of all CRSFs thereon.
Declare that each oriented CRSF (OCRSF) occurs with probability proportional to $\prod_{e\in \text{bushes}} {\bf c}(e) \prod_{\gamma\in{\rm cycles}} {\bf C}(\gamma)(1- \omega(\gamma))$, 
where the first product runs over all edges in the bushes (that is, not in the cycles),
the second product runs over all cycles,
${\bf C}(\gamma)$ is the product of the semiconductances\footnote{The semiconductance of an undirected edge with end vertices $x$ and $y$ is defined as $\frac{1}{2}({\bf c}_{xy}+{\bf c}_{yx})$.} along $\gamma$,
and $\omega(\gamma)$ is the holonomy of $\gamma$. 
The following proposition says that $\det(\mathcal{L}^\omega_{(G,{\bf c})})$ gives the partition function which makes the said CRSF measure a probability measure.

\begin{proposition}[Matrix-CRSF theorem \cite{kenyon}*{Theorem 6}]
\label{prop:MCRSFT}
Let $\mathcal{L}^\omega_{(G,{\bf c})}$ be the line bundle Laplacian \eqref{eq:LBLap}.
Then 
\begin{align}
\label{eq:MCRSFT}
\det\left(\mathcal{L}^\omega_{(G,{\bf c})}\right)= \sum_{{\rm OCRSFs}} \prod_{e\in {\rm bushes}} {\bf c}(e) \prod_{\gamma \in {\rm cycles}} {\bf C}(\gamma)\left(1-\omega(\gamma)\right).
\end{align}
\end{proposition}

\begin{remark}
Note that if ${\bf c}_{xy}= {\bf c}_{yx}$ for all $xy\in E$, then \eqref{eq:MCRSFT} may be written as a sum over \emph{unoriented} CRSFs:
\[
\det\left(\mathcal{L}^\omega_{(G,{\bf c})}\right)= \sum_{{\rm CRSFs}} \prod_{e} {\bf c}(e) \prod_{\gamma \in {\rm cycles}} \left(2-\omega(\gamma)-\frac{1}{\omega(\gamma)}\right),
\]
where the first product is over all edges in the CRSF \cite{kenyon}*{Theorem 5}.
\end{remark}

Like the UST process, the CRSF process is also a determinantal point process on the edge set, with kernel $d\left(\mathcal{L}^\omega_{(G,{\bf c})}\right)^{-1} d^*$ \cite{kenyon}*{Theorem 2}. 
In \cite{KK17}*{p.\@ 938} Kassel and Kenyon gave an elementary sampling algorithm for CRSFs based on loop-erased random walks (LERWs) and cycle popping \`a la Wilson \cite{Wilson96}, valid for any holonomy $e^{2\pi i\gamma}$ with $\gamma \in [-\frac{1}{4},\frac{1}{4}]$.
Using this algorithm, in combination with facts from LERWs and Brownian loop soups \cite{LTF07}, they proved convergence to a loop measure on an oriented Riemannian surface from the CRSF processes on discretizations of the said surface \cite{KK17}*{Theorem 20}.


\subsection{Asymptotic complexity and a large deviations result}
\label{sec:entropy}

It is an open problem to study local properties and scaling limits of the CRSF measures on $SG$.
That said, we can use our results on the magnetic Laplacian determinant (Theorem \ref{thm:3}) to quantify the \textbf{asymptotic complexity} of the CRSF measures.

Let $G_\infty$ be an infinite connected graph which can be exhausted by a sequence of finite connected graphs $\{G_N\}_N$.
We define the asymptotic complexity associated with $\mathcal{L}^\omega_\infty$ on $G_\infty$ by
\begin{align}
\label{eq:ac}
\mathfrak{h}(G_\infty, \mathcal{L}^\omega_\infty) := \lim_{N\to\infty} \frac{\log \left(w(G_N) {\rm det} ' (\mathcal{L}^\omega_N)\right)}{|V_N|}
\end{align}
provided that the RHS limit exists.

The formula \eqref{eq:ac} is classical for the graph Laplacian, \emph{i.e.,} for the enumeration of spanning trees.
In \cites{Lyons, Lyons10} R.\@ Lyons introduced the notion of \textbf{tree entropy} on $G_\infty$, gave several equivalent formulations---one of which is the logarithm of a Fuglede-Kadison determinant \cite{FK} of the ``continuum'' Laplacian---and proved that his tree entropy equals \eqref{eq:ac}.\footnote{If $G_\infty$ has bounded degree, the proof in \cite{Lyons}*{Theorem 4.1} suffices. If $G_\infty$ has unbounded degree, then the proof proceeds according to \cite{Lyons10}*{Theorem 3.1}, which is based on von Neumann algebras.}\footnote{The limit of USTs on an infinite connected graph is a spanning forest. On $\mathbb{Z}^d$ the limit is a tree iff $d\leq 4$ \cites{Pemantle91}.}
For old and new results on tree entropy for various graphs, see \cites{AnemaTsougkas, Lyons, Lyons10, CTT}.
As an example, from \eqref{det1} it is direct to show that the tree entropy on $SG$ equals (\emph{cf.\@} \cite{AnemaTsougkas}*{Corollary 5.2})\footnote{
In \eqref{eq:TESG} the weights associated to the logarithmic factors are probability weights.
This is merely coincidental: for the graphical $(d-1)$-dimensional Sierpinski simplex, the tree entropy equals $\frac{d-2}{d}\log 2+\frac{d-2}{d-1}\log d + \frac{d-2}{d(d-1)}\log (d+2)$ \cite{CTT}*{Corollary 4.1}. 

More generally, the tree entropy of a unimodular random infinite connected weighted graph can take values in $[-\infty, \infty)$. 
For an example of a unimodular random graph with tree entropy equal to $-\infty$, see \cite{Lyons10}*{pp.\@ 308-309}.
}
\begin{align}
\label{eq:TESG}
\mathfrak{h}(SG, \mathcal{L}^{(0,0)}_\infty) = \frac{\log 2}{3} + \frac{\log 3}{2} + \frac{\log 5}{6} = 1.04859\dots.
\end{align}

Our next result gives the asymptotic complexity of each of the three magnetic Laplacians on $SG$.

\begin{corollary}[Asymptotic complexity of the CRSF measures]
\label{cor:ent}
\begin{align*}
\mathfrak{h}&(SG, \mathcal{L}^{(\frac{1}{2},\frac{1}{2})}_\infty) = \frac{\log 2}{3} + \frac{\log 3}{9} + \frac{\log 5}{9} \\
&+ \frac{2}{9} \cdot \frac{1}{3}\sum_{k=0}^\infty \left(\frac{2}{3}\right)^k \frac{\log(H(k)+\frac{1}{2})}{2^{k+1}} 
+ \frac{2}{9}  \cdot \frac{1}{3}\sum_{k=0}^\infty \left(\frac{2}{3}\right)^k \frac{\log(H(k)+\frac{5}{2})}{2^{k+1}},\\
\mathfrak{h}&(SG, \mathcal{L}^{(\frac{1}{2},0)}_\infty) = \frac{13}{27}\log 2 + \frac{\log 3}{27} + \frac{5}{27}\log 5 +\frac{\log 7}{9} + \frac{\log 17}{27} \\
&+ \frac{2}{27} \cdot \frac{1}{3}\sum_{k=0}^\infty \left(\frac{2}{3}\right)^k \frac{\log(\tilde{H}(k)+\frac{1}{2})}{2^{k+1}} 
+ \frac{2}{27}  \cdot \frac{1}{3}\sum_{k=0}^\infty \left(\frac{2}{3}\right)^k \frac{\log(\tilde{H}(k)+\frac{5}{2})}{2^{k+1}},\\
\mathfrak{h}&(SG, \mathcal{L}^{(0,\frac{1}{2})}_\infty) =
\frac{13}{27}\log 2+ \frac{14}{27} \log 3+ \frac{\log 7}{27} \\
   &+ \frac{2}{27} \cdot \frac{1}{3}\sum_{k=0}^\infty \left(\frac{2}{3}\right)^k \frac{\log(\hat{H}(k)+\frac{1}{2})}{2^{k+1}} 
+ \frac{2}{27}  \cdot \frac{1}{3}\sum_{k=0}^\infty \left(\frac{2}{3}\right)^k \frac{\log(\hat{H}(k)+\frac{5}{2})}{2^{k+1}},
\end{align*}
where $H(k)$, $\tilde{H}(k)$, and $\hat{H}(k)$ were defined in Theorem \ref{thm:3}.
\end{corollary}
These follow from Theorem \ref{thm:3} via a computation that is briefly described in \S\ref{sec:pfthm3}. 
If we replace each infinite sum by the corresponding partial sum up to $k=40$, the following lower bounds are obtained: 
\begin{align*}
\mathfrak{h}(SG, \mathcal{L}^{(\frac{1}{2},\frac{1}{2})}_\infty) \geq 1.26388, \quad
\mathfrak{h}(SG, \mathcal{L}^{(\frac{1}{2},0)}_\infty) \geq 1.41685, \quad
\mathfrak{h}(SG, \mathcal{L}^{(0,\frac{1}{2})}_\infty) \geq 1.30625.
\end{align*}
Thus each of the three asymptotic complexities is larger than the tree entropy \eqref{eq:TESG}.

The CRSF asymptotic complexity has a probabilistic interpretation which we explain now.
Let $G$ be a finite connected graph, and $B$ be a subset of $V(G)$ which we declare as the boundary set.
An \textbf{essential CRSF} on $(G,B)$ is a spanning subgraph of $G$, each of whose connected components is either a unicycle not containing any vertex in $B$, or a tree containing a unique vertex in $B$.
The corresponding matrix-CRSF theorem is the analog of Proposition \ref{prop:MCRSFT}, where on the LHS the Laplacian carries Dirichlet boundary condition on $B$,
\begin{align}
(\mathcal{L}^\omega_{(G,B,{\bf c})} u)(x) =  \sum_{y\sim x} {\bf c}_{xy} (u(x) - \omega_{xy} u(y)), \quad u\in \mathbb{C}^{V\setminus B},
\end{align}
the sum being over the neighbors $y$ of $x$ ($y$ can be in $B$); and on the RHS the sum runs over all oriented essential CRSFs.
Let us denote the essential CRSF measure on $(G, B, {\bf c}, \omega)$ by $\mathbb{P}^\omega_{(G,B, {\bf c})}$.
Since the conductance will not play a role in the remainder of this discussion, we will suppress the subscript ${\bf c}$ in what follows.

Let $\mathcal{L}^{\rm Id}_{(G,B)}$ be the magnetic Laplacian on $(G,B)$ with the trivial connection.
It is easy to see from the definition of the CRSF measure that
\begin{align}
\mathbb{P}^\omega_{(G,B)}[\text{no loops}] = \frac{\det \left(\mathcal{L}^{\rm Id}_{(G,B)}\right)}{\det\left(\mathcal{L}^\omega_{(G,B)}\right)}.
\end{align}
We simplify further to get
\begin{align}
\label{eq:logP}
\frac{1}{|V(G)|}\log\mathbb{P}^\omega_{(G,B)}[\text{no loops}] = \frac{\log\left(w(G) \det\left(\mathcal{L}^{\rm Id}_{(G,B)}\right) \right)}{|V(G)|} -\frac{\log\left(w(G) \det\left(\mathcal{L}^\omega_{(G,B)}\right) \right)}{|V(G)|}.
\end{align}
\emph{cf.\@} \cite[Proposition 3.6 and Corollary 3.7]{KasselLevy}.

Suppose we have an increasing sequence of graphs with boundary $\{(G_N,B_N)\}_N$ tending to $G_\infty$ with $\frac{|B_N|}{|V(G_N)|} \to 0$.
Furthermore, suppose that the essential tree entropy and the essential CRSF asymptotic complexity exist. 
Then \eqref{eq:logP} says that the difference of the two asymptotic complexities gives the rate of exponential decay in the probability of observing no loops under $\mathbb{P}^\omega_{(G_N, B_N)}$ as $N\to\infty$; that is, we may define
\begin{align}
\mathfrak{h}_{\rm loop}(G_\infty, \mathcal{L}_\infty^\omega) := \mathfrak{h}(G_\infty, \mathcal{L}^\omega_\infty) - \mathfrak{h}(G_\infty, \mathcal{L}^{\rm Id}_\infty).
\end{align}
Then from \eqref{eq:logP} we obtain
\begin{align}
\label{eq:loopentropy}
\lim_{N\to\infty} \frac{1}{|V(G_N)|}\log\mathbb{P}^\omega_{(G_N, B_N)}[\text{no loops}] = -\mathfrak{h}_{\rm loop}(G_\infty, \mathcal{L}^\omega_\infty).
\end{align}

Let us apply this principle to $SG$.
Recall that in defining the magnetic Laplacians on $SG$, we did not impose any boundary condition.
That said, we can add a single point $b$ and connect it to the origin by an edge of conductance $c$, and regard $G_N \cup \{b\}$ as the graph with boundary $B=\{b\}$.
This one-point modification introduces correction terms on the RHS of \eqref{eq:loopentropy} that vanish as $c\downarrow 0$ for every $N$.
As a result, we can apply \eqref{eq:TESG} and Corollary \ref{cor:ent} to obtain the following asymptotic result.

\begin{corollary}
\label{cor:loopsoupentropy}
Let $\mathbb{P}^{(\alpha,\beta)}_{N,c}$ be the essential CRSF measure on $G_N \cup \{b\}$ where an edge of conductance $c$ connects $o$ and $b$, and with fluxes $\alpha$ and $\beta$ as before.
For $(\alpha,\beta) \in \left\{(0,\frac{1}{2}), (\frac{1}{2},0), (\frac{1}{2},\frac{1}{2})\right\}$,
\begin{align}
\lim_{N\to\infty} \lim_{c\downarrow 0} \frac{1}{|V_N|} \log \mathbb{P}^{(\alpha, \beta)}_{N,c}\left[{\rm no~loops}\right] = -\mathfrak{h}_{\rm loop}(SG, \mathcal{L}^{(\alpha,\beta)}_\infty),
\end{align}
where $\mathfrak{h}_{\rm loop}(SG, \mathcal{L}^{(\alpha,\beta)}_\infty) =\mathfrak{h}(SG,\mathcal{L}^{(\alpha,\beta)}_\infty)- \mathfrak{h}(SG, \mathcal{L}^{(0,0)}_\infty)$, and each term on the RHS was defined in Corollary \ref{cor:ent} and \eqref{eq:TESG}, respectively.
\end{corollary}

\subsection{Proofs of Theorem \ref{thm:3} and Corollary \ref{cor:ent}}
\label{sec:pfthm3}

First let us recall some simple identities. 
Let $P(x) = a_d x^d+ a_{d-1} x^{d-1}+ \cdots + a_0$ be a polynomial of degree $d$. Then
\begin{align*}
\{z: z\in P^{-1}(\alpha)\} &=  \left\{z: \alpha =a_d z^d + a_{d-1}z^{d-1}+\cdots + a_0 \right\}\\
&=\left\{z: a_d z^d + a_{d-1} z^{d-1} +\cdots + (a_0-\alpha)=0\right\}.
\end{align*}
It follows that
\begin{equation}
\label{eq:simple}
\sum_{z\in P^{-1}(\alpha)} z = - \frac{a_{d-1}}{a_d} \quad \text{and} \quad
\prod_{z \in P^{-1}(\alpha)} z  = (-1)^d \frac{a_0-\alpha}{a_d}.
\end{equation}
Assume $R(x)=b_2 x^2+ b_1 x$ is a quadratic polynomial function with the property that $R(0)=0$.
Using \eqref{eq:simple} and induction on $n$, it is easy to deduce that
\begin{align}
\label{eq:prodzR}
\prod_{z\in R^{-n}(\alpha)} z = -\alpha \frac{b_2}{(b_2)^{2^n}}.
\end{align}
See \cite{AnemaTsougkas}*{Lemma 3.3} for a more general version where $R$ is a rational function satisfying $R(0)=0$.
Let us note that the computation of ${\rm det}'(\mathcal{L}^{(0,0)}_N)$ uses \eqref{eq:prodzR} to treat the preimages $R(0,0,\cdot)^{-k}(\alpha)$, $\alpha\in \{\frac{3}{4}, \frac{5}{4}\}$.
For details we refer the reader to \cite{AnemaTsougkas}*{Theorem 5.1}.


For the computation of $\det(\mathcal{L}_N^{(\alpha,\beta)})$ where $(\alpha,\beta)\in \{(\frac{1}{2},0), (0,\frac{1}{2}), (\frac{1}{2}, \frac{1}{2})\}$, we need to involve two additional quadratic polynomials $P$ and $Q$, without the requirement that $P(0)=0$ or $Q(0)=0$. 

\begin{lemma}
\label{lem:recur}
Let $P(x)=a_2 x^2 + a_1 x + a_0$ and $R(x) = b_2 x^2+ b_1 x$. Then
\begin{align*}
F(n, \alpha):=\prod_{z\in P^{-1}(R^{-n}(\alpha))} z 
= c_{n,1}\alpha + c_{n,0},
\end{align*}
where $\displaystyle c_{n,1}= -\frac{b_2}{(a_2 b_2)^{2^n}}$, $\displaystyle c_{n,0}=\frac{1}{(a_2 b_2)^{2^n}} \left(H(n)-\frac{b_1}{2}\right)$, and $H(n)$ satisfies the recurrence relation
\begin{align}
\label{eq:Grec}
H(0)= a_0 b_2 + \frac{b_1}{2}, \quad \text{and for } n\geq 1,~
H(n) = [H(n-1)]^2 + \frac{b_1(2-b_1)}{4}.
\end{align}
\end{lemma}
\begin{proof}
We prove this by induction on $n$.
When $n=0$, we have by \eqref{eq:simple} that $F(0,\alpha)=\frac{a_0-\alpha}{a_2}$.
Now suppose $F(n,\alpha)= c_{n,1} \alpha + c_{n,0}$ holds.
Then denoting the two preimages of $R$ by $R^{-1}_{(1)}$ and $R^{-1}_{(2)}$, we have
\begin{align*}
F&(n+1,\alpha) = 
\prod\left\{P^{-1}(R^{-(n+1)}(w)): w=\alpha\right\}
=
\prod\left\{P^{-1}(R^{-n}(w)): w\in R^{-1}(\alpha)\right\}\\
&= F(n, R^{-1}_{(1)}(\alpha)) F(n, R^{-1}_{(2)}(\alpha))\\
&= c_{n,1}^2 \left(R^{-1}_{(1)}(\alpha) R^{-1}_{(2)}(\alpha)\right)
+c_{n,0}c_{n,1} \left(R^{-1}_{(1)}(\alpha) + R^{-1}_{(2)}(\alpha)\right)
+ c_{n,0}^2 \quad (\text{induction hypothesis})\\
&= c_{n,1}^2 \left(\frac{-\alpha}{b_2}\right) + c_{n,0} c_{n,1} \left(-\frac{b_1}{b_2}\right) + c_{n,0}^2 \quad (\text{by }\eqref{eq:simple}) \\
&= \left(-\frac{c_{n,1}^2}{b_2}\right)\alpha + \left[c_{n,0}^2 - \frac{b_1}{b_2} c_{n,0} c_{n,1} \right].
\end{align*}
We have thus obtained a system of quadratic recurrence relations
\[
c_{n+1, 1} = -\frac{1}{b_2} c_{n,1}^2, \quad c_{n+1,0} = c_{n,0}^2 - \frac{b_1}{b_2} c_{n,0} c_{n,1}
\]
with initial condition $c_{0,1}=-\frac{1}{a_2}$ and $c_{0,0}=\frac{a_0}{a_2}$.
It readily follows that $c_{n,1} = -b_2 /(a_2 b_2)^{2^n}$ and
\[
c_{n+1,0} = c_{n,0}^2 + \frac{b_1}{(a_2 b_2)^{2^n}} c_{n,0} .
\]
To better see the latter relation, we perform a change of variables $g(n) = c_{n,0}+ \frac{1}{2}\frac{b_1}{(a_2 b_2)^{2^n}}$ to obtain
\[
g(n+1) = [g(n)]^2 + \frac{1}{4}\frac{b_1(2-b_1)}{(a_2 b_2)^{2^{n+1}}}.
\]
Making another change of variables to $H(n) = g(n) (a_2 b_2)^{2^n}$, we deduce
\eqref{eq:Grec}.
\end{proof}

In general the quadratic recurrence \eqref{eq:Grec} cannot be solved in closed form, unless the constant term on the RHS is either $0$ or $-2$ \cite{QuadRec}, or under special initial conditions. 
To give an example of a special initial condition: if $a_0=0$, then $G(n)=\frac{b_1}{2}$---and hence $c_{n,0}=0$---for all $n$.

\begin{lemma}
\label{lem:recur2}
Let $P$ and $R$ be as in Lemma \ref{lem:recur}, and $Q(x)= q_2 x^2+ q_1 x + q_0$. Then
\begin{align*}
\tilde{F}(n,\alpha) :=\prod_{z\in Q^{-1}\circ  P^{-1}(R^{-n}(\alpha))} z = \tilde{c}_{n,1}\alpha+ \tilde{c}_{n,0},
\end{align*}
where $\displaystyle \tilde{c}_{n,1}=-\frac{b_2}{(q_2^2 a_2 b_2)^{2^n}} $, $\displaystyle \tilde{c}_{n,0}= \frac{1}{(q_2^2 a_2 b_2)^{2^n}} \left(\tilde{H}(n)-\frac{b_1}{2}\right)$, and $\tilde{H}(n)$ satisfies the recurrence relation
\begin{align}
\tilde{H}(0)=a_2 b_2\left(q_0^2+ q_0\frac{a_1}{a_2} + \frac{a_0}{a_2}\right) + \frac{b_1}{2}, \quad \text{and for } n\geq 1,~
\tilde{H}(n) = [\tilde{H}(n-1)]^2 + \frac{b_1(2-b_1)}{4}.
\end{align}
\end{lemma}
\begin{proof}
We focus on the initial step $n=0$:
\begin{align*}
\tilde{F}(0,\alpha) &= 
\left(
Q^{-1}_{(1)}(P^{-1}_{(1)}(\alpha))
Q^{-1}_{(2)}(P^{-1}_{(1)}(\alpha))
\right)
\left(
Q^{-1}_{(1)}(P^{-1}_{(2)}(\alpha))
Q^{-1}_{(2)}(P^{-1}_{(2)}(\alpha))
\right)\\
&= \left(\frac{q_0- P^{-1}_{(1)}(\alpha)}{q_2}\right)
 \left(\frac{q_0- P^{-1}_{(2)}(\alpha)}{q_2}\right) \quad (\text{by \eqref{eq:simple}})\\
 &= \frac{1}{q_2^2} \left(q_0^2 - q_0 (P^{-1}_{(1)}(\alpha) + P^{-1}_{(2)}(\alpha)) + P^{-1}_{(1)}(\alpha) P^{-1}_{(2)}(\alpha)\right) \\
 &=\left(- \frac{1}{q_2^2 a_2}\right) \alpha + \frac{1}{q_2^2}\left(q_0^2+q_0\frac{a_1}{a_2} + \frac{a_0}{a_2}\right) \quad (\text{by \eqref{eq:simple}}).
\end{align*}
The induction step is the same as in the proof of Lemma \ref{lem:recur}, except for the changes triggered by the initial conditions. The details are therefore omitted.
\end{proof}

\begin{proof}[Proof of Theorem \ref{thm:3}]
Recall Theorem \ref{thm:1}.
For $\det(\mathcal{L}^{(\frac{1}{2}, \frac{1}{2})}_N)$, we use Lemma \ref{lem:recur} to treat the preimages $R(\frac{1}{2},\frac{1}{2},\cdots)^{-1}\circ R(0,0,\cdot)^{-k}(\alpha)$, $\alpha\in \{\frac{3}{4},\frac{5}{4}\}$. Taking into account multiplicities we obtain
\begin{align*}
\det&(\mathcal{L}^{(\frac{1}{2},\frac{1}{2})}_N)=
\left(\frac{1}{2}\right)^{\frac{3^N+3}{2}} \left(\frac{3}{4}\right)^{\frac{3^{N-1}-1}{2}} \left(\frac{5}{4}\right)^{\frac{3^{N-1}+3}{2}} \cdot 2 \\ 
&\times 
\left[
\prod_{k=0}^{N-2} \left(\left(4\cdot \frac{3}{4} + H(k)-\frac{5}{2}\right)\frac{1}{16^{2^k}}\right)^{\frac{3^{N-k-2}+3}{2}}
\right]
\left[
\prod_{k=0}^{N-3} \left(\left(4\cdot \frac{5}{4} + H(k)-\frac{5}{2}\right)\frac{1}{16^{2^k}}\right)^{\frac{3^{N-k-2}-1}{2}}
\right],
\end{align*}
where $H(0)=26.5$, and for $k\geq 1$, $H(k)=[H(k-1)]^2 - \frac{15}{4}$.

For the other two determinants, we apply Lemma \ref{lem:recur2} to treat the preimages and obtain
\begin{align*}
\det&(\mathcal{L}^{(0,\frac{1}{2})}_N)=
\left(\frac{1}{2}\right)^{\frac{3^N+3}{2}} \left(\frac{5}{4}\right)^{\frac{3^{N-1}-1}{2}} \left(\frac{7}{4}\right)^{\frac{3^{N-1}+3}{2}}\cdot \left(\frac{3+\frac{3}{4}}{4}\right)^{\frac{3^{N-2}-1}{2}} \left(\frac{3+\frac{5}{4}}{4}\right)^{\frac{3^{N-2}+3}{2}} \\
&\times
\left[
\prod_{k=0}^{N-3}
\left(
\left(
4\cdot\frac{3}{4} + \tilde{H}(k) - \frac{5}{2}
\right)
\frac{1}{64^{2^k}}
\right)
^{\frac{3^{N-k-3}+3}{2}}
\right]
\left[
\prod_{k=0}^{N-4}
\left(
\left(
4\cdot \frac{5}{4} +\tilde{H}(k) - \frac{5}{2}
\right)
\frac{1}{64^{2^k}}
\right)
^{\frac{3^{N-k-3}-1}{2}}
\right],
\end{align*}
where $\tilde{H}(0)=302.5$, and for $k\geq 1$, $\tilde{H}(k) = [\tilde{H}(k-1)]^2 - \frac{15}{4}$.

\begin{align*}
\det&(\mathcal{L}^{(\frac{1}{2},0)}_N)=
\left(\frac{1}{4}\right)^{\frac{3^{N-1}+3}{2}} \left(\frac{3}{4}\right)^{\frac{3^{N-1}-1}{2}} \left(\frac{3}{2}\right)^{\frac{3^N+3}{2}}\cdot \left(\frac{1+\frac{3}{4}}{4}\right)^{\frac{3^{N-2}-1}{2}} \left(\frac{1+\frac{5}{4}}{4}\right)^{\frac{3^{N-2}+3}{2}} \\
&\times 
\left[
\prod_{k=0}^{N-3}
\left(
\left(
4\cdot\frac{3}{4} + \hat{H}(k) - \frac{5}{2}
\right)
\frac{1}{64^{2^k}}
\right)
^{\frac{3^{N-k-3}+3}{2}}
\right]
\left[
\prod_{k=0}^{N-4}
\left(
\left(
4\cdot \frac{5}{4} +\hat{H}(k) - \frac{5}{2}
\right)
\frac{1}{64^{2^k}}
\right)
^{\frac{3^{N-k-3}-1}{2}}
\right],
\end{align*}
where $\hat{H}(0)=86.5$, and for $k\geq 1$, $\hat{H}(k) = [\hat{H}(k-1)]^2 - \frac{15}{4}$.
\end{proof}

\begin{proof}[Proof of Corollary \ref{cor:ent}]
We present the proof of $\mathfrak{h}(SG, \mathcal{L}^{(\frac{1}{2},\frac{1}{2})}_\infty)$ only, the other two being similar.
In particular, we only demonstrate the contribution to the asymptotic complexity from one of the inhomogeneous products, say,
\begin{equation}
\begin{aligned}
\label{eq:asymp}
\frac{1}{|V_N|}\log\left[\prod_{k=0}^{N-2} \left(4\cdot \frac{3}{4}+H(k)-\frac{5}{2}\right)^{\frac{3^{N-k-2}+3}{2}}\right]
&=\frac{2}{3^{N+1}+3}\sum_{k=0}^{N-2} \frac{3^{N-k-2}+3}{2} \log\left(H(k)+\frac{1}{2}\right) \\
&=\frac{2}{9}\cdot \frac{1}{3}\sum_{k=0}^{N-2} \left(\frac{3^{-k}}{2}+ \Theta(3^{-N})\right) \log\left(H(k)+\frac{1}{2}\right).
\end{aligned}
\end{equation}
Let us observe that if there exists a fixed constant $C$ such that $H(k) = [H(k-1)]^2 +C$ for all $k\geq 1$, then
\[
\log H(k) = 2\log H(k-1) +\log \left(1+\frac{C}{[H(k-1)]^2}\right),
\]
or
\[
\frac{\log H(k)}{2^k} - \frac{\log H(k-1)}{2^{k-1}} = \frac{1}{2^k} \log\left(1+\frac{C}{|H(k-1)|^2}\right) \leq \frac{1}{2^k} \frac{C}{|H(k-1)|^2}.
\]
where we used the inequality $1+x\leq e^x$.
Since $\sum_k \frac{1}{2^k |H(k-1)|^2}$ is summable, it follows that $\lim_{k\to\infty} 2^{-k} \log H(k)$ exists.
By the same rationale, $\xi_k := 2^{-k} \log(H(k)+\frac{1}{2})$ converges as $k\to\infty$.
Thus we can rewrite the RHS of \eqref{eq:asymp} as
\[
\frac{2}{9}\cdot \frac{1}{3} \sum_{k=0}^{N-2} \left(\frac{3^{-k} 2^k}{2} + \Theta(3^{-N})2^k\right) \xi_k = \frac{2}{9}\cdot \frac{1}{3} \sum_{k=0}^{N-2}\left( \frac{2}{3}\right)^k \frac{\xi_k}{2} + \Theta\left(\left(\frac{2}{3}\right)^N\right),
\]
with the intention of collecting terms of order unity.
(Above we used the identity $\frac{1}{3}\sum_{k=0}^\infty \left(\frac{2}{3}\right)^k =1$. 
Also, roughly speaking, the reason for the factor $2$ in $\xi_k/2$ is to account for the double preimages under $R(\frac{1}{2},\frac{1}{2},\cdot)$.)
\end{proof}

\section{Open questions}

We end this paper with several open questions.


\subsubsection*{Spectral questions}

In this paper we focused almost exclusively on the eigenvalues of $\mathcal{L}_N^\omega$, without going into details the structure of the eigenfunctions.
Looking back at the proofs, we saw that the eigenfunctions associated to some, but not all, exceptional values in $\sigma(\mathcal{L}_N^\omega)$ vanish on $V_{N-1}$.
A more careful analysis will show that these eigenfunctions are \emph{localized} with finite support.
The open question is whether the magnetic Laplacian eigenfunctions with finite support are complete, as was the case for the graph Laplacian \cite{Teplyaev}.
(See also \cite{Quint} which considered Dirichlet boundary condition at the origin.)
It seems that the answer should be affirmative at least in the case $\alpha,\beta \in \{0,\frac{1}{2}\}$.

A related problem is to obtain the magnetic spectrum on the compact gasket $K=\overline{\bigcup_{N=0}^\infty \mathfrak{G}_N}$, rather than on the infinite lattice $G_\infty$.
While in the latter case we used that $\lambda\in \sigma(\mathcal{L}_N^{(\alpha,\beta)})$ iff $R(\alpha,\beta,\lambda) \in \sigma(\mathcal{L}_{N-1}^{(\alpha_\downarrow, \beta_\downarrow)})$ for the non-exceptional values $\lambda$, now we find, in the former case, a family of scaling parameters $\{\Lambda(\alpha,\beta)>0: (\alpha,\beta)\in \mathbb{T}^2\}$ such that for each $(\alpha,\beta)\in \mathbb{T}^2$, the limit $\lim_{N\to\infty} [\Lambda(\alpha,\beta)]^N \mathcal{U}^{\circ -N} (\alpha,\beta,w)$ exists for $w$ in a finite subset of $\mathbb{R}$.
The difficulty lies in our lack of explicit understanding of the backward dynamics of $\mathcal{U}$.


\subsubsection*{Point processes induced by the magnetic Laplacian determinants}

To the best of the authors' knowledge, this is the first time CRSFs have been considered in a setting outside of discrete approximations of Euclidean spaces or Riemannian manifolds.
Recall that key results on the properties of the CRSF measure were obtained by Kassel and Kenyon \cite{KK17}, in the context of compact Riemannian surfaces approximated conformally by graphs.
More recently, Finski \cite{finski} established convergence of the CRSF measures for a rank-2 vector bundle on a class of flat surfaces which can be tiled by squares of equal sizes.

Thanks to the explicit characterization of $\sigma(\mathcal{L}_N^{(\alpha,\beta)})$ in the case $\alpha,\beta\in \{0,\frac{1}{2}\}$, we were able to compute the probability of observing no cycles under the CRSF measure on $SG$.
However this is only one aspect of the CRSF measure.
For instance, the following question remains open:

\begin{openquestion}
Characterize the limit point(s) of the sequence of CRSF measures on finite approximations of $SG$.
\end{openquestion}

This question may be solved if one can find the \emph{zeta-regularized determinant}  (also known as the \emph{analytic torsion}) on the compact limit space $K=\overline{\bigcup_{N=0}^\infty \mathfrak{G}_N}$ (as opposed to $G_\infty$), and derive a \emph{renormalized logarithm} of this determinant.
For vector bundle Laplacians we mention the following results:
rank-1 or rank-2 bundles on Euclidean or Riemannian surfaces, see \cite{KK17}*{Theorem 17} and \cite{finski}*{Theorem 1.8}; 
rank-1 bundles on the $d$-dimensional tori, see \cite{Friedli}.
While the problem is open for rank-1 bundles on $SG$, we mention that in \cite{CTT}, the zeta-regularized determinant for the graph Laplacian, and its renormalized logarithm, were computed on $SG$ and several other self-similar fractals.
There it was crucial to derive a functional equation involving the spectral decimation function $R$, and invoke analytic properties of spectral zeta functions.
We are optimistic that this can be done for the magnetic Laplacians in the case $\alpha,\beta\in \{0,\frac{1}{2}\}$.

Local statistics of loops under the CRSF measure on $SG$ is also an open question. 
Based on numerical simulations, we observe a hierarchy in the distribution of unicycle lengths (peaks centered around $3^j$, $j\in \mathbb{N}$), reflecting the spatial self-similarity of $SG$.
 
\begin{openquestion}
Characterize the spectrum of the twisted Laplacian on $SG$ endowed with a rank-2 Hermitian vector bundle with a $SL_2(\mathbb{C})$ connection.
\end{openquestion}

Finally we believe that studying the magnetic Laplacian and the induced CRSF loop measures may shed light on properties of the abelian sandpile model under stationarity.
Recall the well-known bijective correspondence between the sandpile group and spanning trees \cite{MajumdarDhar}.
Kassel and Kenyon \cite{KK17}*{\S6, Question 9} have asked if the loop measures may lead to a better understanding of waves of sandpile avalanches.
On $SG$ we have two specific questions: to prove that the sandpile avalanches exhibit a power law modulated by log-periodic oscillations, which was numerically observed in \cites{Stanley96, DV98, DPV01};\footnote{
A power law modulated by log-periodic oscillations was proved for the growth of deterministic single-source abelian sandpile on $SG$ \cite{CKF19}.}
and to find the sandpile height distributions (or their moments, such as the sandpile density).\footnote{See \cite{MatterThesis}*{Chapter 5} for the proof of sandpile height distributions on the Hanoi tower graphs, a variant of $SG$. A nice exposition of the connection between sandpile density and the looping rate on periodic planar graphs is \cite{KW16}.}

\subsubsection*{Experimental realization of the butterfly}

Last but not least, thanks to advances in scanning electron microscopy and nanoscale engineering over the past 3 decades, there has been impressive progress on measurements of electronic band structures in various (meta)materials, including finite approximations of $SG$.
The most recent work we are aware of is \cite{Kempkes}.
It would be satisfying to see the Hofstadter-Sierpinski butterfly (Figure \ref{fig:newspectrum}) come alive through a laboratory experiment.

\lstset{language=Matlab,%
    basicstyle=\ttfamily\small,%
    breaklines=true,%
    morekeywords={matlab2tikz},
    keywordstyle=\color{blue},%
    morekeywords=[2]{1}, keywordstyle=[2]{\color{black}},
    identifierstyle=\color{black},%
    stringstyle=\color{mylilas},
    commentstyle=\color{mygreen},%
    showstringspaces=false,
    numbers=left,%
    numberstyle={\tiny \color{black}},
    numbersep=9pt, 
    emph=[1]{for,end,break},emphstyle=[1]\color{red}, 
}

\appendix
\section{Numerical approximation of the filled Julia set in Figure \ref{fig:newspectrum}} \label{app}

To numerically generate the filled Julia set of the map $\mathcal{U}$, we initialize with a uniform sample of points $w=(\alpha,\lambda)$ in the rectangle $[0,1]\times [0,2]$.
We then discard points $w$ for which $|\mathcal{U}^k(w)|$ exceeds a threshold ($10$) after $k(=20)$ iterations, and keep those points which remain bounded within.
Below is a working MATLAB code.

\begin{lstlisting}[
	basicstyle=\ttfamily\scriptsize,
]
clear;clc

dpts=301;
lb=0;
ub=2;
th=10;
num_iter=20;

x=@(a,b,l) cos(2*pi*a);
xs=@(a,b,l) sin(2*pi*a);
y=@(a,b,l) cos(2*pi*b);
ys=@(a,b,l) sin(2*pi*b);
cosaplusb=@(a,b,l) x(a,b,l)*y(a,b,l)-xs(a,b,l)*ys(a,b,l);
cosa2plusb=@(a,b,l) (x(a,b,l)^2-xs(a,b,l)^2)*y(a,b,l)-2*xs(a,b,l)*x(a,b,l)*ys(a,b,l);
sinaplusb=@(a,b,l) xs(a,b,l)*y(a,b,l)+x(a,b,l)*ys(a,b,l);
sina2plusb=@(a,b,l) 2*xs(a,b,l)*x(a,b,l)*y(a,b,l)+ys(a,b,l)*(x(a,b,l)^2-xs(a,b,l)^2);
A=@(a,b,l) 16*l^2-(32+4*x(a,b,l))*l+15+4*x(a,b,l)+cosaplusb(a,b,l);
D=@(a,b,l) -l^3+3*l^2-45/16*l+13/16-y(a,b,l)/32;
re_psi=@(a,b,l) (1-l)^2-1/16+(1-l)/4*(2*x(a,b,l)+cosa2plusb(a,b,l))...
    +1/16*(x(a,b,l)^2-xs(a,b,l)^2+2*cosaplusb(a,b,l));
im_psi=@(a,b,l) -(1-l)/4*(2*xs(a,b,l)+sina2plusb(a,b,l))-(1/16)*(2*x(a,b,l)*xs(a,b,l)+2*sinaplusb(a,b,l));
R=@(a,b,l) 1+(A(a,b,l)-64*D(a,b,l)*(1-l))/(16*sqrt(re_psi(a,b,l)^2+im_psi(a,b,l)^2));

aset=linspace(0,1,dpts);
lset=linspace(lb,ub,dpts);
figure
hold on

for i=1:dpts
    for j=1:dpts
        al=aset(i);
        be=aset(i);
        lmd=lset(j);
        count=0;
        while abs(lmd)<th
            count=count+1;
            psi=re_psi(al,be,lmd)+1i*im_psi(al,be,lmd);
            theta=angle(psi);
            if count==num_iter
                plot(aset(i),lset(j),'.','color','k')
                break
            end
            lmd=R(al,be,lmd);
            al_dummy=al;
            be_dummy=be;
            al=mod(3*al_dummy+be_dummy+3*theta/2/pi,1);
            be=mod(3*be_dummy+al_dummy-3*theta/2/pi,1);
        end
    end
end
xlabel('\alpha')
ylabel('\lambda')
title('Filled Julia set of U')
\end{lstlisting}

\bibliographystyle{alpha}
\bibliography{bibliography}

\end{document}